\documentclass[a4paper,openany]{amsbook}
\pdfoutput=1
\usepackage{libertine}
\usepackage[libertine]{newtxmath}
\usepackage{akenthmslight}
\usepackage{aken}
	\todono

\usepackage{dirhott}
\usepackage{mathpartir}

\usepackage[english]{babel}
\usepackage{a4wide}
\usepackage{graphicx}
\usepackage{wrapfig}

\usepackage{chngcntr}
\counterwithin{section}{chapter}

\newcounter{equationbis}
\newcommand{\tagbis}{\stepcounter{equationbis}\tag{\arabic{equation}\alph{equationbis}}}

\newcommand{\inference}[3]{\inferenceright{#1}{#2}{\text{#3}}}
\newcommand{\binference}[3]{\binferenceright{#1}{#2}{\text{#3}}}
\renewcommand{\loch}{{\llcorner\hspace{-.4ex}\lrcorner}}

\newcommand{\Ty}{\mathrm{Ty}}
\newcommand{\Tm}{\mathrm{Tm}}

\newcommand{\Disc}{\name{Disc}}
\newcommand{\uniPsh}{\cat U^{\name{Psh}}}
\newcommand{\uniNDD}{\cat U^{\name{NDD}}}
\newcommand{\uniDD}{\cat U^{\name{DD}}}
\newcommand{\PropD}{\Prop^{\name{D}}}

\newcommand{\IB}{\mathbb{B}}
\newcommand{\IP}{\mathbb{P}}
\newcommand{\IE}{\mathbb{E}}
\newcommand{\XX}{\mathbb{X}}
	\newcommand{\ctxbrid}[1]{{#1} : \IB}
	\newcommand{\ctxpath}[1]{{#1} : \IP}
	\newcommand{\ctxedge}[1]{{#1} : \IE}
	\newcommand{\ctxline}[1]{{#1} : \XX}

\newcommand{\IX}{\mathbb{I}}
\newcommand{\IF}{\mathbb{F}}

\newcommand{\cwfpair}[2]{{#1, #2}}

\newcommand{\ap}{\name{ap}}
\newcommand{\J}{\name{J}}
\newcommand{\tycode}[1]{\ulcorner #1 \urcorner}
\newcommand{\dtycode}[1]{\underline \ulcorner #1 \underline \urcorner}
\newcommand{\dEl}{\underline \El}

\newcommand{\tycodeDD}[1]{\ulcorner #1 \urcorner^{\name{DD}}}
\newcommand{\ElDD}{\El^{\name{DD}}}
\newcommand{\tycodeD}[1]{\ulcorner #1 \urcorner^{\name{D}}}
\newcommand{\ElD}{\El^{\name{D}}}

\newcommand{\dtype}{\,\name{dtype}}
\newcommand{\prop}{\,\name{prop}}
\newcommand{\eqrel}{\,\name{eqrel}\,}

\newcommand{\yoneda}{\mathbf{y}}

\newcommand{\homclass}[1]{\mathscr{#1}}

\newcommand{\sheq}{\name{SE}}
\newcommand{\fpsh}[1]{{#1}^\dagger}%
\newcommand{\fpshadj}[1]{\underline{\alpha}_{#1}}

\newcommand{\lpsh}[1]{\widehat{#1}}

\newcommand{\ftrtm}[2]{{}^{#1} #2}
\newcommand{\ftrvar}[2]{\ftrtm{#1}{\var{#2}}}

\newcommand{\shp}{{\rbag}}%
\newcommand{\quotshp}{\mathrlap{\shp}{\circ}}
	\newcommand{\inquotshp}{\varsigma_\circ}
	\newcommand{\hatinquotshp}{\inquotshp}
\newcommand{\tmshp}[1]{\hatinquotshp(#1)}

\newcommand{\coshp}{\P}

\newcommand{\catTop}{\name{Top}}
\newcommand{\cohpi}{{\sqcap}}
\newcommand{\cohdisc}{{\vartriangle}}
\newcommand{\cohfget}{{\sqcup}}
\newcommand{\cohcodisc}{{\triangledown}}
\newcommand{\cohpaths}{{\boxminus}}

\renewcommand{\isvar}{\sharp}
\renewcommand{\invar}{\circ}

\newcommand{\sys}[1]{\paren{#1}}
\newcommand{\sysclauseb}[2]{#1\,?\,#2}
\newcommand{\sysclause}[2]{\parclr{#1}\,?\,#2}

\newcommand{\Glue}{\name{Glue}}
\newcommand{\Gluesys}[2]{\Glue \accol{#1 \leftarrow \sys{#2}}}
\newcommand{\Gluesysclause}[3]{#1\,?\,{#2, \ptwclr{#3}}}
\newcommand{\Gluesysclauseb}[3]{#1\,?\,{#2, #3}}
\newcommand{\glue}{\name{glue}}
\newcommand{\gluesys}[2]{\glue \accol{#1 \mapsfrom \sys{#2}}}
\newcommand{\unglue}{\name{unglue}}
\newcommand{\ungluesys}[1]{\unglue\,\sys{#1}}
\newcommand{\Gluetp}[4]{\Gluesys{#2}{\Gluesysclause{#1}{#3}{#4}}}
\newcommand{\gluetm}[3]{\gluesys{#2}{\sysclause{#1}{#3}}}
\newcommand{\ungluetm}[3]{\ungluesys{\sysclause{#1}{\ptwclr{#2}}}\,#3}

\newcommand{\Weld}{\name{Weld}}
\newcommand{\Weldsys}[2]{\Weld \accol{#1 \to \sys{#2}}}
\newcommand{\Weldsysclause}[3]{#1\,?\,{#2, \ptwclr{#3}}}
\newcommand{\Weldsysclauseb}[3]{#1\,?\,{#2, #3}}
\newcommand{\weld}{\name{weld}}
\newcommand{\weldsys}[1]{\weld\,\sys{#1}}

\newcommand{\dweld}{\underline{\weld}}
\newcommand{\dunweld}{\underline{\name{unweld}}}
\newcommand{\Weldtp}[4]{\Weldsys{#2}{\Weldsysclause{#1}{#3}{#4}}}
\newcommand{\weldtm}[3]{\weldsys{\sysclause{#1}{\ptwclr{#2}}}\,#3}

\newcommand{\parclr}[1]{\textcolor{purple}{#1}}
\newcommand{\ptwclr}[1]{\textcolor{cyan}{#1}}
\newcommand{\varclr}[1]{\textcolor{orange}{#1}}
\newcommand{\leftflat}[1]{\sharp \setminus {#1}}
\newcommand{\leftsharp}[1]{\coshp \setminus {#1}}
\newcommand{\judty}[1]{\sez \parclr{ #1 } \type}

\newcommand{\judtm}[2]{\sez #1 : \parclr{ #2 }}
\newcommand{\judtmeq}[3]{\sez #1 \jeq #2 : \parclr{ #3 }}

\newcommand{\idmod}{\mathbf{id}}
\newcommand{\ctxctu}[2]{#1 : \parclr{ #2 }}
\newcommand{\ctxpar}[2]{\parclr{ #1 }^\sharp : \parclr{ #2 }}
\newcommand{\ctxptw}[2]{\ptwclr{ #1 }^\coshp : \parclr{ #2 }}
\renewcommand{\ctxvar}[3]{\varclr{ #2 }^{#1} : \parclr{ #3 }}
\newcommand{\ctxface}[1]{\parclr{#1}}
\newcommand{\Pipar}{\forall}
\newcommand{\Piptw}{\Pi^\coshp}
\newcommand{\Pivar}[1]{\Pi^{#1}}
\newcommand{\Sigmapar}{\exists}
\newcommand{\Sigmaptw}{\Sigma^\coshp}
\newcommand{\Sigmavar}[1]{\Sigma^{#1}}
\newcommand{\prodpar}[2]{\Pipar\paren{#1 : #2}}
\newcommand{\prodctu}[2]{\Pi\paren{#1 : #2}}
\newcommand{\prodptw}[2]{\Piptw\paren{\ptwclr{#1} : #2}}
\renewcommand{\prodvar}[3]{\Pivar{#1}\paren{\varclr{#2} : #3}}
\newcommand{\sumpar}[2]{\Sigmapar\paren{#1 : #2}}
\newcommand{\sumctu}[2]{\Sigma\paren{#1 : #2}}
\newcommand{\sumptw}[2]{\Sigmaptw\paren{\ptwclr{#1} : #2}}
\renewcommand{\sumvar}[3]{\Sigmavar{#1}\paren{\varclr{#2} : #3}}

\newcommand{\lamannotpar}[2]{\lambda (\parclr{#1}^\sharp : \parclr{#2})}

\newcommand{\lamannotvar}[3]{\lambda (\varclr{#2}^{#1} : \parclr{#3})}

\newcommand{\appar}[2]{#1\,\parclr{#2}^\sharp}
\newcommand{\apvar}[3]{#2\,\varclr{#3}^{#1}}

\renewcommand{\pairvar}[3]{(\varclr{#2}^{#1}, #3)}
\newcommand{\pairpar}[2]{(\parclr{#1}^\sharp, #2)}

\newcommand{\fstptw}[1]{\name{fst}^\coshp\,\parclr{#1}}
\newcommand{\sndptw}[1]{\name{snd}^\coshp\,#1}

\newcommand{\idpr}[2]{#1 \doteq #2}

\newcommand{\degax}{\name{degax}}
\newcommand{\degaxof}[1]{\degax\,\parclr{#1}}

\newcommand{\Nat}{\name{Nat}}
\renewcommand{\Size}{\name{Size}}
	\newcommand{\szero}{0_{\name{S}}}
	\newcommand{\ssuc}{{\uparrow}}
	
	\newcommand{\smax}[2]{#1 \sqcup #2}
	\newcommand{\sfix}{\name{fix}}
	\newcommand{\sfill}{\name{fill}}
	\newcommand{\sfillsys}[1]{\name{fill}\sys{#1}}
	\newcommand{\sfillsyssharp}[1]{\name{fill}_\sharp\sys{#1}}
	\newcommand{\sfillsysclause}[2]{#1\,?\,#2}
\newcommand{\leqfill}{\name{fill}_\leq}
\newcommand{\leqfillsys}[1]{\leqfill\sys{#1}}
\newcommand{\leqfillsysclause}[2]{\sysclause{#1}{#2}}
	
\renewcommand{\suc}{\name{s}}

\newcommand{\forsub}[1]{#1 !}

\newcommand{\fix}{\name{fix}}

\newcommand{\cubecat}{\name{Cube}}
	
\newcommand{\bpcubecat}{\name{BPCube}}
	\newcommand{\facewkn}[1]{#1 / \oslash}
	
	\newcommand{\bpdisc}{\widehat{\bpcubecat}_\Disc}
		\newcommand{\DTy}{\Ty^\Disc}
\newcommand{\pointcat}{{\name{Point}}}
\newcommand{\catW}{{\cat W}}
\newcommand{\catV}{{\cat V}}
\newcommand{\RGcat}{{\name{RG}}}
	\newcommand{\RGdisc}{\widehat{\RGcat}_\Disc}

\newcommand{\PSub}[2]{#1 \Rrightarrow #2}
\newcommand{\DSub}[2]{#1 \Rightarrow #2}
\newcommand{\Dsez}{\mathrel{\rhd}}
\newcommand{\sub}[1]{\brac{#1}}
\newcommand{\dsub}[1]{{\left[ #1 \right\rangle}}
\newcommand{\psub}[1]{\angles{#1}}
\newcommand{\ssub}[1]{\left\{ #1 \right \}}
\newcommand{\dlambda}{\underline \lambda}
\newcommand{\dap}{\underline \ap}

\newcommand{\emptysub}{\bullet}

\newcommand{\textdef}[1]{\textbf{#1}}

\newcommand{\var}[1]{{\mathbf{#1}}}
\newcommand{\wknvar}[1]{\wkn{\var{#1}}}
\newcommand{\wkn}[1]{\pi^{#1}}
\newcommand{\subext}{{+}}
\newcommand{\fresh}{\mathbf{fr}}

\newcommand{\fst}{\name{fst}}
\newcommand{\snd}{\name{snd}}

\newcommand{\textand}{\,\text{and}\,}

\newcommand{\thetitle}{A Model of Parametric Dependent Type Theory in Bridge/Path Cubical Sets}
\newcommand{\theauthor}{Andreas Nuyts}
\newcommand{\theinstitution}{KU Leuven}
\newcommand{\thedepartment}{imec-DistriNet \\ Dept. of Computer Science}

\begin{document}
	\addtolength{\voffset}{-.5in}

\title{\thetitle}
\date{\today}
\author{\theauthor}

\begin{titlepage}
	\centering
	{\scshape\LARGE \theinstitution \par}
	{\scshape\Large \thedepartment\par}
	\vspace{1cm}
	{\scshape\Large Technical Report\par}
	\vspace{1.5cm}
	{\huge\bfseries \thetitle \par}
	\vspace{2cm}
	{\Large\itshape \theauthor \par}
	\vfill

	\vfill

	{\large \today\par}
\end{titlepage}

\setcounter{tocdepth}{2}
\tableofcontents

\chapter*{Introduction}
The purpose of this text is to prove all technical aspects of our model for dependent type theory with parametric quantifiers \cite{paramdtt}. It is well-known that any presheaf category constitutes a model of dependent type theory \cite{Hofmann97-presheaf-chapter}, including a hierarchy of universes if the metatheory has one \cite{psh-universes}. We construct our model by defining the base category $\bpcubecat$ of \textdef{bridge/path cubes} and adapting the general presheaf model over $\bpcubecat$ to suit our needs. Our model is heavily based on the models by Atkey, Ghani and Johann \cite{dtt-parametricity}, Huber \cite{huber}, Bezem, Coquand and Huber \cite{model-cubical}, Cohen, Coquand, Huber and M\"ortberg \cite{cubical}, Moulin \cite{moulin} and Bernardy, Coquand and Moulin \cite{moulin-param3}.

In \cref{ch:psh}, we review the main concepts of categories with families, and the standard presheaf model of dependent type theory, and we establish the notations we will use.

In \cref{ch:cwf-morphisms}, we capture morphisms of CwFs, and natural transformations and adjunctions between them, in typing rules. We especially study morphisms of CwFs between presheaf categories, that arise from functors between the base categories.

In \cref{ch:bpcubecat}, we introduce the category $\bpcubecat$ of bridge/path cubes, and its presheaf category $\widehat{\bpcubecat}$ of bridge/path cubical sets. There is a rich interaction with the category of cubical sets $\widehat{\cubecat}$ which we investigate more closely using ideas from axiomatic cohesion \cite{adjoint-logic}.

In \cref{ch:discreteness}, we define discrete types and show that they form a model of dependent type theory. We prove some infrastructural results.

In \cref{ch:paramdtt}, we give an interpretation of the typing rules of ParamDTT \cite{paramdtt} in $\widehat{\bpcubecat}$.

\section*{Acknowledgements}
Special thanks goes to Andrea Vezzosi. A cornerstone of this model was Andrea's insight that a shape modality on reflexive graphs is relevant to modelling parametricity. The other foundational ideas -- in particular the use of (cohesive-like) endofunctors of a category with families and the internalization of them as modalities -- were formed in discussion with him. He also injected some vital input during the formal elaboration process and pointed out the relevance of the $\Glue$-operator from cubical type theory \cite{cubical}.

Also thanks to Andreas Abel, Paolo Capriotti, Jesper Cockx, Dominique Devriese, Dan Licata and Sandro Stucki for many fruitful discussions.

The author holds a Ph.D. Fellowship from the Research Foundation - Flanders (FWO).

\section*{Changes in version 2}
\begin{enumerate}
	\item An erratum regarding the semantics of the reflection rule was rectified to the extent possible, see \cref{sec:idtp-semantics}.
	\item Some citations were added to the introduction above.
	\item An unimportant error was fixed in a remark in \cref{def:morphism-of-cwfs}.
	\item A notation convention that was never used, was removed from \cref{sec:cubecat}.
	\item A few typos and typesetting errors were fixed.
\end{enumerate}

\chapter{The standard presheaf model of Martin-L\"of Type Theory}\label{ch:psh}
In this chapter, we introduce the notion of a category with families (CwF, \cite{dybjer-cwf}) and show that every presheaf category constitutes a CwF that supports various interesting type formers. Most of this has been shown by \cite{Hofmann97-presheaf-chapter,psh-universes}; the construction of $\Glue$-types has been shown by \cite{cubical}. The construction of the $\Weld$-type is new. 

\section{Categories with families}
We state the definition of a CwF without referring to the category $\name{Fam}$ of families. Instead, we will make use of the category of elements:
\begin{definition}
	Let $\cat C$ be a category and $A : \cat C \to \Set$ a functor. Then the category of elements $\int_{\cat C} A$ is the category whose
	\begin{itemize}
		\item objects are pairs $(c, a)$ where $c$ is an object of $\cat C$ and $a \in A(c)$,
		\item morphisms are pairs $(\vfi | a) : (c, a) \to (c', a')$ where $\vfi : c \to c'$ and $a' = A(\vfi)(a)$.
	\end{itemize}
	If the functor $A : \cat C\op \to \Set$ is contravariant, we define $\int_{\cat C} A := \paren{\int_{\cat C\op} A}\op$. Thus, its morphisms are pairs $(\vfi|a') : (c,a) \to (c', a')$ where $\vfi : c \to c'$ and $a = A(\vfi)(a')$.
\end{definition}
\begin{definition}\label{def:cwf}
	A \textbf{category with families} (CwF) \cite{dybjer-cwf} consists of:
	\begin{enumerate}
		\item A category $\Ctx$ whose objects we call \textbf{contexts}, and whose morphisms we call \textbf{substitutions}. We also write $\Gamma \ctx$ to say that $\Gamma$ is a context.
		\item A contravariant functor $\Ty : \Ctx\op \to \Set$. The elements $T \in \Ty(\Gamma)$ are called \textbf{types} over $\Gamma$ (also denoted $\Gamma \sez T \type$). The action $\Ty(\sigma) : \Ty(\Gamma) \to \Ty(\Delta)$ of a substitution $\sigma : \Delta \to \Gamma$ is denoted $\loch[\sigma]$, i.e. if $\Gamma \sez T \type$ then $\Delta \sez T[\sigma] \type$.
		\item A contravariant functor $\Tm : \paren{\int_{\Ctx} \Ty}\op \to \Set$ from the category of elements of $\Ty$ to $\Set$. The elements $t \in \Tm(\Gamma, T)$ are called \textbf{terms} of $T$ (also denoted $\Gamma \sez t : T$). The action $\Tm(\sigma | T) : \Tm(\Gamma, T) \to \Tm(\Delta, T[\sigma])$ of $(\sigma | T) : (\Delta, T[\sigma]) \to (\Gamma, T)$ is denoted $\loch[\sigma]$, i.e. if $\Gamma \sez t : T$, then $\Delta \sez t[\sigma] : T[\sigma]$.
		\item A terminal object $()$ of $\Ctx$ called the \textbf{empty context}.
		\item A \textbf{context extension} operation: if $\Gamma \ctx$ and $\Gamma \sez T \type$, then there is a context $\Gamma.T$, a substitution $\pi : \Gamma.T \to \Gamma$ and a term $\Gamma.T \sez \xi : T[\pi]$, such that for all $\Delta$, the map
		\begin{equation*}
			\Hom(\Delta,\Gamma.T) \to \Sigma(\sigma : \Hom(\Delta, \Gamma)). \Tm(\Delta, T[\sigma]) : \tau \mapsto (\pi  \tau , \xi[\tau])
		\end{equation*}
		is invertible. We call the inverse $\loch, \loch$.
		Note that for more precision and less readability, we could write $\pi_{\Gamma, T}$, $\xi_{\Gamma, T}$ and $(\loch, \loch)_{\Gamma, T}$.
		
		If $\sigma : \Delta \to \Gamma$, then we will write $\sigma \subext = (\sigma \pi, \xi) : \Delta.T[\sigma] \to \Gamma.T$.
	
		Sometimes, for clarity, we will use variable names: we write $\Gamma, \var x : T$ instead of $\Gamma.T$, and $\wknvar x : (\Gamma, \var x : T) \to \Gamma$ and $\Gamma, \var x : T \sez \var x : T[\wknvar x]$ for $\pi$ and $\xi$. Their joint inverse will be called $(\loch, \loch/\var x)$.
	\end{enumerate}
\end{definition}

\section{Presheaf categories are CwFs}
This is proven elaborately in \cite{Hofmann97-presheaf-chapter}, though we give an unconventional treatment that views the Yoneda-embedding truly as an embedding, i.e. treating the base category $\catW$ as a fully faithful subcategory of $\widehat{\catW}$.

\subsection{Contexts} Pick a base category $\catW$. We call its objects \textdef{primitive contexts} and its morphisms \textdef{primitive substitutions}, denoted $\vfi : \PSub V W$. The presheaf category $\widehat \catW$ over $\catW$ is defined as the functor space $\Set^{\catW\op}$. We will use $\widehat \catW$ as $\Ctx$. A context $\Gamma$ is thus a \textdef{presheaf} over $\catW$, i.e. a functor $\Gamma : \catW\op \to \Set$. We denote its action on a primitive context $W$ as $\DSub W \Gamma$, and the elements of that set are called \textdef{defining substitutions} from $W$ to $\Gamma$. The action of $\Gamma$ on $\vfi : \PSub V W$ is denoted $\loch \vfi : (\DSub W \Gamma) \to (\DSub V \Gamma)$ and is called \textdef{restriction} by $\vfi$. A substitution $\sigma : \Delta \to \Gamma$ is then a natural transformation $\sigma \loch : (\DSub \loch \Delta) \to (\DSub \loch \Gamma)$.

\subsection{Types} A type $\Gamma \sez T \type$ is a \textdef{dependent presheaf} over $\Gamma$. In categorical language, this is a functor $T : \paren{\int_{\catW} \Gamma}\op \to \Set$. We denote its action on an object $(W, \gamma)$, where $\gamma : \DSub W \Gamma$, as $T \dsub \gamma$; the elements $t \in T \dsub \gamma$ will be called \textdef{defining terms} and denoted $W \Dsez t : T \dsub \gamma$. The action of $T$ on a morphism $(\vfi | \gamma) : (V, \gamma \vfi) \to (W, \gamma)$ is denoted $\loch \psub \vfi : T \dsub \gamma \to T \dsub{\gamma \vfi}$ and is again called \textdef{restriction}. We have now defined $\Ty(\Gamma)$.

Given a substitution $\sigma : \Delta \to \Gamma$, we need an action $\loch \sub \sigma : \Ty(\Gamma) \to \Ty(\Delta)$. This is defined by setting $T \sub \sigma \dsub \delta := T \dsub{\sigma \delta}$ and defining $\loch \psub \vfi^{T[\sigma]} : T \sub \sigma \dsub \delta \to T \sub \sigma \dsub{\delta \vfi}$ as $\loch \psub \vfi^T : T \dsub{\sigma \delta} \to T \dsub{\sigma \delta \vfi}$.

\subsection{Terms} A term $\Gamma \sez t : T$ consists of, for every $\gamma : \DSub W \Gamma$, a defining term $W \Dsez t \dsub \gamma : T \dsub \gamma$. Moreover, this must be natural in $W$, i.e. for every $\vfi : V \to W$, we require $t \dsub \gamma \psub \vfi = t \dsub{\gamma \vfi}$.

\subsection{The empty context} We set $(\DSub W ()) = \accol{\emptysub}$. The unique substitutions $\Gamma \to ()$ will also be denoted $\emptysub$.

\subsection{Context extension} We set $(\DSub W \Gamma.T) = \set{(\gamma, t)}{\gamma : \DSub W \Gamma \textand W \Dsez t : T \dsub \gamma}$, and $(\gamma, t)\vfi = (\gamma \vfi, t \psub \vfi)$. Of course $\pi(\gamma, t) = \gamma$ and $\xi \dsub{(\gamma, t)} = t$. In variable notation, we will write $(\gamma, t/\var x)$ for $(\gamma, t)$.

\subsection{Yoneda-embedding} There is a fully faithful embedding $\yoneda : \catW \to \widehat \catW$, called the Yoneda embedding, given by $(\DSub V {\yoneda W}) := (\PSub V W)$. Fully faithful means that $(\PSub V W) \cong (\yoneda V \to \yoneda W)$. We have moreover that $(\DSub V \Gamma) \cong (\yoneda V \to \Gamma)$ and $\Tm(\yoneda V, T) \cong T \dsub \id_V$ meaning that terms $\yoneda V \sez t : T$ correspond to defining terms $V \Dsez t \dsub \id : T \dsub \id$. We will omit notations for each of these isomorphisms, effectively treating them as equality.

\section{$\Sigma$-types}
\begin{definition}
	We say that a CwF \textdef{supports $\Sigma$-types} if it is closed under the following rules:
	\begin{equation}
		\inference{\Gamma \sez A \type \qquad \Gamma.A \sez B \type}{\Gamma \sez \Sigma A B \type}{} \qquad
		\inference{\Gamma \sez a : A \qquad \Gamma \sez b : B[\id, a]}{\Gamma \sez (a , b) : \Sigma A  B}{}
	\end{equation}
	\begin{equation}
		\inference{\Gamma \sez p : \Sigma A B}{\Gamma \sez \fst\,p : A}{} \qquad
		\inference{\Gamma \sez p : \Sigma A B}{\Gamma \sez \snd\,p : B[\id, \fst\,p]}{}
	\end{equation}
	where $(\fst, \snd)$ and $(\loch, \loch)$ are inverses and all four operations are natural in $\Gamma$:
	\begin{align*}
		(\Sigma A B)[\sigma] &= \Sigma (A[\sigma])(B[\sigma \subext]), \\
		(a , b)[\sigma] &= (a[\sigma] , b[\sigma]), \qquad
		(\fst\,p)[\sigma] = \fst(p[\sigma]), \qquad
		(\snd\,p)[\sigma] = \snd(p[\sigma]).
	\end{align*}
\end{definition}
\begin{proposition}
	Every presheaf category supports $\Sigma$-types.
\end{proposition}
\begin{proof}
	Given $\gamma : \DSub W \Gamma$, we set $(\Sigma A B) \dsub \gamma = \set{(a, b)}{W \Dsez a : A \dsub \gamma \textand W \Dsez b : B \dsub{\gamma, a}}$, and $(a, b) \psub \vfi = (a \psub \vfi, b \psub \vfi)$, which is natural in $\Gamma$.
	
	We define the pair term $(a, b)$ by $(a, b) \dsub \gamma = (a \dsub \gamma, b \dsub \gamma)$; $\fst\,p$ by $(\fst\,p) \dsub \gamma = p \dsub \gamma_1$ and $(\snd\,p) \dsub \gamma = p \dsub \gamma_2$. All of this is easily seen to be natural in $\Gamma$ and $W$.
\end{proof}

\section{$\Pi$-types}\label{sec:psh-pi-types}
\begin{definition}\label{def:pi-types}
	We say that a CwF \textdef{supports $\Pi$-types} if it is closed under the following rules:
	\begin{equation}
		\inference{\Gamma \sez A \type \qquad \Gamma.A \sez B \type}{\Gamma \sez \Pi A B \type}{} \qquad
		\inference{\Gamma.A \sez b : B}{\Gamma \sez \lambda b : \Pi A B}{} \qquad
		\inference{\Gamma \sez f : \Pi A B}{\Gamma.A \sez \ap f : B}{}
	\end{equation}
	such that $\ap$ and $\lambda$ are inverses and such that all three operations commute with substitution:
	\begin{align*}
		(\Pi A B)[\sigma] = \Pi (A[\sigma])(B[\sigma \subext]), \qquad
		(\lambda b)[\sigma] = \lambda(b[\sigma \subext]), \qquad
		(\ap f)[\sigma \subext] = \ap (f[\sigma])
	\end{align*}
	We will write $f\,a$ for $(\ap\,f)\sub{\id, a}$.
\end{definition}
\begin{proposition}
	Every presheaf category supports $\Pi$-types.
\end{proposition}
\begin{proof}
	Given $\gamma : \DSub W \Gamma$, we set $(\Pi A B) \dsub \gamma = \set{\dlambda b}{\yoneda W.A[\gamma] \sez b : B[\gamma \subext]}$, where the label $\dlambda$ is included for clarity but can be implemented as the identity function; and $(\dlambda b) \psub \vfi = \dlambda (b \sub{\vfi \subext})$.
	To see that this is natural in $\Gamma$, take $\sigma : \Delta \to \Gamma$ and $\delta : \DSub W \Delta$ and unfold the definitions of $(\Pi A B)\sub \sigma \dsub \delta$ and $(\Pi (A \sub \sigma) (B \sub{\sigma \subext})) \dsub \delta$.
	
	We define $\lambda b$ by $(\lambda b) \dsub \gamma = \dlambda(b[\gamma \subext])$. To see that $\lambda b$ is a term:
	\begin{equation}
		(\lambda b) \dsub \gamma \psub \vfi
		= (\dlambda (b \sub{\gamma\subext})) \psub \vfi
		= \dlambda (b \sub{\gamma\subext} \sub{\vfi\subext})
		= \dlambda (b \sub{(\gamma \vfi)\subext}) = (\dlambda b) \dsub{\gamma \vfi}.
	\end{equation}
	One easily checks that $\lambda$ is natural in $\Gamma$.
	
	Let $\dap$ be the inverse of $\dlambda$. Then $\dap$ satisfies $(\dap\,f)[\vfi \subext] = \dap\,(f \psub \vfi)$. Write $f \cdot a$ for $(\dap\,f) \dsub{\cwfpair{\id}{a}}$. We have
	\begin{equation}
		(\dap (f)) \dsub{\cwfpair{\vfi}{a}}
		= (\dap (f)) \sub{\vfi \subext} \dsub{\cwfpair{\id}{a}}
		= (\dap (f \psub{\vfi})) \dsub{\cwfpair{\id}{a}}
		= f \psub \vfi \cdot a,
	\end{equation}
	so that a defining term $W \Dsez f : (\Pi A B) \dsub \gamma$ is fully determined if we know $f \psub \vfi \cdot a$ for all $\vfi : \PSub V W$ and $V \Dsez a : A \dsub{\gamma \vfi}$.
	Similarly, a term $\Gamma \sez f : \Pi A B$ is fully determined if we know $f \dsub \gamma \cdot a$ for all $\gamma : \DSub V \Gamma$ and $V \Dsez a : A \dsub{\gamma}$.
	
	We define $\ap\,f$ by $(\ap\,f) \dsub{\gamma, a} = f \dsub \gamma \cdot a$. To see that this is a term:
	\begin{align}
		(f \dsub \gamma \cdot a) \psub \vfi
		&= (\dap\,(f \dsub \gamma)) \dsub{\id, a} \psub \vfi
		= (\dap\,(f \dsub \gamma)) \sub{\vfi \subext} \dsub{\id, a \psub \vfi}
		= (\dap\,(f \dsub{\gamma \vfi})) \dsub{\id, a \psub \vfi}
		\nn \\
		&= (f \dsub{\gamma \vfi}) \cdot (a \psub \vfi)
		= (\ap\,f) \dsub{(\gamma, a) \vfi}.
	\end{align}
	One easily checks that $\ap$ is natural in $\Gamma$.
	
	To see that $\ap\,\lambda b = b$, we can unfold
	\begin{equation}
		(\ap\,\lambda b) \dsub{\gamma, a}
		= (\lambda b) \dsub \gamma \cdot a
		= \dlambda (b \sub{\gamma \subext}) \cdot a
		= b \sub{\gamma \subext} \dsub{\id, a}
		= b \dsub{\gamma, a}.
	\end{equation}
	To see that $\lambda\,\ap\,f = f$:
	\begin{equation}
		(\lambda\,\ap\,f) \dsub \gamma \cdot a
		= \dlambda((\ap\,f) \sub{\gamma \subext}) \cdot a
		= (\ap\,f) \sub{\gamma \subext} \dsub{\id, a}
		= (\ap\,f) \dsub{\gamma, a}
		= f \dsub \gamma \cdot a. \qedhere
	\end{equation}
\end{proof}

\section{Identity type}\label{sec:cwf-idtp}
\begin{definition}
	A CwF \textdef{supports the identity type} if it is closed under the following rules:
	\begin{equation}
		\inference{
			\Gamma \sez A \type \\
			\Gamma \sez a, b : A
		}{\Gamma \sez a \idtp A b \type}{}
		\qquad
		\inference{
			\Gamma \sez a : A
		}{\Gamma \sez \refl\,a : a \idtp A a}{}
		\qquad
		\inference{
			\Gamma \sez a, b : A \\
			\Gamma, \var y : A, \var w : (a[\wknvar y] \idtp{A[\wknvar y]} \var y) \sez C \type \\
			\Gamma \sez e : a \idtp A b \\
			\Gamma \sez c : C[\id, a / \var y, \refl\,a / \var w]
		}{\Gamma \sez \J(a, b, \var y.\var w.C, e, c) : C[\id, b / \var y, e / \var w]}{}
	\end{equation}
	such that all three operations commute with substitution:
	\begin{align*}
		(a \idtp A b)[\sigma] &= \paren{a[\sigma] \idtp{A[\sigma]} b[\sigma]} \\
		(\refl\,a)[\sigma] &= \refl\,(a[\sigma]) \\
		\J(a, b, \var y.\var w.C, p, c)[\sigma] &= \J(a[\sigma], b[\sigma], \var y.\var w.C[\sigma \wknvar y \wknvar w, \var y/\var y, \var w/\var w], p[\sigma], c[\sigma])
	\end{align*}
	and such that $\J(a, a, \var y.\var w.C, \refl\,a, c) = c$.
\end{definition}
\begin{proposition}
	Every presheaf category supports the identity type.\footnote{The identity type in the standard presheaf model, expresses equality of mathematical objects. It supports the reflection rule and axiom K, and not the univalence axiom.}
\end{proposition}
\begin{proof}
	We set $(a \idtp A b)\dsub \gamma$ equal to $\accol \star$ if $a \dsub \gamma = b \dsub \gamma$, and make it empty otherwise. We define $(\refl\,a)\dsub \gamma = \star$ and $\J(a, b, \var y.\var w.C, e, c) \dsub \gamma = c \dsub \gamma$, which is well-typed since $e$ witnesses that $a = b$ and all terms of the identity type are equal.
\end{proof}

\section{Universes}
Assume a functor $\Ty^* : \Ctx\op \to \Set$, and write $\Gamma \sez T \type^*$ for $T \in \Ty^*(\Gamma)$.
\begin{definition}
	We say that a CwF \textdef{supports a universe} for $\Ty^*$ if it is closed under the following rules:
	\begin{equation}
		\inference{\Gamma \ctx}{\Gamma \sez \uni{}^* \type}{} \qquad
		\inference{\Gamma \sez T : \uni{}^*}{\Gamma \sez \El\,T \type^*}{} \qquad
		\inference{\Gamma \sez T \type^*}{\Gamma \sez \tycode T : \uni{}^*}{}
	\end{equation}
	where $\El$ and $\tycode \loch$ are inverses and all three operators commute with substitution:
	\begin{equation}
		\uni{}^*[\sigma] = \uni{}^*, \qquad
		(\El\,T)[\sigma] = \El(T[\sigma]), \qquad
		\tycode T[\sigma] = \tycode{T[\sigma]}.
	\end{equation}
	Note that here, we are switching between term substitution and *-type substitution.
\end{definition}
We assume that the metatheory has Grothendieck universes, i.e. there is a chain
\begin{equation}
	\Set_0 \in \Set_1 \in \Set_2 \in \ldots
\end{equation}
such that every $\Set_k$ is a model of ZF set theory. We say that a type $T \in \Ty(\Gamma)$ has level $k$ if $T \dsub \gamma \in \Set_k$ for every $\gamma : \DSub W \Gamma$. Let $\Ty_k(\Gamma)$ be the set of all level $k$ types over $\Gamma$; this constitutes a functor $\Ty_k : \widehat \catW \op \to \Set_{k+1}$. Write $\Gamma \sez T \type_k$ for $T \in \Ty_k(\Gamma)$.
\begin{proposition}\label{thm:unipsh}
	Every presheaf category (over a base category of level 0) supports a universe $\uni k$ for $\Ty_k$ that is itself of level $k+1$.
\end{proposition}
\begin{proof}
	Given $\gamma : \DSub W \Gamma$, we set $\uni k \dsub \gamma = \set{\dtycode T}{\yoneda W \sez T \type_k} \cong \Ty_k(\yoneda W) \in \Set_{k+1}$. Again, $\dtycode \loch$ is just a label which we add for readability. We set $\dtycode T \psub \vfi = \dtycode{T \sub \vfi}$. Naturality in $\Gamma$ is immediate, as the definition of $\uni k$ does not refer to either $\Gamma$ or $\gamma$.
	
	Given $\Gamma \sez T \type_k$, we define $\Gamma \sez \tycode T : \uni k$ by $\tycode T \dsub \gamma = \dtycode{T \sub \gamma}$, which satisfies
	\begin{equation}
		\tycode T \dsub \gamma \psub \vfi = \dtycode{T \sub \gamma} \psub \vfi = \dtycode{T \sub \gamma \sub \vfi} = \tycode T \dsub{\gamma \vfi}.
	\end{equation}
	
	Write $\dEl$ for the inverse of $\dtycode \loch$. It satisfies $\dEl\,A \sub \vfi = \dEl(A \psub \vfi)$. Given $\Gamma \sez A : \uni k$, we define $\Gamma \sez \El\,A \type_k$ by $\El\,A \dsub \gamma = \dEl(A \dsub \gamma) \dsub \id$. Given $\vfi : V \to W$ and $W \Dsez a : \El\,A \dsub \gamma$, we set $a \psub \vfi^{\El\,A} : \El\,A \dsub{\gamma \vfi}$ equal to $a \psub \vfi^{\dEl(A \dsub \gamma)} : \dEl(A \dsub \gamma) \dsub \vfi$. This is well-typed, because
	\begin{equation}
		\dEl(A \dsub \gamma) \dsub \vfi
		= \dEl(A \dsub \gamma) \sub \vfi \dsub \id
		= \dEl(A \dsub{\gamma \vfi}) \dsub \id
		= \El\,A\dsub{\gamma\vfi}.
	\end{equation}
	To see that $\El\,\tycode T = T$:
	\begin{equation}
		\El\,\tycode T \dsub \gamma
		= \dEl (\tycode T \dsub \gamma) \dsub \id
		= \dEl (\dtycode{T \sub \gamma}) \dsub \id = T \dsub \gamma 
	\end{equation}
	One can check that the substitution operations of $\El\,\tycode T$ and $T$ also match.
	
	Conversely, we show that $\tycode{\El\,A} = A$. To that end, we unpack both completely by applying $\dEl(\loch \dsub \gamma) \dsub \vfi$:
	\begin{align}
		\dEl(\tycode{\El\,A} \dsub \gamma) \dsub \vfi
		&= \dEl(\dtycode{\El\,A \sub \gamma}) \dsub \vfi
		= \El\,A \dsub{\gamma \vfi}
		= \dEl(A \dsub{\gamma \vfi}) \dsub \id, \\
		\dEl(A \dsub \gamma) \dsub \vfi
		&= \dEl(A \dsub \gamma) \sub \vfi \dsub \id
		= \dEl(A \dsub{\gamma \vfi}) \dsub \id. \qedhere
	\end{align}
\end{proof}
We say that a type $T \in \Ty(\Gamma)$ is a \textdef{proposition} if for every $\gamma : \DSub W \Gamma$, we have $T \dsub \gamma \subseteq \accol \star$. We denote this as $T \in \Prop(\Gamma)$ or $\Gamma \sez T \prop$.
\begin{proposition}
	Every presheaf category (over a base category of level 0) supports a universe $\Prop$ of propositions that is itself of level 0.
\end{proposition}
\begin{proof}
	Completely analogous.
\end{proof}
One easily shows that $\Prop$ is closed under $\top$, $\bot$, $\wedge$ and $\vee$. It also clearly contains the identity types. There is an absurd eliminator for $\bot$ and we can construct systems to eliminate proofs of $\vee$.

\section{Glueing}
\begin{definition}
	A CwF \textdef{supports glueing} if it is closed under the following rules:
	\begin{equation*}
		\inference{
			\Gamma \sez P \prop \\
			\Gamma.P \sez T \type \\
			\Gamma.P \sez f : T \to A[\pi] \\
			\Gamma \sez A \type
		}{\Gamma \sez \Gluesys{A}{\Gluesysclauseb{P}{T}{f}} \type}{}, \qquad
		\inference{
			\Gamma \sez \Gluesys{A}{\Gluesysclauseb{P}{T}{f}} \type \\
			\Gamma \sez a : A \\
			\Gamma.P \sez t : T \\
			\Gamma.P \sez ft = a[\pi] : T
		}{\Gamma \sez \gluesys{a}{\sysclauseb{P}{t}} : \Gluesys{A}{\Gluesysclauseb{P}{T}{f}}}{},
	\end{equation*}
	\begin{equation*}
		\inference{
			\Gamma \sez b : \Gluesys{A}{\Gluesysclauseb{P}{T}{f}}
		}{\Gamma \sez \ungluesys{\sysclauseb P f} b : A}{},
	\end{equation*}
	naturally in $\Gamma$, such that
	\begin{align*}
		\Gluesys{A}{\Gluesysclauseb{\top}{T}{f}} &= T[\cwfpair \id \star], \\
		\gluesys{a}{\sysclauseb{\top}{t}} &= t[\cwfpair \id \star], \\
		\ungluesys{\sysclauseb \top f} b &= f[\cwfpair \id \star]\,b, \\
		\ungluesys{\sysclauseb P f} (\gluesys{a}{\sysclauseb{P}{t}}) &= a, \\
		\gluesys{\ungluesys{\sysclauseb P f} b}{\sysclauseb{P}{b[\pi]}} &= b.
	\end{align*}
\end{definition}
\begin{proposition}
	Every presheaf category supports glueing.
\end{proposition}
\begin{proof}
	We assume given the prerequisites of the type former. Write $G = \Gluesys{A}{\Gluesysclauseb{P}{T}{f}}$.
	\begin{description}
		\item[The type] We define $G \dsub \gamma$ by case distinction on $P \dsub \gamma$:
		\begin{enumerate}
			\item If $P \dsub \gamma = \accol \star$, then we set $G \dsub \gamma = T \dsub{\gamma, \star}$.
			\item If $P \dsub \gamma = \eset$, we let $P \dsub \gamma$ be the set of pairs $(a \mapsfrom t)$ where $a : A \dsub \gamma$ and $\yoneda W.P[\gamma] \sez t : T[\gamma \subext]$, such that for every $\vfi$ for which $P \dsub{\gamma \vfi} = \accol \star$, the application $f \dsub{\gamma \vfi, \star} \cdot t \dsub{\vfi, \star}$ is equal to $a \psub \vfi$.
		\end{enumerate}
		Given $g : G \dsub \gamma$ and $\vfi : \PSub V W$, we need to define $g \psub \vfi$.
		\begin{enumerate}
			\item If $P \dsub \gamma = P \dsub{\gamma \vfi} = \accol \star$, then we use the definition from $T$.
			\item If $P \dsub \gamma = P \dsub{\gamma \vfi} = \eset$, then we set $(a \mapsfrom t) \psub \vfi = (a \psub \vfi \mapsfrom t[\vfi \subext])$.
			\item If $P \dsub \gamma = \eset$ and $P \dsub{\gamma \vfi} = \accol \star$, then we set $(a \mapsfrom t) \psub \vfi = t \dsub{\vfi, \star}$.
		\end{enumerate}
		One can check that this definition preserves composition and identity, and that this entire construction is natural in $\Gamma$.
		
		\item[The constructor] Write $g = \gluesys{a}{\sysclauseb{P}{t}}$. We define $g \dsub \gamma$ by case distinction on $P \dsub \gamma$:
		\begin{enumerate}
			\item If $P \dsub \gamma = \accol \star$, then we set $g \dsub \gamma = t \dsub{\gamma, \star}$.
			\item If $P \dsub \gamma = \eset$, then we set $g \dsub \gamma = (a \dsub \gamma \mapsfrom t \sub{\gamma \subext})$.
		\end{enumerate}
		By case distinction, it is easy to check that this is natural in the domain $W$ of $\gamma$. Naturality in $\Gamma$ is straightforward.
		
		\item[The eliminator] Write $u = \ungluesys{\sysclauseb P f} b$. We define $u \dsub \gamma$ by case distinction on $P \dsub \gamma$:
		\begin{enumerate}
			\item If $P \dsub \gamma = \accol \star$, then we set $u \dsub \gamma = f \dsub{\gamma, \star} \cdot b \dsub{\gamma}$.
			\item If $P \dsub \gamma = \eset$, then $b \dsub \gamma$ is of the form $(a \mapsfrom t)$ and we set $u \dsub \gamma = a$.
		\end{enumerate}
		Naturality in the domain $W$ of $\gamma$ is evident when we consider non-cross-case restrictions. Naturality for cross-case restrictions is asserted by the condition on pairs $(a \mapsfrom t)$. Again, naturality in $\Gamma$ is straightforward.
		
		\item[The $\beta$-rule] Pick $\gamma : \DSub W \Gamma$. Write $g = \gluesys{a}{\sysclauseb{P}{t}}$.
		\begin{enumerate}
			\item If $P \dsub \gamma = \accol \star$, then we get
			\begin{equation}
				\mathrm{LHS} \dsub \gamma
				= f \dsub{\gamma, \star} \cdot g \dsub{\gamma}
				= f \dsub{\gamma, \star} \cdot t \dsub{\gamma, \star}
				= (f\,t) \dsub{\gamma, \star} = a \dsub{\gamma}
			\end{equation}
			by the premise of the $\glue$ rule.
			\item If $P \dsub \gamma = \eset$, then we have $g \dsub \gamma = (a \dsub \gamma \mapsfrom t \sub{\gamma \subext})$ and $\unglue$ simply extracts the first component.
		\end{enumerate}
		
		\item[The $\eta$-rule] Pick $\gamma : \DSub W \Gamma$. Write $u = \ungluesys{\sysclauseb P f} b$.
		\begin{enumerate}
			\item If $P \dsub \gamma = \accol \star$, then we have $\mathrm{LHS} \dsub \gamma = b \sub \pi \dsub{\gamma, \star} = b \dsub \gamma$.
			\item If $P \dsub \gamma = \eset$, then $b \dsub \gamma$ has the form $(a \mapsfrom t)$ and we get
			\begin{equation}
				\mathrm{LHS} \dsub \gamma
				= (u \dsub \gamma \mapsfrom b \sub \pi \sub{\gamma \subext})
				= (a \mapsfrom b \sub{\gamma \pi})
				=^{(\dagger)} (a \mapsfrom t).
			\end{equation}
			The last step $(\dagger)$ is less than trivial. We show that $\yoneda W.P[\gamma] \sez b \sub{\gamma \pi} = t : T$. Pick any $(\vfi, \star) : \DSub V {(\yoneda W.P[\gamma])}$. If you manage to pick one, then $P \dsub{\gamma \vfi} = \accol \star$, so $b \sub{\gamma \pi} \dsub{\vfi, \star} = b \dsub{\gamma \vfi} = (a \mapsfrom t) \psub \vfi = t \dsub{\vfi, \star}$. \qedhere
		\end{enumerate} 
	\end{description}
\end{proof}

\section{Welding}
\begin{definition}
	A CwF \textdef{supports welding} if it is closed under the following rules:
	\begin{equation}
		\inference{
			\Gamma \sez P \prop \\
			\Gamma, \var p : P \sez T \type \\
			\Gamma, \var p : P \sez f : A[\pi] \to T \\
			\Gamma \sez A \type
		}{\Gamma \sez \Weldsys{A}{\Weldsysclauseb{\var p : P}{T}{f}} \type}{}, \qquad
		\inference{
			\Gamma \sez \Weldsys{A}{\Weldsysclauseb{\var p : P}{T}{f}} \type \\
			\Gamma \sez a : A
		}{\Gamma \sez \weldsys{\sysclauseb{\var p : P}{f}} a : \Weldsys{A}{\Weldsysclauseb{\var p : P}{T}{f}}}{},
	\end{equation}
	\begin{equation}
		\inference{
			\Gamma, \var y : \Weldsys{A}{\Weldsysclauseb{\var p : P}{T}{f}} \sez C \type \\
			\Gamma, \var p : P, \var y : T \sez d : C[\wknvar p \subext] \\
			\Gamma, \var x : A \sez c : C[\wknvar x, \weldsys{\sysclauseb{\var p : P[\wknvar x \subext]}{f[\wknvar x \subext]}} \var x/\var y] \\
			\Gamma, \var p : P, \var x : A[\wknvar p] \sez d[\wknvar x, f[\wknvar x]\,\var x/\var y] = c[\wknvar p \subext] : C[\wknvar p \wknvar x, f[\wknvar x]\,\var x/\var y]\\
			\Gamma \sez b : \Weldsys{A}{\Weldsysclauseb{\var p : P}{T}{f}}
		}{\Gamma \sez \ind_\Weld(\var y.C, \sys{\sysclauseb{\var p : P}{\var y.d}}, \var x.c, b) : C[\id, b/\var y]}{}
	\end{equation}
	naturally in $\Gamma$, such that
	\begin{align*}
		\Weldsys{A}{\Weldsysclauseb{\var p : \top}{T}{f}} &= T[\id, \star/\var p], \\
		\weldsys{\sysclauseb{\var p : \top}{f}} a &= f[\id, \star/\var p] a, \\
		\ind_\Weld(\var y.C, \sys{\sysclauseb{\var p : \top}{\var y.d}}, \var x.c, b) &= d[\id, \star / \var p, b / \var y], \\
		\ind_\Weld(\var y.C, \sys{\sysclauseb{\var p : P}{\var y.d}}, \var x.c, \weldsys{\sysclauseb{\var p : P}{f}} a) &= c[\id, a/\var x].
	\end{align*}
\end{definition}
\begin{proposition}
	Every presheaf category supports welding.
\end{proposition}
\begin{proof}
	We assume given the prerequisites of the type former. Write $\Omega = \Weldsys{A}{\Weldsysclauseb{\var p : P}{T}{f}}$.
	\begin{description}
		\item[The type] We define $W \dsub \gamma$ by case distinction on $P \dsub \gamma$:
		\begin{enumerate}
			\item If $P \dsub \gamma = \accol \star$, then we set $\Omega \dsub \gamma = T \dsub{\gamma, \star/\var p}$.
			\item If $P \dsub \gamma = \eset$, then we set $\Omega \dsub \gamma = \set{\dweld\,a}{a : A \dsub \gamma} \cong A \dsub \gamma$. Once more, $\dweld$ is a meaningless label that we add for readability.
		\end{enumerate}
		Given $w : \Omega \dsub \gamma$ and $\vfi : \PSub V W$, we need to define $w \psub \vfi$.
		\begin{enumerate}
			\item If $P \dsub \gamma = P \dsub{\gamma \vfi} = \accol \star$, then we use the definition from $T$.
			\item If $P \dsub \gamma = P \dsub{\gamma \vfi} = \eset$, then we set $(\dweld\,a) \psub \vfi = \dweld (a \psub \vfi)$
			\item If $P \dsub \gamma = \eset$ and $P \dsub{\gamma \vfi} = \accol \star$, then we set $(\dweld\,a) \psub \vfi = f \dsub{\gamma \vfi, \star / \var p} \cdot a \psub \vfi$.
		\end{enumerate}
		One can check that this definition preserves composition and identity, and that this entire construction is natural in $\Gamma$.
		
		\item[The constructor] Write $w = \weldsys{\sysclauseb{\var p : P}{f}} a$.
		\begin{enumerate}
			\item If $P \dsub \gamma = \accol \star$, then we set $w \dsub \gamma = f \dsub{\gamma, \star / \var p} \cdot a \dsub \gamma$.
			\item If $P \dsub \gamma = \eset$, then we set $w \dsub \gamma = \dweld(a \dsub \gamma)$.
		\end{enumerate}
		This is easily checked to be natural in $\Gamma$ and the domain $W$ of $\gamma$.
		
		\item[The eliminator] Write $z = \ind_\Weld(\var y.C, \sys{\sysclauseb{\var p : P}{\var y.d}}, \var x.c, b)$.
		\begin{enumerate}
			\item If $P \dsub \gamma = \accol \star$, then $b \dsub \gamma : T \dsub{\gamma, \star / \var p}$, and we can set $z \dsub \gamma = d \dsub{\gamma, \star / \var p, b \dsub \gamma / \var y}$.
			\item If $P \dsub \gamma = \eset$, then $\dunweld(b \dsub \gamma) : A \dsub \gamma$ and we can set $z \dsub \gamma = c \dsub{\gamma, \dunweld(b \dsub \gamma) / \var x}$, where $\dunweld$ removes the $\dweld$ label.
		\end{enumerate}
		We need to show that this is natural in the domain $W$ of $\gamma$, which is only difficult in the cross-case-scenario. So pick $\vfi : \PSub V W$ such that $P \dsub \gamma = \eset$ and $P \dsub{\gamma \vfi} = \accol \star$. We need to show that $z \dsub \gamma \psub \vfi = z \dsub{\gamma \vfi}$. Write $b \dsub \gamma = \dweld\,a$. We have
		\begin{align}
			z \dsub \gamma \psub \vfi
			&= c \dsub{\gamma, a / \var x} \psub \vfi
			= c \dsub{\gamma \vfi, a \psub \vfi / \var x}, \\ \nn
			z \dsub{\gamma \vfi}
			&= d \dsub{\gamma \vfi, \star / \var p, b \dsub{\gamma \vfi} / \var y}
			= d \dsub{\gamma \vfi, \star / \var p, (\dweld\,a) \psub \vfi / \var y} \\ \nn
			&= d \dsub{\gamma \vfi, \star / \var p, f \dsub{\gamma \vfi, \star / \var p} \cdot a \psub \vfi / \var y}.
		\end{align}
		From the premises, we know that $d[\wknvar x, f[\wknvar x]\,\var x/\var y] = c[\wknvar p \subext]$. Applying $\loch \dsub{\gamma \vfi, \star / \var p, a \psub \vfi / \var x}$ to this equation yields
		\begin{equation}
			d \dsub{\gamma \vfi, \star / \var p, f \dsub{\gamma \vfi, \star / \var p} \cdot a \psub \vfi / \var y} = c \dsub{\gamma \vfi, a \psub \vfi / \var x}.
		\end{equation}
		
		\item[The $\beta$-rule] Write $w = \weldsys{\sysclauseb{\var p : P}{f}} a$.
		\begin{enumerate}
			\item If $P \dsub \gamma = \accol \star$, then
			\begin{align}
				\mathrm{LHS} \dsub \gamma
				&= d \dsub{\gamma, \star / \var p, w \dsub \gamma / \var y}
				= d \dsub{\gamma, \star / \var p, f \dsub{\gamma, \star / \var p} \cdot a \dsub \gamma / \var y}, \\ \nn
				\mathrm{RHS} \dsub \gamma
				&= c \dsub{\gamma, a \dsub \gamma / \var x}.
			\end{align}
			Again, the premises give us this equality.
			\item If $P \dsub \gamma = \eset$, then the equality is trivial. \qedhere
		\end{enumerate}
	\end{description}
\end{proof}

\chapter{Internalizing transformations of semantics}\label{ch:cwf-morphisms}
Given categories with families (CwFs) $\cat C$ and $\cat D$, we can consider functors $F : \cat C \to \cat D$ that sufficiently preserve the CwF structure to preserve semantical truth (though not necessarily falsehood) of judgements. Such functors will be called \textbf{morphisms of CwFs}.

In \cref{sec:cwf-transform} of this chapter, we are concerned with how we can internalize a CwF morphism and even more interestingly, how we can internalize a natural transformation between CwF morphisms. That is, we want to answer the question: What inference rules become meaningful when we know of the existence of (natural transformations between) CwF morphisms? Finally, we consider adjoint CwF morphisms, which of course give rise to unit and co-unit natural transformations.

In \cref{sec:psh-transform}, we delve deeper and study the implications of functors, natural transformations and adjunctions between categories $\catV$ and $\catW$ for the CwFs $\widehat{\catV}$ and $\widehat{\catW}$.

Throughout the chapter, we will need to annotate symbols like $\Ty$, $\sez$ and $\Dsez$ with the CwF that we are talking about.

\section{Categories with families}\label{sec:cwf-transform}

\subsection{Morphisms of CwFs}
\begin{definition}\label{def:morphism-of-cwfs}
	A \textbf{morphism of CwFs} $F : \cat C \to \cat D$ consists of:
	\begin{enumerate}
		\item A functor $F_\Ctx : \cat C \to \cat D$,
		\item A natural transformation $F_\Ty : \Ty_{\cat C} \to \Ty_{\cat D} \circ F_\Ctx$,
		\item A natural transformation $F_\Tm : \Tm_{\cat C} \to \Tm_{\cat D} \circ F_{\int}$, where $F_{\int} : \int_{\cat C} \Ty_{\cat C} \to \int_{\cat D} \Ty_{\cat D}$ is easily constructed from $F_\Ctx$ and $F_\Ty$,
		\item such that $F_\Ctx () = ()$,
		\item such that $F_\Ctx(\Gamma.T) = (F_\Ctx \Gamma).(F_\Ty T)$, $F_\Ctx \pi = \pi$ and $F_\Tm \xi = \xi$.
	\end{enumerate}
	The images of a context $\Gamma$, a substitution $\sigma$, a type $T$ and a term $t$ are also denoted $F\Gamma$, $F \sigma$, $FT$ and $\ftrtm F t$ respectively. We choose to denote the action of terms differently because CwF morphisms act very differently on types and on terms of the universe: in general $\ftrtm F A$ will be quite different from $F(\El\,A)$.
\end{definition}
A morphism of CwFs $F : \cat C \to \cat D$ is easy to internalize:
\begin{equation}
	\inference{\Gamma \sez_\cat C \Ctx}{F\Gamma \sez_\cat D \Ctx}{} \qquad
	\inference{\Gamma \sez_\cat C T \type}{F\Gamma \sez_\cat D FC \type}{} \qquad
	\inference{\Gamma \sez_\cat C t : T}{F\Gamma \sez_\cat D \ftrtm F t : FT}{}
\end{equation}
with equations for context formation:
\begin{equation}
	F () = (), \qquad F(\Gamma.T) = F\Gamma.FT,
\end{equation}
substitution:
\begin{equation}
	F\id = \id, \qquad
	F(\tau \sigma) = (F\tau) (F\sigma), \qquad
	F(T[\sigma]) = (FT)[F\sigma], \qquad
	\ftrtm F {(t[\sigma])} = (\ftrtm F t)[F\sigma],
\end{equation}
and pairing and projecting:
\begin{equation}
	F \pi = \pi, \qquad
	\ftrtm F \xi = \xi, \qquad
	F(\cwfpair \sigma t) = (F \sigma, \ftrtm F t).
\end{equation}
When using variable notation, the equation $\ftrtm F \xi = \xi$ inspires us to write $F(\Gamma, \var x : T) = (F\Gamma, \ftrtm{F}{\var x} : FT)$. Here, $\ftrtm F {\var x}$ is to be regarded as an atomic variable name, which happens to be equal to the compound term of the same notation. Then we get $F \wknvar x = \wkn{\ftrtm F x}$.

\subsection{Natural transformations of CwFs}
In this section, we consider morphisms of CwFs $F, G : \cat C \to \cat D$ and a natural transformation $\nu : F \to G$ between the underlying functors. It is clear that for any context $\Gamma$, we get a substitution between its respective images:
\begin{equation}
	\inference{\Gamma \sez_\cat C \Ctx}{\nu : F\Gamma \xrightarrow{\cat D} G\Gamma}{}
\end{equation}
Just like we did for $\pi$ and $\xi$, we will omit the index $\Gamma$ on $\nu$. For extended contexts $\Gamma, \var x : T$, naturality on $\wknvar x$ shows that $\wkn{\ftrtm G \var x} \circ \nu = \nu \circ \wkn{\ftrtm F \var x} : (F\Gamma, \ftrvar F x : FT) \to G\Gamma$.
\begin{proposition}\label{thm:nattrans-function}
	We have an operation $\nu_\loch(\loch)$ for applying $\nu$ to terms:
	\begin{equation}
		\inference{
			\Gamma \sez_{\cat C} T \type \qquad
			\sigma : \Delta \xrightarrow{\cat D} F \Gamma \qquad
			\Delta \sez_{\cat D} t : (FT)[\sigma]
		}{\Delta \sez_{\cat D} \nu_\sigma(t) : (GT)[\nu \sigma]}{}.
	\end{equation}
	\begin{enumerate}
		\item This operation is natural in $\Delta$, i.e. $\nu_\sigma(t)[\tau] = \nu_{\sigma \tau}(t[\tau])$.
		\item It is also natural in $\Gamma$, i.e. if $\rho : \Gamma' \xrightarrow{\cat C} \Gamma$ and $\sigma : \Delta \xrightarrow{\cat D} F\Gamma'$, then $\nu_{F\rho \circ \sigma}(t) = \nu_\sigma(t)$. For this reason, we will write $\nu(t)$ for $\nu_\sigma(t)$.
		\item For a context $\Gamma, \var x : T$, we have $\nu = (\nu \wkn{\ftrtm F \var x}, \nu(\ftrtm F \var x) / \ftrtm G \var x) : (F \Gamma, \ftrtm F \var x : FT) \xrightarrow{\cat D} (G \Gamma, \ftrtm G \var x : GT)$.
		\item We have $\nu(\ftrtm F {t'}) = \ftrtm G {t'}[\nu]$.
		\item We have $(\nu \mu)(t) = \nu(\mu(t))$ and $\id(t) = t$.
		\item We have $\ftrtm{R}{(\nu(t))} = (R\nu)(\ftrtm R t)$.
	\end{enumerate}
\end{proposition}
\begin{proof}
	The unpairing of $\nu$ requires that $F\Gamma, \ftrtm F \var x : FT \sez \ftrtm G \var x[\nu] = \nu_{(\wkn{\ftrtm F \var x})}(\ftrtm F \var x) : (GT)[\nu \wkn{\ftrtm F \var x}]$. Now $t = (\ftrtm F \var x)[\sigma, t / \ftrtm F \var x]$, so naturality requires us to define $\nu_\sigma(t) = (\ftrtm G \var x)[\nu(\sigma, t / \ftrtm F \var x)]$.
	\begin{enumerate}
		\item This is easily seen to be natural in $\Delta$.
		\item For naturality in $\Gamma$:
		\begin{align}
			\nu_{F \rho \circ \sigma}(t)
			= (\ftrtm G \var x)[\nu(F \rho \circ \sigma, t / \ftrtm F \var x)]
			&= (\ftrtm G \var x)[\nu(F\rho \subext)(\sigma, t / \ftrtm F \var x)] \\ \nn
			&= (\ftrtm G \var x)[(G\rho \subext)\nu(\sigma, t / \ftrtm F \var x)]
			= (\ftrtm G \var x)[\nu(\sigma, t / \ftrtm F \var x)]
			= \nu_\sigma(t).
		\end{align}
		\item We have $\wkn{\ftrtm G \var x} \nu = \nu \wkn{\ftrtm F \var x}$ and $\ftrtm G \var x[\nu] = \ftrtm G \var x[\nu(\pi, \ftrtm F \var x / \ftrtm F \var x)] = \nu(\ftrtm F \var x)$.
		\item $\nu_\id(\ftrtm F {t'}) = (\ftrtm G \var x)[\nu (\id, \ftrtm F {t'} / \ftrtm F \var x)] = (\ftrtm G \var x)[\nu \circ F(\id, t' / \var x)] = (\ftrtm G \var x)[G(\id, t' / \var x) \circ \nu] = \ftrtm{G}{t'}[\nu]$.
		\item Assume $\mu : E \to F$ and $\nu : F \to G$. We have
		\begin{equation}
			(\nu \mu)(t)
			= \ftrtm G \var x[\nu \mu(\sigma, t / \ftrtm E \var x)]
			= \ftrtm G \var x[\nu (\mu \wkn{\ftrtm E \var x}, \mu(\ftrtm E \var x)/\ftrtm F \var x)(\sigma, t / \ftrtm E \var x)] = \ftrtm G \var x[\nu(\mu \sigma, \mu(t) / \ftrtm F \var x)] = \nu(\mu(t)).
		\end{equation}
		\item This is immediate from the definition. \qedhere
	\end{enumerate}
\end{proof}

\subsection{Adjoint morphisms of CwFs}
In this section, we consider functors $L : \cat C \to \cat D$ and $R : \cat D \to \cat C$, where $L$ may be and $R$ is a morphism of CwFs, such that $\alpha : L \dashv R$. Then $LR : \cat D \to \cat D$ will be a comonad and $RL : \cat C \to \cat C$ will be a monad. Of course, we have unit and co-unit natural transformations
\begin{equation}
	\eps : LR \to \Id_{\cat D}, \qquad \eta : \Id_{\cat C} \to RL,
\end{equation}
such that $\eps L \circ L \eta = \id_L$ and $R\eps \circ \eta R = \id_R$. The isomorphism $\alpha : \Hom(L\loch, \loch) \cong \Hom(\loch, R\loch)$ can be retrieved from unit and co-unit:
\begin{equation}
	\alpha(\sigma) = R\sigma \circ \eta, \qquad
	\alpha\inv(\tau) = \eps \circ L\tau,
\end{equation}
and vice versa:
\begin{equation}
	\eta = \alpha(\id), \qquad \eps = \alpha\inv(\id).
\end{equation}
Finally, $\alpha$ is natural in the following sense:
\begin{equation}
	\alpha(\tau \circ \sigma \circ L\rho) = R\tau \circ \alpha(\sigma) \circ \rho, \qquad
	\alpha\inv(R\tau \circ \sigma \circ \rho) = \tau \circ \alpha\inv(\sigma) \circ L\rho.
\end{equation}
If $L$ is a CwF morphism, then \cref{thm:nattrans-function} gives us functions $\eta : T \to (RLT)[\eta]$ and $\eps : LRT \to T[\eps]$. Moreover, $\eta(t) = (\ftrtm{RL}{t})[\eta]$ and $\eps(\ftrtm{LR}{t}) = t[\eps]$.

\begin{proposition}\label{thm:adjunction-rules}
	Assume $\Gamma \sez_{\cat D} T \type$. We have mutually inverse rules:
	\begin{equation}
		\inference{
			\sigma : L\Delta \xrightarrow{\cat D} \Gamma \\
			L\Delta \sez_{\cat D} t : T[\sigma]
		}{\Delta \sez_{\cat C} \alpha_\sigma(t) : (RT)[\alpha(\sigma)]}{}
		\qquad \qquad
		\inference{
			\tau : \Delta \xrightarrow{\cat D} R\Gamma \\
			\Delta \sez_{\cat C} t' : (RT)[\tau]
		}{L\Delta \sez_{\cat D} \alpha\inv_{\alpha\inv(\tau)}(t') : T[\alpha\inv(\tau)]}{}
	\end{equation}
	\begin{enumerate}
		\item These operations are natural in $\Delta$, i.e. $\alpha_\sigma(t)[\tau] = \alpha_{\sigma \circ L \tau}(t[L\tau])$.
		\item These operations are natural in $\Gamma$, i.e if $\rho : \Gamma' \to \Gamma$ and $\sigma : L\Delta \xrightarrow{\cat D} \Gamma'$, then $\alpha_{\rho \sigma}(t) = \alpha_\sigma(t)$. For this reason, we will omit the subscript on $\alpha$.
		\item $\alpha(t) = (\ftrtm R t)[\eta]$.
		\item If $L$ is a CwF morphism, then $\alpha\inv(t') = \eps(\ftrtm L {t'})$. In general, we have $\alpha\inv(t') = \xi[\eps \circ L(\eta, t')]$.
	\end{enumerate}
\end{proposition}
\begin{proof}
	We have $\alpha : (L\Delta \to (\Gamma, \var x : T)) \cong (\Delta \to (R\Gamma, \ftrtm{R}{\var x} : RT))$. We will define $\alpha_\sigma(t)$ and $\alpha_{\alpha\inv(\tau)}(t')$ by the (equivalent) equations:
	\begin{equation}
		\alpha(\sigma, t / \var x) = (\alpha(\sigma), \alpha_\sigma(t) / \ftrtm R \var x),
		\qquad
		\alpha\inv(\tau, t' / \ftrtm R \var x) = (\alpha\inv(\tau), \alpha\inv_{\alpha\inv(\tau)}(t') / \var x).
	\end{equation}
	The first components of these equations are correct:
	\begin{align}
		\wkn{\ftrtm R \var x} \circ \alpha(\sigma, t / \var x)
		= \alpha(\wknvar x \circ (\sigma, t / \var x))
		= \alpha(\sigma), \\
		\wknvar x \circ \alpha\inv(\tau, t'/\ftrtm R \var x)
		= \alpha\inv(\wkn{\ftrtm R \var x} \circ (\tau, t' / \ftrtm R \var x))
		= \alpha\inv(\tau).
	\end{align}
	\begin{enumerate}
		\item This follows from naturality of the adjunction $\alpha$.
		\item Note that $R(\rho \subext) = (R \rho) \subext$. We have $\alpha_{\rho \sigma}(t) = \ftrvar R x[\alpha(\rho \sigma, t / \var x)] = \ftrvar R x[R \rho \subext \circ \alpha(\sigma, t / \var x)] = \ftrtm R x[\alpha(\sigma, t / \var x)] = \alpha_\sigma(t)$. Then naturality for $\alpha\inv$ holds because it is inverse to $\alpha$.
		\item We have $\alpha_\sigma(t) = \ftrtm R \var x[\alpha(\sigma, t / \var x)] = \ftrtm R \var x[(R\sigma, \ftrtm R t / \ftrtm R \var x) \circ \eta] = \ftrtm R t[\eta]$.
		\item We have $\alpha\inv_{\alpha\inv(\tau)}(t') = \alpha\inv_{\alpha\inv(\eta)}(t') = \var x[\alpha\inv(\eta, t' / \ftrtm R \var x)] = \var x[\eps \circ L(\eta, t' / \ftrtm R \var x)]$. If $L$ is a CwF morphism, then this reduces further to $\eps(\ftrtm{LR} \var x[L\eta, \ftrtm{L}{t'}/\ftrtm{LR} \var x]) = \eps(\ftrtm{L}{t'})$. \qedhere
	\end{enumerate}
\end{proof}
\begin{corollary}
	We have naturality rules as for ordinary adjunctions:
	\begin{equation}
		\begin{array}{r c l c r c l}
			\alpha(\tau \circ \sigma \circ L\rho)
			&=& R\tau \circ \alpha(\sigma) \circ \rho
			&\qquad \qquad&
			\alpha\inv(R\tau \circ \sigma \circ \rho)
			&=& \tau \circ \alpha\inv(\sigma) \circ L\rho
			\\
			\alpha(t[\sigma][L\rho])
			&=& (\ftrtm R t)[\alpha(\sigma)][\rho]
			&&
			\alpha\inv((\ftrtm R t)[\sigma][\rho])
			&=& t[\alpha\inv(\sigma)][L\rho]
			\\
			\alpha(\nu(s[L \rho])) &=& (R\nu)(\alpha(s))[\rho]
			&&
			\alpha\inv((R\nu)s[\rho]) &=& \nu(\alpha\inv(s)[L \rho])
			\\
			\alpha(\nu(\mu(\ftrtm L r))) &=& (R\nu)(\alpha(\mu)(r))
			&&
			\alpha\inv((R\nu)(\mu(r))) &=& \nu(\alpha\inv(\mu)(\ftrtm L r))
		\end{array}
	\end{equation}
	where $\rho, \sigma, \tau$ denote substitutions, $s, t$ denote terms and $\mu, \nu$ denote natural transformations.
\end{corollary}
\begin{proof}
	Each equation on the right is equivalent to its counterpart on the left. The first equation on the left is old news. The other equations follow from $\alpha(t) = (\ftrtm R t)[\eta]$.
\end{proof}

\section{Presheaf models}\label{sec:psh-transform}
In this section, we study the implications of functors, natural transformations and adjunctions between categories $\catV$ and $\catW$ for the CwFs $\widehat{\catV}$ and $\widehat{\catW}$.

\subsection{Lifting functors}
A functor $F : \catV \to \catW$ gives rise to a functor $\fpsh F : \widehat{\catW} \to \widehat{\catV} : \Gamma \mapsto \Gamma \circ F$.
\begin{theorem}
	For any functor $F : \catV \to \catW$, the functor $\fpsh F$ is a morphism of CwFs.
\end{theorem}
\begin{proof}
Throughout the proof, it is useful to think of $\fpsh F$ as being right adjoint to $F$. For that reason we will again add an ignorable label $\fpshadj F$ for readability.
\begin{enumerate}
	\item A context $\Gamma \in \widehat{\catW}$ is mapped to the context $\fpsh F \Gamma = \Gamma \circ F \in \Psh(\catV)$. A substitution $\sigma : \Delta \to \Gamma$ is mapped to a substitution $\fpsh F \sigma : \fpsh F \Delta \to \fpsh F \Gamma$ by functoriality of composition. Unpacking this and adding labels, we get:
	\begin{itemize}
		\item $(\DSub{V}{\fpsh F \Gamma}) = \set{\fpshadj F(\gamma)}{\gamma : \DSub{FV}{\Gamma}}$,
		\item For $\vfi : \PSub{V'}{V}$ and $\gamma : \DSub{V}{\fpsh F \Gamma}$, we have $\fpshadj F(\gamma) \circ \vfi = \fpshadj F(\gamma \circ F \vfi) : \DSub{V'}{\fpsh F \Gamma}$.
		\item For $\sigma : \Delta \to \Gamma$, we get $\fpsh F \sigma \circ \fpshadj F(\delta) = \fpshadj F(\sigma \circ \delta)$.
	\end{itemize}
	
	\item We easily construct a functor $\int_{\catV}\fpsh F \Gamma \to \int_{\catW} \Gamma$ sending $(V, \gamma)$ to $(FV, \gamma)$. Precomposing with this functor, yields a map $\Ty(\Gamma) \to \Ty(\fpsh F \Gamma) : T \mapsto \fpsh F T$.
	Let us unpack and label this construction. Given $(\Gamma \sez_{\widehat{\catW}} T \type)$, the type $(\fpsh F \Gamma \sez_{\widehat{\catV}} \fpsh F T \type)$ is defined by:
	\begin{itemize}
		\item Terms $V \Dsez \fpshadj F(t) : (\fpsh F T) \dsub{\fpshadj F(\gamma)}$ are obtained by labelling terms $FV \Dsez_{\widehat{\catW}} t : T \dsub \gamma$.
		\item $\fpshadj F(t) \psub \vfi = \fpshadj F(t \psub{F \vfi})$.
	\end{itemize}
	This is natural in $\Gamma$, because
	\begin{align}
		\fpsh F(T \sub \sigma) \dsub{\fpshadj F(\gamma)} &\cong T \sub \sigma \dsub \gamma = T \dsub{\sigma \gamma} \cong (\fpsh F T) \dsub{\fpshadj F(\sigma \gamma)} = (\fpsh F T) \sub{\fpsh F \sigma} \dsub{\fpshadj F(\gamma)},
	\end{align}
	where the isomorphisms become equalities if we ignore labeling.
	
	\item Given $\Gamma \sez_{\widehat{\catW}} t : T$, we define $\fpsh F \Gamma \sez_{\widehat{\catV}} \ftrtm{\fpsh F}{t} : \fpsh F T$ by setting $(\ftrtm{\fpsh F}{t})\dsub{\fpshadj F(\gamma)}$ equal to $\fpshadj F(t \dsub \gamma)$. This is natural in the domain $V$ of $\fpshadj F(\gamma)$:
	\begin{equation}
		(\ftrtm{\fpsh F} t) \dsub{\fpshadj F(\gamma)} \psub \vfi
		= \fpshadj F(t \dsub \gamma) \psub{\vfi}
		= \fpshadj F(t \dsub{\gamma \circ F\vfi})
		= (\ftrtm{\fpsh F} t) \dsub{\fpshadj F(\gamma \circ F \vfi)}
		= (\ftrtm{\fpsh F} t) \dsub{\fpshadj F(\gamma) \circ \vfi}.
	\end{equation}
	It is also natural in $\Gamma$ by the same reasoning as for types.
	
	\item The terminal presheaf over $\catW$ is automatically mapped to the terminal presheaf over $\catV$.
	
	\item It is easily checked that comprehension, $\pi$, $\xi$ and $(\cwfpair \loch \loch)$ are preserved on the nose if we assume that $\fpshadj F(\gamma, t) = (\fpshadj F(\gamma), \fpshadj F(t))$, e.g. by ignoring labels.\qedhere
\end{enumerate}
\end{proof}
\begin{proposition}
	For any functor $F : \catV \to \catW$, the morphism of CwFs $\fpsh F$ preserves all operators related to $\Sigma$-types, $\Weld$-types and identity types.
\end{proposition}
\begin{proof}
	The defining terms of each of these types, are built from of other defining terms. This is in constrast with $\Pi$- and $\Glue$-types, where defining terms also contain non-defining terms, and the universe, whose defining terms even contain types. As $\fpsh F$ merely reshuffles defining terms, it respects each of these operations (ignoring labels).
\end{proof}

The reason we use a different notation for terms and for types ($\ftrtm{\fpsh F} t$ versus ${\fpsh F} T$) is to avoid confusion when it comes to encoding and decoding types: $\fpsh F$ acts very differently on types and on terms of the universe. To begin with, ${\fpsh F} \uni{}$ will typically not be the universe of $\widehat{\catV}$. Indeed, its primitive terms still originate as types over the primitive contexts of $\catW$. So if $\Gamma \sez_{\widehat{\catW}} A : \uni{}$, then $\ftrtm{\fpsh F} A$ is a representation in $\widehat{\catV}$ of a type from $\widehat{\catW}$. In contrast, $\fpsh F(\El\,A)$ is a type in $\widehat{\catV}$. Put differently still, when applied to an element of the universe, ${\fpsh F}$ reorganizes the presheaf structure of the universe. When applied to a type $T$, ${\fpsh F}$ reorganizes the presheaf structure of $T$.

\begin{example}\label{ex:presheaf-sets}
	Let $\pointcat$ be the terminal category with just a single object $()$ and only the identity morphism. It is easy to see that $\widehat{\pointcat} \cong \Set$. Meanwhile, let $\RGcat$ be the base category with objects $()$ and $(\var i : \IE)$ and maps non-freely generated by $() : (\var i : \IE) \to ()$ and $(0/\var i), (1/\var i) : () \to (\var i : \IE)$, such that $\widehat \RGcat$ is the category of reflexive graphs. We have a functor $F : \pointcat \to \RGcat$ sending $()$ to $()$. This constitutes a functor $\fpsh F : \widehat \RGcat \to \widehat \pointcat$ sending a reflexive graph to its set of nodes. It is not surprising that this functor is sufficiently well-behaved to be a CwF morphism.
	
	This example also illustrates that the universe is not preserved. The nodes of the universe in $\widehat \RGcat$ are reflexive graphs. Thus $\fpsh F \uni{}$ will be the set of reflexive graphs. Then if $\Gamma \sez_{\widehat \RGcat} A : \uni{}$ in $\widehat \RGcat$, then $\fpsh F \Gamma \sez_{\widehat \pointcat} \ftrtm{\fpsh F}{A} : \fpsh F \uni{}$ sends nodes of $\Gamma$ to reflexive graphs. The edges \emph{between} types are forgotten. However, the type $\fpsh F \Gamma \sez_{\widehat \pointcat} \fpsh F (\El\,A) \type$ is simply the dependent set of nodes of $\El\,A$. Here, the edges \emph{within} $\El\,A$ are also forgotten.
	
	In general, from $\var X : \uni{} \sez \El\,\var X \type$, we can deduce $\ftrtm{\fpsh F}{\var X} : \fpsh F \uni{} \sez \fpsh F(\El\,\var X) \type$.
\end{example}

\subsection{Lifting natural transformations}
Assume we have functors $F, G : \catV \to \catW$ and a natural transformation $\nu : F \to G$. Then we get morphisms of CwFs $\fpsh F, \fpsh G : \Psh(\catW) \to \Psh(\catV)$ and, since presheaves are contravariant, a natural transformation $\fpsh \nu : \fpsh G \to \fpsh F$.

Let us see how $\fpsh \nu$ works. Pick $\fpshadj G(\gamma) : \DSub{V}{\fpsh G \Gamma}$. Then $\gamma : \DSub{GV}{\Gamma}$ and $\gamma \circ \nu : \DSub{FV}{\Gamma}$. Now $\fpsh \nu \circ \fpshadj G(\gamma) = \fpshadj F(\gamma \circ \nu)$. This is natural because
\begin{equation}
	\fpsh \nu \circ (\fpshadj G(\gamma) \circ \vfi)
	= \fpsh \nu \circ \fpshadj G(\gamma \circ G\vfi)
	= \fpshadj F(\gamma \circ G\vfi \circ \nu)
	= \fpshadj F(\gamma \circ \nu \circ F\vfi)
	= \fpshadj F(\gamma \circ \nu) \circ \vfi
	= (\fpsh \nu \circ \fpshadj F(\gamma)) \circ \vfi.
\end{equation}

\subsection{Lifting adjunctions}
\begin{proposition}\label{thm:lifting-adjunctions}
	Assume we have functors $L : \catV \to \catW$ and $R : \catW \to \catV$ such that $\alpha : L \dashv R$ with unit $\eta : \Id \to RL$ and co-unit $\eps : LR \to \Id$. Then we get $\fpsh L : \widehat \catW \to \widehat \catV$ and $\fpsh R : \widehat \catV \to \widehat \catW$ and we have $\fpsh \alpha : \fpsh L \dashv \fpsh R$ with unit $\fpsh \eps : \Id \to \fpsh R \fpsh L$ and co-unit $\fpsh \eta : \fpsh L \fpsh R \to \Id$. Moreover, $\fpsh L \circ \yoneda \cong \yoneda \circ R : \catW \to \widehat \catV$.
\end{proposition}
So $\fpsh \loch$ is an operation that takes the right adjoint of a functor and immediately extends it to the entire presheaf category.
\begin{proof}
	To prove the adjunction, it is sufficient to prove that $\fpsh \eta \circ \fpsh L \fpsh \eps = \id_{\fpsh L}$ and $\fpsh R \fpsh \eta \circ \fpsh \eps = \id_{\fpsh R}$. We need to assume that $\fpshadj{\Id} = \id$ and $\fpshadj{FG} = \fpshadj{G} \circ \fpshadj{F}$.
	
	Pick $\fpshadj L(\gamma) : \DSub{V}{\fpsh L \Gamma}$. Then
	\begin{align}
		\fpsh \eta \circ \fpsh L \fpsh \eps \circ \fpshadj L(\gamma)
		&= \fpsh \eta \circ \fpshadj L(\fpsh \eps \circ \gamma)
		= \fpsh \eta \circ \fpshadj L(\fpshadj{R} (\fpshadj{L}(\gamma \circ \eps))) \\ \nn
		&= \fpshadj{L}(\gamma \circ \eps) \circ \eta
		= \fpshadj{L}(\gamma \circ \eps \circ L\eta)
		= \fpshadj{L}(\gamma).
	\end{align}
	
	Similarly, pick $\fpshadj R(\delta) : \DSub{W}{\fpsh R \Gamma}$. Then
	\begin{align}
		\fpsh R \fpsh \eta \circ \fpsh \eps \circ \fpshadj R(\delta)
		&= \fpsh R \fpsh \eta \circ \fpshadj R(\fpshadj L(\fpshadj R(\delta) \circ \eps))
		= \fpsh R \fpsh \eta \circ \fpshadj R(\fpshadj L(\fpshadj R(\delta \circ R\eps))) \\ \nn
		&= \fpshadj R(\fpsh \eta \circ \fpshadj L(\fpshadj R(\delta \circ R\eps)))
		= \fpshadj R(\delta \circ R \eps \circ \eta) = \fpshadj R(\delta).
	\end{align}
	To see the isomorphism:
	\begin{equation}
		(\DSub V {\fpsh L \yoneda W})
		\cong (\DSub{LV}{\yoneda W})
		= (\PSub{LV}{W})
		\cong (\PSub{V}{RW})
		= (\DSub{V}{\yoneda RW}).
	\end{equation}
	This is clearly natural.
\end{proof}

\subsection{The left adjoint to a lifted functor}\label{sec:left-adjoint-to-lifted-functor}
Assume a functor $F : \catV \to \catW$. Under reasonable conditions, one finds a general construction of a \emph{functor} $\lpsh F : \widehat \catV \to \widehat \catW$ that is left adjoint to $\fpsh F : \widehat \catW \to \widehat \catV$ \cite[00VC]{stacks-project17}. Although we will need such a left adjoint at some point, the general construction is overly complicated for our needs. Therefore, we will construct that functor ad hoc when we need it.
For now, we simply assume its existence and prove a lemma and some bad news:
\begin{lemma}\label{thm:yoneda-and-left-adjoint}
	Suppose we have a functor $\lpsh F : \widehat \catV \to \widehat \catW$ and a functor $F : \catV \to \catW$, such that $\lpsh F \dashv \fpsh F$. Then $\lpsh F \circ \yoneda \cong \yoneda \circ F$.
\end{lemma}
\begin{proof}
	We have a chain of isomorphisms, natural in $W$ and $\Gamma$:
	\begin{equation}
		(\lpsh F \yoneda W \to \Gamma) \cong (\yoneda W \to \fpsh F \Gamma) \cong (\DSub{F W}{\Gamma}) \cong (\yoneda F W \to \Gamma).
	\end{equation}
	Call the composite of these isomorphisms $f$. Then we have $f(\id_{\lpsh F \yoneda}) : \yoneda F \to \lpsh F \yoneda$ and $f\inv(\id_{\yoneda F}) : \lpsh F \yoneda \to \yoneda F$. By naturality in $\Gamma$, we have:
	\begin{equation}
		f(\id) \circ f\inv(\id) = f\inv(f(\id) \circ \id) = f\inv(f(\id)) = \id,
	\end{equation}
	\begin{equation}
		f\inv(\id) \circ f(\id) = f(f\inv(\id) \circ \id) = f(f\inv(\id)) = \id. \qedhere
	\end{equation}
\end{proof}
\begin{proposition}\label{thm:left-adjoint-not-cwf}
	The functor $\lpsh F : \widehat \catV \to \widehat \catW$ is not in general a morphism of CwFs.
\end{proposition}
\begin{proof}
	Consider the only functor $G : \RGcat \to \pointcat$; it sends $()$ and $(\var i : \IE)$ to $()$. Then $\fpsh G : \widehat \pointcat \to \widehat \RGcat$ sends a set $S$ to the discrete reflexive graph on $S$.
	
	Its left adjoint $\lpsh G$ sends a reflexive graph $\Gamma$ to its colimit. That is, it identifies all edge-connected nodes. Now consider a type $\Gamma \sez_{\widehat \RGcat} T \type$:
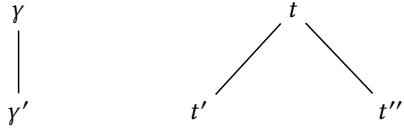
	\begin{equation}
		\xymatrix{
			\gamma \ar@{-}[d] &&& t \ar@{-}[dl] \ar@{-}[dr]
			\\
			\gamma' && t' && t''
		}
	\end{equation}
	That is: $\Gamma$ contains two nodes and an edge connecting them (as well as constant edges). $T$ has one node above $\gamma$ and two nodes above $\gamma'$ and they are connected as shown.
	
	There are two substitutions from $\yoneda O$, namely $\gamma$ and $\gamma' : \yoneda O \to \Gamma$. Clearly, $\lpsh G(T[\gamma]) \neq \lpsh G(T[\gamma'])$. But $\lpsh G (\gamma) = \lpsh G (\gamma')$. Thus, $\lpsh G$ cannot preserve substitution in the sense that $\lpsh G(T[\sigma])$ is equal to $\lpsh GT[\lpsh G \sigma]$.
\end{proof}

\chapter{Bridge/path cubical sets}\label{ch:bpcubecat}
In this chapter, we move from the general presheaf model to the category $\widehat{\cubecat}$ of cubical sets and the category $\widehat{\bpcubecat}$ of bridge/path cubical sets. Throughout, we assume the existence of an infinite alphabet $\aleph$ of variable names, as well as a function $\fresh : \Pi(A \subseteq \aleph).(\aleph \setminus A)$ that picks a fresh variable for a given set of variables $A$.

\section{The category of cubes}\label{sec:cubecat}
In this section, we define the category of cubes $\cubecat$; presheaves over this category are called \textdef{cubical sets}. The reason we are interested in cubical sets, is that they generalize reflexive graphs: they contain points and edges, but also edges between edges (squares), edges between squares (cubes), etc. Imagine we have a type $\IE$ that contains two points, connected by an edge (as opposed to a bridge or a path). Then an $n$-dimensional cube is a value that ranges over $n$ variables of type $\IE$. For this reason, we define a \textdef{cube} as a finite set $W \subseteq \aleph$. We write $()$ for the 0-dimensional cube, $(W, \ctxedge{\var i})$ for $W \uplus \accol{\var i}$, implying that $\var i \not\in W$, and similarly $(V, W)$ for $V \uplus W$.

A face map $\vfi : \PSub V W$ assigns to every variable $\var i \in W$ a value $\var i \psub \vfi$ which is either 0, 1 (up to $n - 1$ for $n$-ary parametricity), or a variable in $V$. We use common substitution notation to denote face maps. So $\var i \psub{()} = \var i$, while $\var i \psub{\vfi, t / \var i} = t$. We write $\vfi = (\psi, f(\var i)/\var i \in W') : \PSub V W$ to denote that $\var i \psub \vfi = f(\var i)$ for all $\var i \in W' \subseteq W$. To emphasize that a face map does not use a variable $\var i$, we write $\vfi = (\psi, \facewkn{\var i})$. We make sure that different clauses in the same substitution do not conflict; hence we need no precedence rules.

Then a cubical set $\Gamma$ contains, for every cube $W$, a set of $|W|$-dimensional cubes $\DSub W \Gamma$. Every such cube has $2^{|W|}$ vertices, extractable using the face maps $(\epsilon(\var i)/\var i \in M) : \PSub \eset W$ for all $\epsilon : W \to \accol{0, 1}$. By substituting 0 or 1 for only some variables, we obtain the sides of a cube. We can create flat (\textdef{degenerate}/\textdef{constant}) cubes by not using variables, e.g. $(\facewkn{\var i}) : \PSub {(M, \ctxedge{\var i})} M$. Cubes also have diagonals, e.g. $(\var i/\var j) : \PSub{(\ctxedge{\var i})}{(\ctxedge{\var i, \var j})}$.

It is easy to see that $\cubecat$ is closed under finite cartesian products; namely $V \times W = (V, W)$. Because the Yoneda-embedding preserves limits such as cartesian products, we have $\yoneda W \cong (\yoneda \IE)^W$.

This is not the only useful definition. For example, \cite{moulin-param3} and \cite{huber} consider cubes which have no diagonals, while \cite{cubical} uses cubes that have `connections', a way of adding a dimension by folding open a line to become a square with two adjacent constant sides.

\section{The category of bridge/path cubes}
The novel category of bridge/path cubes $\bpcubecat$ is similar to $\cubecat$ but its cubes have two flavours of dimensions: bridge dimensions and path dimensions. So a \textdef{bridge/path cube} $W$ is a pair $(W_\IB, W_\IP)$, where $W_\IB$ and $W_\IP$ are disjoint subsets of $\aleph$. We write $()$ for the 0-dimensional cube, $(W, \ctxbrid{\var i})$ for $(W_\IB \uplus \accol{\var i}, W_\IP)$ and $(W, \ctxpath{\var i})$ for $(W_\IB, W_\IP \uplus \accol{\var i})$.

A face map $\vfi : \PSub V W$ assigns to every bridge variable $(\ctxbrid{\var i}) \in W$ either 0, 1 or a bridge variable from $V$, and to every path variable $(\ctxpath{\var i}) \in W$ either 0, 1, or a path or bridge variable from $V$. We will sometimes add a superscript to make the status of a variable clear, e.g. $(\var j^\IB / \var i^\IP) : \PSub{(\var j : \IB)}{(\var i : \IP)}$.

Then a bridge/path cubical set $\Gamma$ contains, for every bridge/path cube $W$, a set of cubes with $|W_\IB|$ bridge dimensions and $|W_\IP|$ path dimensons. Again we can extract vertices and faces and we can create flat cubes by introducing bridge or path dimensions. We can weaken paths to bridges, e.g. $(\var j/\var i) : \PSub{(W, \ctxbrid{\var j})}{(W, \ctxpath{\var i})}$. Finally, we can extract diagonals, but if the cube of which we take the diagonal, has at least one bridge dimension, then the diagonal has to be a bridge, e.g. $(\var j^\IB/\var i^\IP) : \PSub{(W, \ctxbrid{\var j})}{(W, \ctxpath{\var i}, \ctxbrid{\var j})}$.

\section{The cohesive structure of $\widehat{\bpcubecat}$ over $\widehat{\cubecat}$}
In this section, we construct a chain of five adjoint functors (all but the leftmost one morphisms of CwFs) between $\widehat{\bpcubecat}$ and $\widehat{\cubecat}$. By composing each one with its adjoint, these give rise to a chain of four adjoint endofunctos (all but the leftmost one endomorphisms of CwFs) on $\widehat{\bpcubecat}$.

We are not interested in the adjoint quintuple between $\widehat{\bpcubecat}$ and $\widehat{\cubecat}$ per se, but making a detour along them reveals a structure similar to what is studied in cohesive type theory, and may also be beneficial in order to understand intuitively what is going on in the rest of this text.

\subsection{Cohesion}
Let $\cat S$ be a category whose objects are some notion of spaces. Then a notion of \emph{cohesion} on objects of $\cat S$ gives rise to a category $\cat C$ of cohesive spaces and a \textbf{forgetful functor} $\cohfget : \cat C \to \cat S$ which maps a cohesive space $C$ to the underlying space $UC$, forgetting its cohesive structure.

Typically, if $\cohfget : \cat C \to \cat S$ appeals to the intuition about cohesion, then it is part of an adjoint quadruple of functors\footnote{More often, these are denoted $\Pi \dashv \Delta \dashv U \dashv \nabla$, but some of these symbols are already heavily in use in type theory.}
\begin{equation}
	\cohpi \dashv \cohdisc \dashv \cohfget \dashv \cohcodisc.
\end{equation}
Here, the \textbf{discrete functor} $\cohdisc$ equips a space $S$ with a discrete cohesive structure, i.e. in the cohesive space $\cohdisc S$, nothing is stuck together. As such, a cohesive map $\cohdisc S \to C$ amounts to a map $S \to UC$.

Dually, the \textbf{codiscrete functor} $\cohcodisc$ equips a space $S$ with a codiscrete cohesive structure, sticking everything together. As such, a cohesive map $C \to \cohcodisc S$ amounts to a map $UC \to S$.

Finally, the functor $\cohpi$ maps a cohesive space $C$ to its space of cohesively connected components. A map $C \to \cohdisc S$ will necessarily be constant on cohesive components, as $\cohdisc S$ is discrete, and hence amounts to a map $\cohpi C \to S$.

Typically, the composites $\cohpi \cohdisc, \cohfget \cohdisc, \cohfget \cohcodisc : \cat S \to \cat S$ will be isomorphic to the identity functor, the latter two even equal. Indeed: if we equip a space with discrete cohesion and then contract components, we have done essentially nothing. If we equip a space with a discrete or codiscrete cohesion, and then forget it again, we have litterally done nothing.

The composites $\shp = \cohdisc \cohpi$, $\flat = \cohdisc \cohfget$ and $\sharp = \cohcodisc \cohfget : \cat C \to \cat C$ form a more interesting adjoint triple $\shp \dashv \flat \dashv \sharp$ of endofunctors on $\cat C$. The \textbf{shape} functor $\shp$ contracts cohesive components. The \textbf{flat} functor $\flat$ removes the existing cohesion in favour of the discrete one, and the \textbf{sharp} functor $\sharp$ removes it in favour of the codiscrete one. If the adjoint triple on $\cat S$ is indeed as described above, then we can show
\begin{equation}
	\begin{array}{c c c c}
		\shp \shp \cong \shp &
		\shp \flat \cong \flat &&\\
		\flat \shp = \shp &
		\flat \flat = \flat &
		\flat \sharp = \flat &\\&
		\sharp \flat = \sharp &
		\sharp \sharp = \sharp,
	\end{array}
\end{equation}
Moreover, $\shp \dashv \flat$ will have (essentially) the same co-unit $\kappa : \flat \to \Id$ as $\cohdisc \dashv \cohfget$ and the same unit $\varsigma : \Id \to \shp$ as $\cohpi \dashv \cohdisc$. The adjunction $\flat \dashv \sharp$ will have the same unit $\kappa : \flat \to \Id$ as $\cohdisc \dashv \cohfget$ and the same co-unit $\iota : \Id \to \sharp$ as $\cohfget \dashv \cohcodisc$. For more information, see e.g. \cite{adjoint-logic}.

\begin{example}
	Let $\cat C = \catTop$ (the category of topological spaces) and $\cat S = \Set$. Then $\cohfget : \catTop \to \Set$ maps a topological space $(X, \tau)$ to the underlying set $X$. The discrete functor $\cohdisc$ equips a set $X$ with its discrete topology $2^X$ and $\cohcodisc$ equips it with the codiscrete topology $(\eset, X)$. Finally, $\cohpi$ maps a topological space to its set of connected components.
\end{example}
\begin{example}
	Let $\cat C = \Cat$ and $\cat S = \Set$. Then we can take $\cohfget \cat A = \Obj~\cat A$, make $\cohdisc X$ the discrete category on $X$ with only identity morphisms, $\cohcodisc X$ the codiscrete category on $X$ where every Hom-set is a singleton, and $\cohpi \cat A$ the set of zigzag-connected components of $\cat A$, i.e. $\cohpi \cat A = \Obj(\cat A)/\Hom$.
\end{example}
\begin{example}
	Let $\cat C = \widehat \RGcat$, the category of reflexive graphs, and $\cat S = \Set$. Then we can let $\cohfget$ map a reflexive graph $\Gamma$ to its set of nodes $\cohfget \Gamma = (\DSub{()} \Gamma)$; $\cohdisc$ will map a set $X$ to the discrete reflexive graph with only constant edges (i.e. $(\DSub{()}{\cohdisc X}) = (\DSub{(\var i : \IE)}{\cohdisc X}) = X$), and $\cohcodisc$ maps a set $X$ to the codiscrete reflexive graph with a unique edge between any two points (i.e. $(\DSub{()}{\cohcodisc X}) = X$ and $(\DSub{(\var i : \IE)}{\cohcodisc X}) = X \times X$). Finally, $\cohpi$ maps a graph $\Gamma$ to its set $\cohpi \Gamma$ of edge-connected components.
\end{example}
This last example is interesting, because we know that $\cohfget : \widehat{\RGcat} \to \widehat{\pointcat} \cong \Set$ is a morphism of CwFs, arising as $\cohfget = \fpsh F$ with $F : \pointcat \to \RGcat$ the unique functor mapping $()$ to $()$ (see \cref{ex:presheaf-sets}). The functor $\cohdisc$ is also a morphism of CwFs, arising as $\cohdisc = \fpsh G$, with $G$ the unique functor $\RGcat \to \pointcat$.

\subsection{The cohesive structure of $\widehat{\bpcubecat}$ over $\widehat{\cubecat}$, intuitively}
In the remainder of this section, we establish a chain of no less than five adjoint functors between $\widehat{\bpcubecat}$ and $\widehat{\cubecat}$. The reason we have more than in other situations, is that bridge/path cubical sets can be seen as cohesive cubical sets in two ways: we can either view cubical sets as bridge-only cubical sets, in which case the forgetful functor $\cohfget : \widehat{\bpcubecat} \to \widehat{\cubecat}$ forgets the cohesive structure given by the paths; or we can view cubical sets as path-only cubical sets, in which case the forgetful functor $\cohpaths : \widehat{\bpcubecat} \to \widehat{\cubecat}$ forgets the cohesion given by the bridges.

The chain of functors we obtain is the following:
\begin{equation}
	\cohpi \dashv \cohdisc \dashv \cohfget \dashv \cohcodisc \dashv \cohpaths, \qquad
	\cohpi, \cohfget, \cohpaths : \widehat{\bpcubecat} \to \widehat{\cubecat}, \qquad
	\cohdisc, \cohcodisc : \widehat{\cubecat} \to \widehat{\bpcubecat}.
\end{equation}
and it can likely be extended by a sixth functor on the right, that we currently have no use for. The (cohesion-as-paths) forgetful functor $\cohfget$ maps a bridge/path cubical set $\Gamma$ to the cubical set $\cohfget \Gamma$ made up of its bridges, forgetting which bridges are in fact paths. The (cohesion-as-paths) discrete functor $\cohdisc$ introduces a discrete path relation, the bridges of $\cohdisc \Gamma$ are the edges of $\Gamma$, whereas the paths of $\cohdisc \Gamma$ are all constant. The (cohesion-as-paths) codiscrete functor $\cohcodisc$ introduces a path relation which is codiscrete in the sense that there are as many paths as there can be: every bridge is a path. So a bridge in $\cohcodisc \Gamma$ is the same as an edge in $\Gamma$, and a path in $\cohcodisc \Gamma$ is also the same as an edge in $\Gamma$. Note that $\cohcodisc$ is also the cohesion-as-bridges discrete functor: viewing $\Gamma$ as a path-only cubical set, it equips $\cohcodisc \Gamma$ with the fewest bridges possible: only when there is a path, there will also be a bridge. The paths functor $\cohpaths$, which is the cohesion-as-bridges forgetful functor, maps a bridge/path cubical set $\Gamma$ to its cubical set of paths $\cohpaths \Gamma$. Finally, $\cohpi$ divides out a bridge/path cubical set by its path relation, obtaining a bridge-only cubical set.

\subsection{The cohesive structure of $\bpcubecat$ over $\cubecat$}
We saw in \cref{thm:lifting-adjunctions} that if we have adjoint functors $L \dashv R$ on the base categories, then we obtain functors $\fpsh L \dashv \fpsh R$ on the presheaf categories and moreover $\fpsh L$ extends $R$. So $\fpsh \loch$ takes the right adjoint of a functor and at the same time extends it from the category of primitive contexts, to the entire presheaf category. In this sense, it is a good idea to start by defining the functors
\begin{equation}
	\cohpi \dashv \cohdisc \dashv \cohfget \dashv \cohcodisc, \qquad
	\cohpi, \cohfget : \bpcubecat \to \cubecat, \qquad
	\cohdisc, \cohcodisc : \cubecat \to \bpcubecat,
\end{equation}
on the base categories. We define these functors as in \cref{fig:cohesion-base}.
This may not be entirely intuitive. The key here is that every path dimension can be weakened to a bridge dimension. Thus, two adjacent vertices of a bridge/path cube are always connected by a bridge, and only connected by a path if they are adjacent along a path dimension. The $\cohpi$ functor leaves bridges alone (converting them to edges), but contracts paths. The $\cohdisc$ functor turns edges into bridges, but does not produce paths. The $\cohfget$ functor keeps bridges (converting them to edges) and forgets paths, but remembers the bridges they weaken to. The $\cohcodisc$ functor turns edges into paths, which are also bridges.
\begin{figure}
	\begin{equation*}
		\begin{array}{c || c c c | c c c c}
			& () & (W, \ctxbrid{\var i}) & (W, \ctxpath{\var i})
			& () & (\vfi, \var j^\IB /\var i^\IB) & (\vfi, \var j^\IB /\var i^\IP) & (\vfi, \var j^\IP /\var i^\IP) \\
			\hline \hline
			\cohpi
			& () & (\cohpi W, \ctxedge{\var i}) & \cohpi W 
			& () & (\cohpi \vfi, \var j^\IE / \var i^\IE) & \cohpi \vfi & \cohpi \vfi \\
			\cohfget
			& () & (\cohfget W, \ctxedge{\var i}) & (\cohfget W, \ctxedge{\var i})
			& () & (\cohfget \vfi, \var j^\IE / \var i^\IE) & (\cohfget \vfi, \var j^\IE / \var i^\IE) & (\cohfget \vfi, \var j^\IE / \var i^\IE) \\
			\hline
			\shp
			& () & (\shp W, \ctxbrid{\var i}) & \shp W
			& () & (\shp \vfi, \var j^\IB / \var i^\IB) & \shp \vfi & \shp \vfi \\
			\flat
			& () & (\flat W, \ctxbrid{\var i}) & (\flat W, \ctxbrid{\var i})
			& () & (\flat \vfi, \var j^\IB/\var i^\IB) & (\flat \vfi, \var j^\IB/\var i^\IB) & (\flat \vfi, \var j^\IB/\var i^\IB) \\
			\sharp
			& () & (\sharp W, \ctxpath{\var i}) & (\sharp W, \ctxpath{\var i})
			& () & (\sharp \vfi, \var j^\IP / \var i^\IP) & (\sharp \vfi, \var j^\IP / \var i^\IP) & (\sharp \vfi, \var j^\IP / \var i^\IP) \\
			\hline
			\varsigma : \Id \to \shp
			& () & (\varsigma_W, \var i^\IB / \var i^\IB) & (\varsigma_W, \facewkn{\var i^\IP})\\
			\kappa : \flat \to \Id
			& () & (\kappa_W, \var i^\IB / \var i^\IB) & (\kappa_W, \var i^\IB / \var i^\IP) &&& \circlearrowright \\
			\iota : \Id \to \sharp
			& () & (\iota_W, \var i^\IB / \var i^\IP) & (\iota_W, \var i^\IP / \var i^\IP)
		\end{array}
	\end{equation*}
	\begin{equation*}
		\begin{array}{c || c c | c c}
			& () & (W, \ctxedge{\var i})
			& () & (\vfi, \var j^\IE / \var i^\IE) \\
			\hline \hline
			\cohdisc & () & (\cohdisc W, \ctxbrid{\var i}) 
			& () & (\cohdisc \vfi, \var j^\IB / \var i^\IB) \\
			\cohcodisc & () & (\cohcodisc W, \ctxpath{\var i}) 
			& () & (\cohcodisc \vfi, \var j^\IP / \var i^\IP)
		\end{array}
	\end{equation*}
	\caption{The cohesive structure of $\bpcubecat$ over $\cubecat$.}
	\label{fig:cohesion-base}
\end{figure}

Note also that $\cohpaths$ cannot be defined this way as it does not map all primitive contexts to primitive contexts. For example, $(\ctxbrid{\var i})$ consists of two points connected by a bridge. The only paths are constant. Hence, forgetting the bridge structure yields two loose points, which together do not form a cube.

\begin{lemma}\label{thm:uniqueness-of-nattrans-base}
	Let $\catV, \catW \in \accol{\cubecat, \bpcubecat}$ and let $F, G : \catV \to \catW$ be composites of the functors $\cohpi$, $\cohdisc$, $\cohfget$ and $\cohcodisc$. Then all natural transformations $F \to G$ are equal.
\end{lemma}
\begin{proof}
	Let $\nu : F \to G$ be a natural transformation. We show that $\nu$ is completely determined. Since $F$ and $G$ preserve products, $\nu$ is determined by its action on single-variable contexts. Now if $\Gamma$ has a single variable $\var i$, then either $G\Gamma = ()$ in which case $\nu_\Gamma : F \Gamma \to ()$ is determined, or $G\Gamma$ also contains $\var i$ as its only variable. Now $\var i \psub{\nu_\Gamma} \neq 0$ because then the following diagram could not commute:
\begin{equation}
	\xymatrix{
		() \ar[r]^{\nu_{()}} \ar[d]_{F(1/\var i)}
		& () \ar[d]^{G(1/\var i) = (1 / \var i)}
		\\
		F\Gamma \ar[r]^{\nu_\Gamma}
		& G\Gamma
	}
\end{equation}
Similarly, $\var i \psub{\nu_\Gamma} \neq 1$. Then $F\Gamma$ must contain a variable, implying that it contains only the variable $\var i$ and $\var i \psub{\nu_\Gamma} = \var i$.
\end{proof}

\begin{proposition}[The cohesive structure of $\bpcubecat$ over $\cubecat$]\label{thm:cohesion-base}
	These four functors are adjoint: $\cohpi \dashv \cohdisc \dashv \cohfget \dashv \cohcodisc$. On the $\cubecat$-side, we have
	\begin{equation}
		\cohpi\cohdisc = \Id \quad \dashv \quad
		\cohfget\cohdisc = \Id \quad \dashv \quad
		\cohfget\cohcodisc = \Id \quad 
		: \quad \cubecat \to \cubecat.
	\end{equation}
	On the $\bpcubecat$-side, we write
	\begin{equation}
		\shp := \cohdisc\cohpi \quad \dashv \quad
		\flat := \cohdisc\cohfget \quad \dashv \quad
		\sharp := \cohcodisc\cohfget \quad
		: \quad \bpcubecat \to \bpcubecat.
	\end{equation}
	By consequence, we have
	\begin{equation}
		\begin{array}{c c c c}
			\shp \shp = \shp &
			\shp \flat = \flat &&\\
			\flat \shp = \shp &
			\flat \flat = \flat &
			\flat \sharp = \flat &\\&
			\sharp \flat = \sharp &
			\sharp \sharp = \sharp.
		\end{array}
	\end{equation}
	The following table lists the units and co-units of all adjunctions involved:
	\begin{equation}
		\begin{array}{c || c | c | c || c | c}
			& \cohpi \dashv \cohdisc & \cohdisc \dashv \cohfget & \cohfget \dashv \cohcodisc & \shp \dashv \flat & \flat \dashv \sharp \\
			\hline
			\text{unit} & \varsigma : \Id \to \shp & \id : \Id \to \Id & \iota : \Id \to \sharp & \varsigma : \Id \to \shp & \iota : \Id \to \sharp \\
			\hline
			\text{co-unit} & \id : \Id \to \Id & \kappa : \flat \to \Id & \id : \Id \to \Id & \kappa : \flat \to \Id & \kappa : \flat \to \Id
		\end{array}
	\end{equation}
	The functors $\shp$, $\flat$ and $\sharp$ and the natural transformations $\varsigma$, $\kappa$ and $\iota$ are given in \cref{fig:cohesion-base}. Finally, the following natural transformations are all the identity:
	\begin{equation}\label{eq:cohesion-base-identities}
		\begin{array}{c c c c | c c c}
			\cohpi \to \cohpi & \cohdisc \to \cohdisc & \cohfget \to \cohfget & \cohcodisc \to \cohcodisc & \shp \to \shp & \flat \to \flat & \sharp \to \sharp \\ \hline
			\cohpi \varsigma & \varsigma \cohdisc &&&
			\shp \varsigma = \varsigma \shp & \varsigma \flat \qquad{} \\
			& \kappa \cohdisc & \cohfget \kappa &&
			{} \qquad \kappa \shp & \kappa \flat = \flat \kappa & \sharp \kappa \qquad{} \\
			&& \cohfget \iota & \iota \cohcodisc &
			& {}\qquad \flat \iota & \sharp \iota = \iota \sharp
		\end{array}
	\end{equation}
\end{proposition}
\begin{proof}
	The equalities are immediate from \cref{thm:uniqueness-of-nattrans-base}.
\end{proof}
\begin{proof}
	Each of the transformations in \cref{eq:cohesion-base-identities} is easily seen to be the identity by inspecting the definitions in \cref{fig:cohesion-base}. In order to prove that $L \dashv R$ with unit $\eta$ and co-unit $\eps$, it suffices to check that $\eps L \circ L \eta = \id : L \to L$ and $R \eps \circ \eta R = \id : R \to R$, which also follows from \cref{thm:uniqueness-of-nattrans-base}. In other words, the mere existence of well-typed candidates for the unit and co-unit is sufficient to conclude adjointness.
\end{proof}

\subsection{The cohesive structure of $\widehat{\bpcubecat}$ over $\widehat{\cubecat}$, formally}
We now define functors and natural transformations of the same notation by
\begin{equation}
	\cohdisc := \fpsh \cohpi, \quad
	\cohcodisc := \fpsh \cohfget \quad
	: \widehat{\cubecat} \to \widehat{\bpcubecat},
\end{equation}
\begin{equation}
	\cohfget := \fpsh \cohdisc, \quad
	\cohpaths := \fpsh \cohcodisc \quad
	: \widehat{\bpcubecat} \to \widehat{\cubecat},
\end{equation}
\begin{equation}
	\flat := \fpsh \shp = \cohdisc \cohfget, \quad
	\sharp := \fpsh \flat = \cohcodisc \cohfget, \quad
	\coshp := \fpsh \sharp = \cohcodisc \cohpaths \quad : \widehat{\bpcubecat} \to \widehat{\bpcubecat}.
\end{equation}
\begin{equation}
	\kappa := \fpsh \varsigma : \flat \to \Id, \qquad
	\iota := \fpsh \kappa : \Id \to \sharp, \qquad
	\vartheta := \fpsh \iota : \coshp \to \Id.
\end{equation}
From \cref{sec:left-adjoint-to-lifted-functor}, we know that there is a further left adjoint $\cohpi \dashv \cohdisc$, which we should not expect to be a morphism of CwFs. Indeed, the proof of \cref{thm:left-adjoint-not-cwf} is easily adapted to show the contrary. We postpone its construction to \cref{sec:def-cohpi}; however, by \cref{thm:yoneda-and-left-adjoint} it will satisfy the property that $\cohpi \circ \yoneda \cong \yoneda \circ \cohpi$, which we can use to characterize its behaviour. We take a moment to see how each of these functors behaves:
\begin{description}
	\item[$\cohpi$] We know that $W$-cubes $\gamma : \DSub W \Gamma$ correspond to substitutions $\gamma : \yoneda W \to \Gamma$ and hence give rise to substitutions $\cohpi \gamma : \cohpi(\yoneda W) \to \cohpi \Gamma$, which in turn correspond to $\cohpi W$-cubes $\DSub{\cohpi W}{\cohpi \Gamma}$. So a bridge $\DSub{(\ctxbrid{\var i})}{\Gamma}$ is turned into an edge $\DSub{(\ctxedge{\var i})}{\cohpi \Gamma}$, whereas a path $\DSub{(\ctxpath{\var i})}{\Gamma}$ is contracted to a point $\DSub{()}{\cohpi \Gamma}$. Simply put, $\cohpi$ contracts paths to points.
	\item[$\cohdisc$] A $W$-cube $\fpshadj{\cohpi}(\gamma) : \DSub{W}{\cohdisc \Gamma}$ is a $\cohpi W$-cube $\gamma : \DSub{\cohpi W}{\Gamma}$. So a bridge $\DSub{(\ctxbrid{\var i})}{\cohdisc \Gamma}$ is the same as an edge $\DSub{(\ctxedge{\var i})}{\Gamma}$ and a path $\DSub{(\ctxpath{\var i})}{\cohdisc \Gamma}$ is the same as a point $\DSub{()}{\Gamma}$, which in turn is the same as a point $\DSub{()}{\cohdisc \Gamma}$, showing that there are only constant paths.
	
	Viewed differently, using that $\cohdisc \circ \yoneda = \yoneda \circ \cohdisc$ by \cref{thm:lifting-adjunctions}, we can say that an edge $\DSub{(\ctxedge{\var i})}{\Gamma}$ gives rise to a bridge $\DSub{(\ctxbrid{\var i})}{\cohdisc \Gamma}$, while there is nothing that gives rise to (non-trivial) paths.
	\item[$\cohfget$] A $W$-cube $\fpshadj \cohdisc(\gamma) : \DSub{W}{\cohfget \Gamma}$ is the same as a brdige $\gamma : \DSub{\cohdisc W}{\Gamma}$. So an edge in $\cohfget \Gamma$ is a bridge in $\Gamma$. Alternatively, we can say that any bridge and any path in $\Gamma$ gives rise to an edge in $\cohfget \Gamma$.
	\item[$\cohcodisc$] A bridge in $\cohcodisc \Gamma$ is the same as an edge in $\Gamma$. A path in $\cohcodisc \Gamma$ is also the same as an edge in $\Gamma$. Alternatively, we can say that an edge in $\Gamma$ gives rise to a path in $\cohcodisc \Gamma$, which can then also be weakened to a bridge.
	\item[$\cohpaths$] An edge in $\cohpaths \Gamma$ is the same as a path in $\Gamma$. The alternative formulation --- a path in $\Gamma$ gives rise to an edge in $\cohpaths \Gamma$; a bridge in $\Gamma$ is forgotten --- cannot be formalized as in the previous cases, because $\cohpaths$ was not defined for primitive contexts and hence the property $\cohpaths \circ \yoneda = \yoneda \circ (\ldots)$ cannot be formulated
\end{description}
\begin{proposition}[The cohesive structure of $\widehat{\bpcubecat}$ over $\widehat{\cubecat}$]\label{thm:cohesion-psh}
	These five functors are adjoint: $\cohpi \dashv \cohdisc \dashv \cohfget \dashv \cohcodisc \dashv \cohpaths$. On the $\widehat{\cubecat}$-side, we have
	\begin{equation}
		\bar \shp := \cohpi\cohdisc \cong \Id \quad \dashv \quad
		\cohfget\cohdisc = \Id \quad \dashv \quad
		\cohfget\cohcodisc = \Id \quad \dashv \quad 
		\cohpaths\cohcodisc = \Id \quad
		: \quad \widehat{\cubecat} \to \widehat{\cubecat}.
	\end{equation}
	On the $\widehat{\bpcubecat}$-side, we write
	\begin{equation}
		\shp := \cohdisc\cohpi \quad \dashv \quad
		\flat := \cohdisc\cohfget \quad \dashv \quad
		\sharp := \cohcodisc\cohfget \quad \dashv \quad
		\coshp := \cohcodisc\cohpaths \quad
		: \quad \widehat{\bpcubecat} \to \widehat{\bpcubecat}.
	\end{equation}
	By consequence, we have
	\begin{equation}
		\begin{array}{c c c c c}
			\shp \shp \cong \shp &
			\shp \flat \cong \flat &&\\
			\flat \shp = \shp &
			\flat \flat = \flat &
			\flat \sharp = \flat &\\&
			\sharp \flat = \sharp &
			\sharp \sharp = \sharp &
			\sharp \coshp = \coshp \\&&
			\coshp \sharp = \sharp &
			\coshp \coshp = \coshp,
		\end{array}
	\end{equation}
	
	The following tables lists the units and co-units of all adjunctions involved:
	\begin{equation}
		\begin{array}{c || c | c | c | c}
			& \cohpi \dashv \cohdisc & \cohdisc \dashv \cohfget & \cohfget \dashv \cohcodisc & \cohcodisc \dashv \cohpaths \\
			\hline
			\text{unit} & \varsigma : \Id \to \shp & \id : \Id \to \Id & \iota : \Id \to \sharp & \id : \Id \to \Id \\
			\hline
			\text{co-unit} & \bar \varsigma\inv : \bar \shp \cong \Id & \kappa : \flat \to \Id & \id : \Id \to \Id & \vartheta : \coshp \to \Id
		\end{array}
	\end{equation}
	\begin{equation}
		\begin{array}{c || c || c | c | c}
			& \bar \shp \dashv \Id & \shp \dashv \flat & \flat \dashv \sharp & \sharp \dashv \coshp \\
			\hline
			\text{unit} & \bar \varsigma : \Id \cong \bar \shp & \varsigma : \Id \to \shp & \iota : \Id \to \sharp & \iota : \Id \to \sharp \\
			\hline
			\text{co-unit} & \bar \varsigma \inv : \bar \shp \cong \Id & \kappa \circ (\varsigma \flat)\inv : \shp \flat \to \Id & \kappa : \flat \to \Id & \vartheta : \coshp \to \Id
		\end{array}
	\end{equation}
	Finally, the following natural transformations are all (compatible with) the identity:
	\begin{equation}\label{eq:cohesion-psh-identities}
		\begin{array}{c c c c c | c c c c}
			\cohpi \to \cohpi & \cohdisc \to \cohdisc & \cohfget \to \cohfget & \cohcodisc \to \cohcodisc & \cohpaths \to \cohpaths & \shp \to \shp & \flat \to \flat & \sharp \to \sharp & \coshp \to \coshp \\ \hline
			(\cohpi \varsigma) & (\varsigma \cohdisc) &&&&
			(\shp \varsigma = \varsigma \shp) & (\varsigma \flat) \qquad{} \\
			& \kappa \cohdisc & \cohfget \kappa &&&
			{} \qquad \kappa \shp & \kappa \flat = \flat \kappa & \sharp \kappa \qquad{} \\
			&& \cohfget \iota & \iota \cohcodisc &&
			& {}\qquad \flat \iota & \sharp \iota = \iota \sharp & \iota \coshp \qquad{} \\
			&&& \vartheta \cohcodisc & \cohpaths \vartheta &
			&& {} \qquad \vartheta \sharp & \vartheta \coshp = \coshp \vartheta
		\end{array}
	\end{equation}
	The ones involving $\kappa$, $\iota$ and $\vartheta$ are actually equal, while for $\varsigma$ we have
	\begin{equation*}
		\cohpi \varsigma = \bar \varsigma \cohpi : \cohpi \cong \cohpi \shp, \qquad
		\varsigma \cohdisc = \cohdisc \bar \varsigma : \cohdisc \cong \shp \cohdisc, \qquad
		\shp \varsigma = \varsigma \shp = \cohdisc \bar \varsigma \cohpi : \shp \shp \cong \shp, \qquad
		\varsigma \flat = \cohdisc \bar \varsigma \cohfget : \shp \flat \cong \flat.
	\end{equation*}
\end{proposition}
\begin{lemma}
	A natural transformation $\nu : \fpsh F \to H : \widehat \catV \to \widehat \catW$ whose domain is a lifted functor, is fully determined by $\nu \yoneda F : \fpsh F \yoneda F \to H \yoneda F : \catW \to \widehat \catW$. If $H = \fpsh G$, then $\nu = \fpsh{\tilde \nu}$ for some $\tilde \nu : G \to F : \catW \to \catV$. If $H = \lpsh K \dashv \fpsh K$ for some $K : \catV \to \catW$, then $\nu$ corresponds to a natural transformation $\mu : \Id \to KF : \catW \to \catW$.
\end{lemma}
\begin{proof}
	Pick a presheaf $\Gamma \in \widehat \catV$ and a defining substitution $\fpshadj F (\gamma) : \DSub{W}{\fpsh F \Gamma}$. Then the following diagram commutes:
	\begin{equation}
		\xymatrix{
			& \fpsh F \yoneda F W \ar[r]^{\nu \yoneda FW} \ar[d]^{\fpsh F \gamma}
			& H \yoneda F W \ar[d]^{H \gamma} 
			\\
			W \ar@{=>}[ru]^{\fpshadj F (\id)} \ar@{=>}[r]_{\fpshadj F (\gamma)}
			& \fpsh F \Gamma \ar[r]_{\nu \Gamma}
			& H\Gamma,
		}
	\end{equation}
	showing that $\nu_\Gamma \circ \fpshadj F(\gamma)$ is determined by $\nu \yoneda F W$.
	\begin{itemize}
		\item If $H = \fpsh G$, then we define $\tilde \nu W = \fpshadj G\inv(\nu \yoneda F W \circ \fpshadj F(\id)) : \PSub{GW}{FW}$. Note that if $\nu = \fpsh \mu$, then we would find
		\begin{equation}
			\tilde \nu W = \fpshadj G\inv(\fpsh \mu \yoneda F W \circ \fpshadj F(\id))
			= \fpshadj G\inv(\fpshadj G(\id \circ \mu W)) = \mu W : \PSub{GW}{FW}.
		\end{equation}
		In general, this is a natural transformation because if we have $\vfi : \PSub V W$, then\footnote{Remember that we write $\gamma : \yoneda W \to \Gamma$ when $\gamma : \DSub W \Gamma$, implying that we do not write $\yoneda$ when applied to a morphism.}
		\begin{align*}
			F \vfi \circ \tilde \nu V
			&= \fpshadj G \inv (\fpsh G (F \vfi) \circ \nu \yoneda F V \circ \fpshadj F (\id)) \\
			&= \fpshadj G \inv (\nu \yoneda F W \circ \fpsh F (F \vfi) \circ \fpshadj F (\id)) \\
			&= \fpshadj G \inv (\nu \yoneda F W \circ \fpshadj F (F \vfi)) \\
			&= \fpshadj G \inv (\nu \yoneda F W \circ \fpshadj F (\id) \circ \vfi) \\
			&= \fpshadj G \inv (\nu \yoneda F W \circ \fpshadj F (\id)) \circ G \vfi
			= \tilde \nu W \circ G \vfi.
		\end{align*}
		Moreover, the above diagram shows that $\nu = \fpsh{\tilde \nu}$ because
		\begin{equation}
			\nu \circ \fpshadj F(\gamma) = \fpsh G \gamma \circ \nu \yoneda F W \circ \fpshadj F(\id) = \fpsh G \gamma \circ \fpshadj G (\tilde \nu) = \fpshadj G (\gamma \circ \tilde \nu) = \fpsh{\tilde \nu} \circ \fpshadj G(\gamma).
		\end{equation}
		
		\item If $H = \lpsh K \dashv \fpsh K$, then we show that natural transformations $\fpsh F \to \lpsh K : \widehat \catV \to \widehat \catW$ correspond to natural transformations $\Id \to KF : \catW \to \catW$. We already know that $\nu : \fpsh F \to \lpsh K$ is determined by $\nu \yoneda F : \fpsh F \yoneda F \to \lpsh K \yoneda F$, and \cref{thm:yoneda-and-left-adjoint} tells us that $\zeta : \lpsh K \yoneda F \cong \yoneda K F$.
		
		Given $\nu$, we now define $\mu : \Id \to KF$ by $\mu W = \zeta W \circ \nu \yoneda F W \circ \fpshadj F(\id_W) \in (\DSub{W}{\yoneda KF W}) = (\PSub{W}{KFW})$. Conversely, given $\mu : \Id \to KF$, we define $\nu : \fpsh F \to \lpsh K$ by setting for every $\fpshadj F(\gamma) : \DSub{W}{\fpsh F \Gamma}$ (i.e. $\gamma : \DSub{FW}{\Gamma}$), the composition $\nu \Gamma \circ \fpshadj F(\gamma)$ equal to $\lpsh K \gamma \circ (\zeta W)\inv \circ \mu W : \DSub W {\lpsh K \Gamma}$. These operations are inverse:
		\begin{align*}
			\lpsh K \gamma \circ (\zeta W)\inv \circ \mu W
			&= \lpsh K \gamma \circ (\zeta W)\inv \circ \zeta W \circ \nu \yoneda F W \circ \fpshadj F(\id_W) \\
			&= \lpsh K \gamma \circ \nu \yoneda F W \circ \fpshadj F(\id_W)
			= \nu \Gamma \circ \fpsh F \gamma \circ \fpshadj F(\id_W)
			= \nu \Gamma \circ \fpshadj F(\gamma). \\
			\zeta W \circ \nu \yoneda F W \circ \fpshadj F(\id_W)
			&= \zeta W \circ \lpsh K\,\id_W \circ (\zeta W)\inv \circ \mu W = \mu W. \qedhere
		\end{align*}
	\end{itemize}
\end{proof}
\begin{corollary}\label{thm:uniqueness-of-nattrans-psh}
	Let $\catV, \catW \in \accol{\cubecat, \bpcubecat}$ and $F, G : \widehat \catV \to \widehat \catW$. Then all natural transformations from $F$ to $G$ are equal in each of the following cases:
	\begin{enumerate}
		\item If $F$ and $G$ are composites of $\cohdisc$, $\cohfget$, $\cohcodisc$ and $\cohpaths$;
		\item If $F$ is a composite of $\cohdisc$, $\cohfget$, $\cohcodisc$ and $\cohpaths$, and $G$ is a composite of $\cohpi$, $\cohdisc$, $\cohfget$ and $\cohcodisc$;
		\item If $F$ factors as $LP$, where $L \dashv R$ and one of the previous cases apply to $P$ and $RG$.
	\end{enumerate}
\end{corollary}
\begin{proof}
	Pick $\nu : F \to G$.
	\begin{enumerate}
		\item Then both $F$ and $G$ are lifted so that $\nu = \fpsh{\tilde \nu}$, and $\tilde \nu$ is completely determined by \cref{thm:uniqueness-of-nattrans-base}.
		\item Then $F$ is lifted and $G$ is left adjoint to a lifted functor, so that $\nu$ corresponds to a natural transformation of primitive contexts, which is uniquely determined because of \cref{thm:uniqueness-of-nattrans-base}.
		\item Natural transformations $LP \to G$ are in bijection with natural transformations $P \to RG$ because $L \dashv R$. \qedhere
	\end{enumerate}
\end{proof}
\begin{lemma}\label{thm:unicity-of-left-adjoint}
	Assume $L_1, L_2 : \catV \to \catW$ and $R : \catW \to \catV$ such that $\alpha_i : L_i \dashv R$ with unit $\eta_i : \Id \to R L_i$ and co-unit $\eps_i : L_i R \to \Id$. Then there is a natural isomorphism $\zeta : L_1 \cong L_2$ such that $R \zeta \circ \eta_1 = \eta_2$ and $\eps_1 = \eps_2 \circ \zeta R$.
\end{lemma}
\begin{proof}
	We set $\zeta = \eps_1 L_2 \circ L_1 \eta_2$ and $\zeta\inv = \eps_2 L_1 \circ L_2 \eta_1$. We show that $\zeta \circ \zeta\inv = \id$; the other equation holds by symmetry of the indices. Observe the commutative diagram:
	\begin{equation}
		\xymatrix{
			L_2 \ar[rr]^{L_2 \eta_1} \ar[d]_{L_2 \eta_2}
			\ar@/^{2em}/[rrrr]^{\zeta}
			&& L_2 R L_1 \ar[rr]^{\eps_2 L_1} \ar[d]_{L_2 R L_1 \eta_2}
			&& L_1 \ar[d]_{L_1 \eta_2}
			\ar@/^{2em}/[dd]^{\zeta\inv}
			\\
			L_2 R L_2 \ar[rr]^{L_2 \eta_1 R L_2} \ar[rrd]_{\id}
			&& L_2 R L_1 R L_2 \ar[rr]^{\eps_2 L_1 R L_2} \ar[d]^{L_2 R \eps_1 L_2}
			&& L_1 R L_2 \ar[d]_{\eps_1 L_2}
			\\
			&& L_2 R L_2 \ar[rr]_{\eps_2 L_2}
			&& L_2
		}
	\end{equation}
	The top right square applies naturality of $\eps_2 L_1$ to $L_1 \eta_2$. The top left square still holds after removing $L_2$ on the left and is then an instance of naturality of $\eta_1$. The lower right square still holds after removing $L_2$ on the right and is then an instance of naturality of $\eps_2$. The lower right triangle commutes because $R \eps_1 \circ \eta_1 R = \id$. Finally, the entire left-lower side composes to the identity. We have
	\begin{align*}
		R\zeta \circ \eta_1
		&= R \eps_1 L_2 \circ R L_1 \eta_2 \circ \eta_1
		= R \eps_1 L_2 \circ \eta_1 RL_2 \circ \eta_2 = \eta_2, \\
		\eps_2 \circ \zeta R
		&= \eps_2 \circ \eps_1 L_2 R \circ L_1 \eta_2 R
		= \eps_1 \circ L_1 R \eps_2 \circ L_1 \eta_2 R = \eps_1. \qedhere
	\end{align*}
\end{proof}
\begin{proof}[Proof of \cref{thm:cohesion-psh}]
	The adjunctions $\cohdisc \dashv \cohfget \dashv \cohcodisc \dashv \cohpaths$ follow from \cref{thm:lifting-adjunctions}. For now, we just assume $\cohpi$ to be some left adjoint to $\cohdisc$.
	
	The fact that $\cohfget \cohdisc = \Id$, $\cohfget\cohcodisc = \Id$ and $\cohpaths\cohcodisc = \Id$ follows from the fact that $\fpsh \loch$ swaps composition and preserves identity.
	
	It is clear that both $\bar \shp := \cohpi \cohdisc$ and $\Id$ are left adjoint to $\cohfget \cohdisc = \Id$. Then \cref{thm:unicity-of-left-adjoint} below gives us, after filling in the identity in various places, an isomorphism $\bar \varsigma : \Id \cong \bar \shp$ which is the unit of $\bar \shp \dashv \Id$, while the co-unit is $\bar \varsigma \inv$.
	
	The equalities and isomorphisms are obvious.
	
	We \emph{define} $\varsigma$ as the unit of $\cohpi \dashv \cohdisc$.
	The rest of the theorem now follows from \ref{thm:uniqueness-of-nattrans-psh}.
\end{proof}

\subsection{Characterizing cohesive adjunctions}
We currently have various cohesion-based ways of manipulating terms: we can apply functors, turning ($t \mapsto \ftrtm F t$), we can apply natural transformations ($t \mapsto \nu(t)$), we can instead substitute with natural transformations ($t \mapsto t[\nu]$) and apply adjunctions ($t \mapsto \alpha(t)$). We have various equations telling us how these relate, but altogether it becomes hard to tell whether terms are equal. For this reason, we will at least try to write the relevant adjunctions in terms of the other constructions.
\begin{proposition}
	For any contexts $\Gamma$ and $\Delta$, we have the following diagrams, in which all arrows are invertible:
	\begin{equation}
		\xymatrix{
			(\shp \Gamma \to \Delta)
				\ar[rr]^{\alpha_{\shp \dashv \flat}}
				\ar@{<-}[rd]_{\kappa \circ \loch}
			&&
			(\Gamma \to \flat \Delta)
				\ar@{<-}[ld]^{\loch \circ \varsigma}
			\\&
			(\shp \Gamma \to \flat \Delta)
		\\
			(\flat \Gamma \to \Delta)
				\ar[rr]^{\alpha_{\flat \dashv \sharp}}
				\ar[rd]_{\iota \circ \loch}
			&&
			(\Gamma \to \sharp \Delta)
				\ar[ld]^{\loch \circ \kappa}
			\\&
			(\flat \Gamma \to \sharp \Delta)
		\\
			(\sharp \Gamma \to \Delta)
				\ar[rr]^{\alpha_{\sharp \dashv \coshp}}
				\ar@{<-}[rd]_{\vartheta \circ \loch}
			&&
			(\Gamma \to \coshp \Delta)
				\ar@{<-}[ld]^{\loch \circ \iota}
			\\&
			(\sharp \Gamma \to \coshp \Delta)
		}
	\end{equation}
\end{proposition}
\begin{proof}
	For all diagonal arrows except $\kappa \circ \loch$, we can again take the adjunction isomorphism since $\flat$, $\sharp$ and $\coshp$ are idempotent. One can check that these boil down to composition with a (co)-unit, e.g. for $\sigma : \flat \Gamma \to \Delta$ we have
	\begin{equation}
		\alpha_{\flat \dashv \sharp}^{\flat \Gamma, \Delta}(\sigma)
		= \sharp \sigma \circ \iota \flat \Gamma
		= \iota \Delta \circ \sigma.
	\end{equation}
	For $\sigma : \sharp \Gamma \to \coshp \Delta$, we have
	\begin{equation}
		\alpha_{\sharp \dashv \coshp}^{\Gamma, \coshp \Delta}(\sigma)
		= \coshp \sigma \circ \iota \Gamma = \sigma \circ \iota \Gamma
	\end{equation}
	because $\coshp \sigma = \vartheta \coshp \Delta \circ \coshp \sigma = \sigma \circ \vartheta \sharp \Gamma = \sigma$. For the arrow $\kappa \circ \loch$, we use a composition of isomorphisms
	\begin{equation}
		(\shp \Gamma \to \flat \Delta) \xrightarrow{(\alpha_{\shp \dashv \flat}^{\shp \Gamma, \Delta})\inv} (\shp \shp \Gamma \to \flat \Delta) \xrightarrow{\loch \circ \varsigma \shp \Gamma} (\shp \Gamma \to \Delta).
	\end{equation}
	A substitution $\sigma : \shp \Gamma \to \flat \Delta$ is then mapped to
	\begin{equation}
		(\alpha_{\shp \dashv \flat}^{\shp \Gamma, \Delta})\inv(\sigma) \circ \varsigma \shp \Gamma
		= \kappa \Delta \circ (\varsigma \flat \Delta)\inv \circ \shp \sigma \circ \varsigma \shp \Gamma
		= \kappa \Delta \circ (\varsigma \flat \Delta)\inv \circ \varsigma \flat \Delta \circ \sigma
		= \kappa \Delta \circ \sigma.
	\end{equation}
	Finally, one can check that the diagrams commute, e.g. if $\sigma : \shp \Gamma \to \flat \Delta$ then
	\begin{equation}
		\alpha_{\shp \dashv \flat}^{\Gamma, \Delta}(\kappa \Delta \circ \sigma)
		= \flat \kappa \Delta \circ \flat \sigma \circ \varsigma \Gamma
		= \flat \sigma \circ \varsigma \Gamma = \sigma \circ \varsigma \Gamma
	\end{equation}
	because $\flat \sigma = \kappa \flat \Delta \circ \flat \sigma = \sigma \circ \kappa \shp \Gamma = \sigma$.
\end{proof}
\begin{notation}
	When it exists, we write $\tau \setminus \sigma$ for the unique substitution such that $\tau \circ (\tau \setminus \sigma) = \sigma$, and $\sigma / \tau$ for the unique substitution such that $(\sigma / \tau) \circ \tau = \sigma$. Uniqueness implies that $\tau \setminus (\tau \circ \sigma) = \sigma$ and $(\sigma \circ \tau) \setminus \tau = \sigma$ if the left hand side exists.
	
	Similarly, when it exists, we write $\nu\inv(t)$ for the unique term such that $\nu(\nu\inv(t)) = t$ and $t[\nu]\inv$ for the unique term such that $t[\nu]\inv[\nu] = t$. Uniqueness implies that $\nu\inv(\nu(t)) = t$ and $t[\nu][\nu]\inv = t$.

	The above theorem then justifies the following notations:
	\begin{equation}
		\begin{array}{r | c | c | c}
			& \shp \dashv \flat & \flat \dashv \sharp & \sharp \dashv \coshp
			\\ \hline
			\alpha(\sigma)
			& (\kappa \setminus \sigma) \circ \varsigma
			& (\iota \circ \sigma) / \kappa
			& (\vartheta \setminus \sigma) \circ \iota
			\\
			\alpha\inv(\tau)
			& \kappa \circ (\tau / \varsigma)
			& \iota \setminus (\tau \circ \kappa)
			& \vartheta \circ (\tau / \iota)
			\\
			\alpha(t)
			& \kappa\inv(t)[\varsigma]
			& \iota(t)[\kappa]\inv
			& \vartheta\inv(t)[\iota]
			\\
			\alpha\inv(u)
			& \kappa(t[\varsigma]\inv)
			& \iota\inv(t[\kappa])
			& \vartheta(t[\iota]\inv)
		\end{array}
	\end{equation}
	Note that terms correspond to substitutions to an extended context, which are then subject to the diagrams above.
\end{notation}

\chapter{Discreteness}\label{ch:discreteness}
It is common in categorical models of dependent type theory to designate a certain class of morphisms $\homclass F$ in the category of contexts, typically called \textbf{fibrations}, and to require that for any type $\Gamma \sez T \type$, the morphism $\pi : \Gamma.T \to \Gamma$ is a fibration. Types satisfying this criterion are then called \textbf{fibrant}. A context $\Gamma$ is called \textbf{fibrant} if the map $\Gamma \to ()$ is a fibration.

Typically, the fibrations can be characterized using a lifting property with respect to another class of morphisms $\homclass H$ which we will call \textbf{horn inclusions}. If $\eta : \Lambda \to \Delta$ is a horn inclusion, then we call a map $\sigma : \Lambda \to \Gamma$ a \textbf{horn} in $\Gamma$ and if $\sigma$ factors as $\tau \eta$, then $\tau : \Delta \to \Gamma$ is called a \textbf{filler} of $\sigma$.

Now the fibrations are usually those morphisms $\rho : \Gamma' \to \Gamma$ such that any horn $\sigma$ in $\Gamma'$ which has a filler $\tau$ in $\Gamma$ (meaning that $\rho\sigma$ has a filler $\tau$), also has a (compatible) filler in $\Gamma'$, i.e. commutative squares like the following have a diagonal:
\begin{equation}
	\xymatrix{
		\Lambda \ar[r]^\sigma \ar[d]_{\eta \in \homclass H}
		& \Gamma' \ar[d]^{\rho}
		\\
		\Delta \ar[r]_{\tau} \ar@{.>}[ru]
		&
		\Gamma
	}
\end{equation}
In this text, we will call the fibrant types and contexts \textbf{discrete} and we will speak of \textbf{discrete maps} instead of fibrations, because this better reflects the idea behind what we are doing.

In many models, only fibrant contexts are used. However, we will also consider non-discrete contexts because our modalities do not preserve discreteness.

\section{Definition}
\begin{definition}
	Let $\XX$ stand for $\IB$, $\IP$ or $\IE$.
	We say that a defining substitution $\gamma : \DSub{(W, \ctxline{\var i})}{\Gamma}$ or a defining term $(W, \ctxline{\var i}) \Dsez t : T \dsub \gamma$ is \textdef{degenerate in $\var i$} if it factors over $(\facewkn{\var i}) : \PSub{(W, \ctxline{\var i})}{W}$.
\end{definition}
The notion of degeneracy is thus meaningful in the CwFs $\widehat \cubecat$ and $\widehat \bpcubecat$. Thinking of $\var i$ as a variable, this means that $\gamma$ and $t$ do not refer to $\var i$. Thinking of $\var i$ as a dimension, this means that $\gamma$ and $t$ are flat in dimension $\var i$. Note that $t$ can only be degenerate in $\var i$ if $\gamma$ is.
\begin{corollary}
	For a defining substitution $\gamma : \DSub{(W, \ctxline{\var i})}{\Gamma}$ or a defining term $(W, \ctxline{\var i}) \Dsez t : T \dsub \gamma$, the following are equivalent:
	\begin{enumerate}
		\item $\gamma$/$t$ is degenerate in $\var i$,
		\item $\gamma = \gamma \circ (0/\var i, \facewkn{\var i})$; $t = t \psub{0/\var i, \facewkn{\var i}}$,
		\item $\gamma = \gamma \circ (1/\var i, \facewkn{\var i})$; $t = t \psub{1/\var i, \facewkn{\var i}}$. \qed
	\end{enumerate}
\end{corollary}
\begin{definition}
	We call a context \textbf{discrete} if all of its cubes are degenerate in every path dimension.
	
	We call a map $\rho : \Gamma' \to \Gamma$ \textbf{discrete} if every defining substitution $\gamma$ of $\Gamma'$ is degenerate in every path dimension in which $\rho \circ \gamma$ is degenerate.
	
	We call a type $\Gamma \sez T \type$ \textbf{discrete} (denoted $\Gamma \sez T \dtype$) if every defining term $t : T \dsub \gamma$ is degenerate in every path dimension in which $\gamma$ is degenerate.
\end{definition}
\begin{proposition}
	A type $\Gamma \sez T \type$ is discrete if and only if $\pi : \Gamma.T \to \Gamma$ is discrete.
\end{proposition}
\begin{proof}
	\begin{itemize}
		\item[$\Rightarrow$] Assume that $T$ is discrete. Pick $(\gamma, t) : \DSub{(W, \ctxpath{\var i})}{\Gamma.T}$ such that $\pi \circ (\gamma, t) = \gamma$ is degenerate in $\var i$. Then $t$ is degenerate in $\var i$ by discreteness of $T$ and so is $(\gamma, t)$.
		\item[$\Leftarrow$] Assume that $pi$ is discrete. Pick $t : T \dsub \gamma$ where $\gamma$ is degenerate in $\var i$. Then $(\gamma, t)$ is degenerate in $\var i$ since $\pi(\gamma, t) = \gamma$, and hence $t$ is degenerate in $\var i$. \qedhere
	\end{itemize}
\end{proof}
\begin{proposition}
	A context $\Gamma$ is discrete if and only if $\Gamma \to ()$ is discrete.
\end{proposition}
\begin{proof}
	Note that every defining substitution of $()$ is degenerate in every dimension. This proves the claim.
\end{proof}
\begin{proposition}
	A map $\rho : \Gamma' \to \Gamma$ is discrete if and only if it has the lifting property with respect to all horn inclusions $(\facewkn{\var i}) : \yoneda(W, \ctxpath{\var i}) \to \yoneda W$.
\end{proposition}
\begin{proof}
	\begin{itemize}
		\item[$\Rightarrow$] Suppose that $\rho$ is discrete and consider a square
		\begin{equation}
			\xymatrix{
				\yoneda(W, \ctxpath{\var i}) \ar[r]^{\gamma'} \ar[d]_{(\facewkn{\var i})}
				& \Gamma' \ar[d]^{\rho}
				\\
				\yoneda W \ar[r]_{\gamma}
				&
				\Gamma.
			}
		\end{equation}
		Then the defining substitution $\rho \circ \gamma' : \DSub{(W, \ctxpath{\var i})}{\Gamma}$ clearly factors over $(\facewkn{\var i})$ so that it is degenerate in $\var i$. By degeneracy of $\rho$, the same holds for $\gamma'$, yielding the required diagonal.
		
		\item[$\Leftarrow$] Suppose that $\rho$ has the lifting property and take $\gamma' : \DSub{(W, \ctxpath{\var i})}{\Gamma'}$ such that $\rho \circ \gamma'$ is degenerate in $\var i$. This gives us a square as above, which has a diagonal, showing that $\gamma'$ is degenerate. \qedhere
	\end{itemize}
\end{proof}
\begin{example}
	In the examples, we will develop the content of this chapter for the CwF $\widehat{\RGcat}$ of reflexive graphs. This will fail when we come to product types, which is the reason why we choose to work with presheaves over $\widehat{\bpcubecat}$ which contain not only points, paths and bridges, but also coherence cubes. An alternative is to require edges (bridges and paths) to be proof-irrelevant in the style of \cite{dtt-parametricity}, but this property cannot be satisfied by the universe. Remember that $\RGcat$ has objects $()$ and $(\ctxedge{\var i})$ and the same morphisms between them as we find in $\cubecat$.
	
	We call an edge $p : \DSub{(\ctxedge{\var i})}{\Gamma}$ \textbf{degenerate} if it is the constant edge on some point $x : \DSub{()}{\Gamma}$, i.e.\ $p = x \circ (\facewkn{\var i})$. This point is uniquely determined, as it must be equal to the edge's source and target. So we can say that $p$ is degenerate iff $p = p \circ (0/\var i) \circ (\facewkn{\var i})$ iff $p \circ (1/\var i) \circ (\facewkn{\var i})$.
	
	We call a context (reflexive graph) \textbf{discrete} if all of its edges are degenerate.
	
	We call a map $\rho : \Gamma' \to \Gamma$ \textbf{discrete} if every edge $\gamma : \DSub{(\ctxedge{\var i})}{\Gamma}$ for which $\rho \circ \gamma$ is degenerate, is itself degenerate.
	
	We call a type $\Gamma \sez T \type$ \textbf{discrete} if every defining edge $(\ctxedge{\var i}) \Dsez t : T \dsub \gamma$ over a degenerate edge $\gamma$, is degenerate.
	
	One can prove:
	\begin{itemize}
		\item A type $\Gamma \sez T \type$ is discrete if and only if $\pi : \Gamma.T \to \Gamma$ is discrete.
		\item A context $\Gamma$ is discrete if and only if $\Gamma \to ()$ is discrete.
		\item A map $\rho : \Gamma' \to \Gamma$ is discrete if it has the lifting property with respect to the horn inclusion $(\facewkn{\var i}) : \yoneda(\ctxedge{\var i}) \to \yoneda()$.
	\end{itemize}
\end{example}

\section{A model in terms of discrete types}
\begin{theorem}
	If we define $\DTy(\Gamma)$ to be the set of all \emph{discrete} types over $\Gamma$, then we obtain a new CwF $\bpdisc$ which also supports dependent products, dependent sums and identity types.
\end{theorem}
We prove this theorem in several parts. In the examples, we will try and fail to prove the corresponding theorem for $\RGcat$: we do have a CwF $\RGdisc$ and it supports dependent sums and identity types, but we will fail to prove that it supports dependent products.

\subsection{The category with families $\bpdisc$}
\begin{lemma}
	$\bpdisc$ is a well-defined category with families (see \cref{def:cwf}).
\end{lemma}
\begin{proof}
	The only thing we need to prove in order to show this is that $\DTy$ still has a morphism part, i.e. that discreteness of types is preserved under substitution.

	So pick a substitution $\sigma : \Delta \to \Gamma$ and a discrete type $\Gamma \sez T \dtype$. Take a defining term $(W, \ctxpath{\var i}) \Dsez t : T \sub{\sigma} \dsub \delta$ and assume that $\delta$ is degenerate along $\var i$. We have to prove that $t$ is, too. But if $\delta$ factors over $(\facewkn{\var i})$, then so does $\sigma \delta$, and therefore also $t : T \dsub{\sigma \delta}$ by discreteness of $T$. Since restriction for $T[\sigma]$ is inherited from $T$, the term $t$ is also degenerate as a defining term of $T[\sigma]$.
\end{proof}
\begin{example}
	Show that $\RGdisc$ is a well-defined category with families.
\end{example}

\subsection{Dependent sums}
\begin{lemma}
	The category with families $\bpdisc$ supports dependent sums.
\end{lemma}
\begin{proof}
	Take a context $\Gamma$ and discrete types $\Gamma \sez A \dtype$ and $\Gamma.A \sez B \dtype$. It suffices to show that $\Sigma A B$ is discrete. Pick a defining term $(W, \ctxpath{\var i}) \Dsez (a, b) : \Sigma A B \dsub \gamma$ where $\gamma$ is degenerate along $\var i$. Then by discreteness of $A$, $a$ is degenerate along $\var i$ and so is $(\gamma, a) : \DSub{(W, \ctxpath{\var i})}{\Gamma.A}$. Then by discreteness of $B$, $b$ is also degenerate along $\var i$ and hence so is $(a, b)$.
\end{proof}

\subsection{Dependent products}
\begin{lemma}
	The category with families $\bpdisc$ supports dependent products.
\end{lemma}
In fact, the proof of this lemma proves something stronger:
\begin{lemma}\label{thm:disc-prod}
	Given a context $\Gamma$, an arbitrary type $\Gamma \sez A \type$ and a discrete type $\Gamma.A \sez B \dtype$, the type $\Pi A B$ is discrete.
\end{lemma}
\begin{example}
	In order to clarify the idea behind the proof for $\bpdisc$, we will first (vainly) try to prove the same lemma for $\RGdisc$.
	
	Pick a context $\Gamma$, a type $\Gamma \sez A \type$ and a discrete type $\Gamma.A \sez B \dtype$. It suffices to show that $\Pi A B$ is discrete. Pick an edge $(\ctxedge{\var i}) \Dsez h : (\Pi A B) \dsub{\gamma (\facewkn{\var i})}$ over the constant edge at point $\gamma : \DSub{()}{\Gamma}$. Write
	\begin{equation}
		() \Dsez f := h \psub{0/\var i}, g := h \psub{1 / \var i} : (\Pi A B) \dsub \gamma.
	\end{equation}
	In order to have a visual representation, assume that $A$ looks like the upper diagram here; then the image of $h$ is the lower one (degenerate edges are hidden in both diagrams). Every cell projects to the cell from $\Gamma$ shown on its left, or a constant edge of it:
	\begin{equation}
		\xymatrix{
			a \ar@{-}[rr]^{a_1}
			&& a' \ar@{-}[rr]^{a_2}
			&& a''
			\\
			f \cdot a
				\ar@{-}[rr]^{f \psub{\facewkn{\var i}} \cdot a_1} 
				\ar@{-}[dd]^(.7){h \cdot (a \psub{\facewkn{\var i}})}
			  	\ar@{-}[rrdd]^{h \cdot a_1}
			&& f \cdot a'
				\ar@{-}[rr]^{f \psub{\facewkn{\var i}} \cdot a_2}
				\ar@{-}[dd]^(.7){h \cdot (a' \psub{\facewkn{\var i}})}
			  	\ar@{-}[rrdd]^{h \cdot a_2}
			&& f \cdot a'' 
				\ar@{-}[dd]^(.7){h \cdot (a'' \psub{\facewkn{\var i}})}
			\\ \\
			g \cdot a
				\ar@{-}[rr]_{g \psub{\facewkn{\var i}} \cdot a_1}
			&& g \cdot a'
				\ar@{-}[rr]_{g \psub{\facewkn{\var i}} \cdot a_2}
			&& g \cdot a''
		}
	\end{equation}
	We try to show that $h$ is degenerate by showing that $f \psub{\facewkn{\var i}} = h$. Recall that a defining $W \Dsez k : (Pi A B)\dsub \gamma$ is fully determined if we know all $k \psub \vfi \cdot a$ for all $\vfi : \PSub V W$ and all $V \Dsez a : A \dsub{\gamma \vfi}$. There are five candidates for $\vfi$:
	\begin{itemize}
		\item[$(0/\var i)$] We have $f \psub{\facewkn{\var i}} \psub{0/\var i} = f = h \psub{0/\var i}$.
		\item[$(0/\var i, \facewkn{\var i})$] This follows by further restriction by $(\facewkn{\var i})$.
		\item[$(1/\var i)$] Take a node $() \Dsez a : A \dsub \gamma$. We need to show that $f \psub{\facewkn{\var i}} \psub{1/\var i} \cdot a = h \psub{1/\var i} \cdot a$, i.e.\ $f \cdot a = g \cdot a$. To this end, consider $h \cdot (a \psub{\facewkn{\var i}})$. The fact that $\loch \cdot \loch$ commutes with restriction, guarantees that this is an edge from $f \cdot a$ to $g \cdot a$. However, it has type $(\ctxedge{\var i}) \Dsez h \cdot (a \psub{\facewkn{\var i}}) : B [(\gamma (\facewkn{\var i})) \subext] \dsub{\id, a \psub{\facewkn{\var i}}} = B \dsub{(\gamma, a)(\facewkn{\var i})}$. So it clearly lives over a degenerate edge and hence it is degenerate, implying that $f \cdot a = g \cdot a$.
		\item[$(1/\var i, \facewkn{\var i})$] Take an edge $(\ctxedge{\var i}) \Dsez a : A \dsub{\gamma (\facewkn{\var i})}$. We need to show that $f \psub{\facewkn{\var i}} \psub{1/\var i, \facewkn{\var i}} \cdot a = h \psub{1/\var i, \facewkn{\var i}} \cdot a$, i.e. $f \psub{\facewkn{\var i}} \cdot a = g \psub{\facewkn{\var i}} \cdot a$. This is an equality of edges. If we could introduce an additional variable $\var j$, we could apply the same technique as for the source extractor $(0 / \var i)$, considering a square from $f \psub{\facewkn{\var i}} \cdot a$ to $g \psub{\facewkn{\var i}} \cdot a$ that would have to be degenerate in one dimension. However, the squares would cause functions to have more components, which we would have to prove equal, requiring a third variable. This is why we chose to start from a model $\bpcubecat$ in which we can have arbitrarily many dimension variables. However, in $\RGdisc$, we cannot proceed.
		\item[$\id$] Take an edge $(\ctxedge{\var i}) \Dsez a : A \dsub{\gamma (\facewkn{\var i})}$. We need to show that $f \psub{\facewkn{\var i}} \cdot a = h \cdot a$. Now $h \cdot a$ is going to be the diagonal of the square we fancied in the previous clause. If we know that this degenerate, then the diagonal is equal to the sides. However, in $\RGdisc$, we cannot proceed.
	\end{itemize}	 
\end{example}
\begin{proof}
	Pick a context $\Gamma$, a type $\Gamma \sez A \type$ and a discrete type $\Gamma.A \sez B \dtype$. It suffices to show that $\Pi A B$ is discrete. Pick $(W, \ctxpath{\var i}) \Dsez h : (\Pi A B) \dsub{\gamma(\facewkn{\var i})}$. We name the $\var i$-source $f := h \psub{0/\var i}$ and the $\var i$-target $g := h \psub{1 / \var i}$. We have $W \Dsez f, g : (\Pi A B) \dsub \gamma$.
	
	We show that $h$ is degenerate along $\var i$ by showing that $f \psub{\facewkn{\var i}} = h$. So pick some $\vfi : \PSub{V}{(W, \ctxpath{\var i})}$. We make a case distinction by inspecting $\var i \psub{\vfi} \in \accol{0, 1} \uplus V$:
	\begin{description}
		\item[$\var i \psub \vfi = 0$] Then $\vfi = (0/\var i)\psi$ for some $\psi : \PSub V W$. Then we have $f \psub{\facewkn{\var i}} \psub{\vfi} = f \psub \psi$ and $h \psub \vfi = f \psub \psi$.
		
		\item[$\var i \psub \vfi = 1$] Then $\vfi = (1/\var i)\psi$ for some $\psi : \PSub V W$. Then we have $f \psub{\facewkn{\var i}} \psub{\vfi} = f \psub \psi$ and $h \psub \vfi = g \psub \psi$ and we need to show that $V \Dsez f \psub \psi \cdot a = g \psub \psi \cdot a : B \dsub{\psi, a}$ for every $V \Dsez a : A \dsub \psi$.
		
		Without loss of generality, we may assume that $\var i \not\in V$.
		Then we have a path $(V, \ctxpath{\var i}) \Dsez h \psub{\psi, \var i/\var i} \cdot (a \psub{\facewkn{\var i}}) : B \dsub{(\psi, a)(\facewkn{\var i})}$, with source $f \psub \psi$ and target $g \psub \psi$. Moreover, by discreteness of $B$, it is degenerate along $\var i$, implying that source and target are equal.
		
		\item[$\var i \psub \vfi \in V$] Without loss of generality, we may assume that $(W, \ctxpath{\var i})$ and $V$ are disjoint. Write $\var k = \var i \psub \vfi$ (and note that $\var k$ may be either a bridge or a path variable). Then $\vfi$ factors as $(\psi, \var i / \var i)(\var k / \var i) = (\var k / \var i)(\psi, \var k / \var k)$ for some $\psi : \PSub V W$. We have $f \psub{\facewkn{\var i}} \psub{\vfi} = f \psub \psi$ and $h \psub \vfi = h \psub{\var k / \var i} \psub \psi$. We have to show that $V \Dsez f \psub \psi \cdot a = h \psub{\psi, \var i / \var i} \psub{\var k / \var i} \cdot a : B \dsub{\psi, a}$ for all $V \Dsez a : A \dsub \psi$.
		
		Again, we have a path $(V, \ctxpath{\var i}) \Dsez h \psub{\psi, \var i/\var i} \cdot (a \psub{\facewkn{\var i}}) : B \dsub{(\psi, a)(\facewkn{\var i})}$ with $\var i$-source $f \psub \psi \cdot a$ and $(\var k/\var i)$-diagonal $h \psub{\psi, \var i / \var i} \psub{\var k / \var i} \cdot a$. This path is again degenerate in $\var i$, showing that the source and the diagonal are equal. \qedhere
	\end{description}
\end{proof}

\subsection{Identity types and propositions}
\begin{lemma}
	The category with families $\bpdisc$ supports identity types.
\end{lemma}
We even have a stronger result:
\begin{lemma}
	Propositions are discrete. \qed
\end{lemma}

\subsection{Glueing}
\begin{lemma}
	The category with families $\bpdisc$ supports glueing.
\end{lemma}
\begin{proof}
	Suppose we have $\Gamma \sez A \dtype$, $\Gamma \sez P \prop$, $\Gamma.P \sez T \dtype$ and $\Gamma.P \sez f : T \to A[\pi]$. It suffices to show that $G = \Gluesys{A}{\Gluesysclauseb{P}{T}{f}}$ is discrete. So pick $(W, \ctxpath{\var i}) \Dsez b : G \dsub{\gamma}$ where $\gamma$ is degenerate along $\var i$.
	
	If $P \dsub \gamma = \accol \star$, then $b \psub{0/\var i, \facewkn{\var i}}^G = b \psub{0/\var i, \facewkn{\var i}}^{T[\id, \star]} = b$ by discreteness of $T$.
	
	If $P \dsub \gamma = \eset$, then $b$ is of the form $(a \mapsfrom t)$. Then $(a \mapsfrom t) \psub{0/\var i, \facewkn{\var i}} = (a \psub{0/\var i, \facewkn{\var i}} \mapsfrom t[(0/\var i, \facewkn{\var i}) \subext]) = (a \mapsfrom t[(0/\var i, \facewkn{\var i}) \subext])$ by discreteness of $A$. Finally, discreteness of $\Pi P T$ shows that
	\begin{equation}
		t[(0/\var i, \facewkn{\var i}) \subext]
		= \dap (\dlambda (t[(0/\var i, \facewkn{\var i}) \subext]))
		= \dap((\dlambda t) \psub{0/\var i, \facewkn{\var i}})
		= \dap (\dlambda t) = t. \qedhere
	\end{equation}
\end{proof}

\subsection{Welding}
\begin{lemma}
	The category with families $\bpdisc$ supports welding.
\end{lemma}
\begin{proof}
	Suppose we have $\Gamma \sez A \dtype$, $\Gamma \sez P \prop$, $\Gamma.P \sez T \dtype$ and $\Gamma.P \sez f : A[\pi] \to T$. It suffices to show that $\Omega = \Weldsys{A}{\Weldsysclauseb P T f}$ is discrete. So pick $(W, \ctxpath{\var i}) \Dsez w : \Omega \dsub \gamma$ where $\gamma$ is degenerate along $\var i$.
	
	If $P \dsub \gamma = \accol \star$, then $w \psub{0/\var i, \facewkn{\var i}}^\Omega = w \psub{0/\var i, \facewkn{\var i}}^{T[\id, \star]} = w$ by discreteness of $T$.
	
	If $P \dsub \gamma = \eset$, then $w \psub{0/\var i, \facewkn{\var i}}^\Omega = w \psub{0/\var i, \facewkn{\var i}}^{A} = w$ by discreteness of $A$.
\end{proof}

\section{Discreteness and cohesion}
In this section, we consider the interaction between discreteness and cohesion. In the first subsection, we characterize discrete contexts as those that are in the image of the discrete functor $\cohdisc$, or equivalently in the image of $\flat$. In the second one, we show that $\coshp$ preserves discreteness. In the rest of the section we are concerned with making things discrete by quotienting out paths. For a context $\Gamma$, we will define a discrete context $\quotshp \Gamma$ and a substitution $\inquotshp : \Gamma \to \quotshp \Gamma$ into it. This will enable us to finally construct $\cohpi$. For a type $\Gamma \sez T \type$, we will define a discrete type $\Gamma \sez \quotshp T \dtype$ and a mapping $\hatinquotshp : \Tm(\Gamma, T) \to \Tm(\Gamma, \quotshp T)$. This is a prerequisite for defining existential types.
\begin{remark}
	Note that in models of HoTT, this process of forcing a type to be fibrant is ill-behaved in the sense that it does not commute with substitution: we will typically not have $\quotshp (T[\sigma]) = (\quotshp T)[\sigma]$. We will show that our shape operator does commute with substitution. This is likely related to the fact that all our horn inclusions are epimorphisms (levelwise surjective presheaf maps), so that forcing something to be discrete is an operation that transforms presheaves locally, not globally. Put differently: if a type is fibrant in HoTT, we obtain transport functions which allow us to move things around and derive a contradiction (see \cite{nlab:fibrant-replacement} for details). Discrete types however do not provide any transport or composition operations.
\end{remark}
\subsection{Discrete contexts and the discrete functor}
\begin{proposition}\label{thm:discrete-contexts-and-cohesion}
	For a context $\Gamma \in \Psh(\cat P)$, the following are equivalent:
	\begin{enumerate}
		\item $\Gamma$ is discrete,
		\item $\Gamma$ is isomorphic to $\cohdisc \Theta$ for some cubical set $\Theta \in \Psh(\cat Q)$,
		\item The substitution $\kappa : \flat \Gamma \to \Gamma$ is an isomorphism.
	\end{enumerate}
\end{proposition}
\begin{proof}
	\begin{description}
		\item[$1 \Rightarrow 3$.] Assume that $\Gamma$ is discrete. We show that $\kappa : \flat \Gamma \to \Gamma$ is an isomorphism. Pick a defining substitution $\gamma : \DSub W \Gamma$. Because $\Gamma$ is discrete, $\gamma$ is degenerate in every path dimension, i.e. it factors over $\varsigma : W \to \shp W$, say $\gamma = \gamma' \varsigma$. Then we have $\fpshadj{\shp}(\gamma') : \DSub{W}{\flat \Gamma}$ and moreover $\kappa \circ \fpshadj{\shp}(\gamma') = \fpsh \varsigma \circ \fpshadj{\shp}(\gamma') = \gamma' \circ \varsigma = \gamma$.
		
		\item[$3 \Rightarrow 2$.] Note that $\flat \Gamma = \cohdisc \cohfget \Gamma$.
		
		\item[$2 \Rightarrow 1$.] It suffices to prove that $\cohdisc \Theta$ is discrete. Pick some $\fpshadj \cohpi(\theta) : \DSub{(W, \ctxpath{\var i})}{\cohdisc \Theta}$. Then we have $\theta : \DSub{\cohpi (W, \ctxpath{\var i}) = \cohpi W}{\Theta}$ and hence $\fpshadj \cohpi(\theta) : \DSub W {\cohdisc \theta}$. Moreover, $\fpshadj \cohpi(\theta) \circ (\facewkn{\var i}) = \fpshadj \cohpi(\theta \circ \cohpi(\facewkn{\var i})) = \fpshadj \cohpi(\theta)$, showing that the picked defining substitution is degenerate along $\var i$. \qedhere
	\end{description}
\end{proof}

\subsection{The $\coshp$ functor preserves discreteness}
\begin{lemma}\label{thm:coshp-preserves-discreteness}
	For any discrete type $\Gamma \sez T \dtype$, the type $\coshp \Gamma \sez \coshp T \dtype$ is also discrete.
\end{lemma}
\begin{proof}
	Pick a defining term $(W, \ctxpath{\var i}) \Dsez \fpshadj \sharp(t) : (\coshp T) \dsub{\fpshadj \sharp(\gamma) (\facewkn{\var i})}$; we will show that it is degenerate. Note that $\fpshadj \sharp(\gamma) (\facewkn{\var i}) = \fpshadj \sharp (\gamma \circ \sharp (\facewkn{\var i})) = \fpshadj \sharp (\gamma (\facewkn{\var i}))$. Hence, we have $\sharp (W, \ctxpath{\var i}) = (\sharp W, \ctxpath{\var i}) \Dsez t : T \dsub{\gamma (\facewkn{\var i})}$. By discreteness of $T$, $t$ factors over $(\facewkn{\var i})$, i.e.\ $t = t' \psub{\facewkn{\var i}}$. Then $\fpshadj \sharp(t) = \fpshadj \sharp(t' \psub{\facewkn{\var i}}) = \fpshadj \sharp(t' \psub{\sharp (\facewkn{\var i})}) = \fpshadj \sharp(t) \psub{\facewkn{\var i}}$.
\end{proof}

\subsection{Equivalence relations on presheaves}
Before we can define $\quotshp T$, we need a little bit of theory on equivalence relations on (dependent) presheaves. We will then be able to define $\quotshp T$ straightforwardly as a quotient of $T$.
\begin{definition}
	An \textbf{equivalence relation $E$ on a presheaf} $\Gamma \in \widehat \catW$ consists of:
	\begin{itemize}
		\item For every $W \in \catW$, an equivalence relation $E_W$ on $(\DSub W \Gamma)$,
		\item So that if $E_W(\gamma, \gamma')$ and $\vfi : \PSub V W$, then $E_V(\gamma \vfi, \gamma' \vfi)$.
	\end{itemize}
	We will denote this as $E \eqrel \Gamma$.
	
	Similarly, an \textbf{equivalence relation $E$ on a dependent presheaf} $(\Gamma \sez T \type)$ consists of:
	\begin{itemize}
		\item For every $W \in \catW$ and every $\gamma : \DSub W \Gamma$, an equivalence relation $E \dsub \gamma$ on $T \dsub \gamma$,
		\item So that if $E \dsub \gamma (s, t)$ and $\vfi : V \to W$, then $E \dsub{\gamma \vfi}(s \psub \vfi, t \psub \vfi)$.
	\end{itemize}
	We will denote this as $\Gamma \sez E \eqrel T$.
\end{definition}
Given a presheaf map $\sigma : \Delta \to \Gamma$ and an equivalence relation $\Gamma \sez E \eqrel T$, we can easily define its substitution $\Delta \sez E[\sigma] \eqrel T[\sigma]$, by setting $E[\sigma]\dsub \delta = E\dsub{\sigma \delta}$. Substitution of equivalence relations obviously respects identity and composition.
\begin{lemma}
	The intersection of arbitrarily many equivalence relations on a given (dependent) presheaf, is again an equivalence relation on that presheaf. \qed
\end{lemma}
\begin{lemma}[Substitution of equivalence relations, has a right adjoint]
	Given a substitution $\sigma : \Delta \to \Gamma$ and equivalence relations $\Delta \sez F \eqrel T[\sigma]$ and $\Gamma \sez E \eqrel T$, there is an equivalence relation $\Gamma \sez \forall_\sigma F \eqrel T$ such that $E[\sigma] \subseteq F$ if and only if $E \subseteq \forall_\sigma F$.
\end{lemma}
The notation $\forall$ is related to the notation of the product type. Indeed, the product type is right adjoint to weakening, which is a special case of substitution.
\begin{proof}
	Given $W \Dsez x, y : T \dsub \gamma$, we set $(\forall_\sigma F) \dsub \gamma (x, y)$ if and only if for every face map $\vfi : \PSub V W$ and every $\delta : \DSub V \Delta$ such that $\sigma \delta = \gamma \vfi$, we have $F \dsub \delta(x \psub \vfi, y \psub \vfi)$. The quantification over $V$ guarantees that equivalence is preserved under restriction.
	
	We now show that $E[\sigma] \subseteq F$ if and only if $E \subseteq \forall_\sigma F$.
	\begin{itemize}
		\item[$\Rightarrow$] Assume that $E[\sigma] \subseteq F$. Pick $W \Dsez x, y : T \dsub \gamma$ such that $E \dsub \gamma (x, y)$. We show that $\forall_\sigma F \dsub \gamma(x, y)$. For any $\vfi : \PSub V W$ we have $E \dsub{\gamma \vfi}(x \psub \vfi, y \psub \vfi)$ and hence for any $\delta : \DSub V \Delta$ such that $\sigma \delta = \gamma \vfi$, we have $E[\sigma] \dsub{\delta}(x \psub \vfi, y \psub \vfi)$, implying $F \dsub \delta(x \psub \vfi, y \psub \vfi)$.
		
		\item[$\Leftarrow$] Assume that $E \subseteq \forall_\sigma F$. Pick $W \Dsez x, y : T[\sigma] \dsub \delta$ and assume that $E[\sigma]\dsub \delta(x, y)$. We show that $F \dsub \delta(x, y)$. Clearly, we have $\forall_\sigma F \dsub{\sigma \delta}(x, y)$. Instantiating $\vfi$ with $\id$, we can conclude $F \dsub{\delta}(x, y)$. \qedhere
	\end{itemize}
\end{proof}
\begin{lemma}[Applying a lifted functor to an equivalence relation, has a right adjoint]
	Assume a functor $K : \catV \to \catW$ and equivalence relations $\Gamma \sez_{\widehat \catW} E \eqrel T$ and $\fpsh K \Gamma \sez_{\widehat \catV} F \eqrel \fpsh K T$. There is an equivalence relation $\Gamma \sez_{\widehat \catW} \forall_{K} F \eqrel T$ such that $\fpsh K E \subseteq F$ if and only if $E \subseteq \forall_{K} F$. Here, $\fpsh K E$ is defined by $\fpsh K E \dsub{\fpshadj K (\gamma)}(\fpshadj K(x), \fpshadj K(y)) = E \dsub \gamma(x, y)$.
\end{lemma}
\begin{proof}
	The idea is entirely the same. Given $W \Dsez_{\widehat \catW} x, y : T \dsub \gamma$, we set $\forall_K F \dsub \gamma(x, y)$ if and only if for every $V \in \cat V$ and every face map $\vfi : \PSub{KV}{W}$, we have $F \dsub{\fpshadj K(\gamma \vfi)}(\fpshadj K(x \psub \vfi), \fpshadj K(y \psub \vfi))$. The quantification over $V$ guaranties that equivalence is preserved under restriction.
	
	We now show that $\fpsh K E \subseteq F$ if and only if $E \subseteq \forall_K F$.
	\begin{itemize}
		\item[$\Rightarrow$] Assume that $\fpsh K E \subseteq F$. Pick $W \Dsez_{\widehat \catW} x, y : T \dsub \gamma$ and assume that $E \dsub \gamma(x, y)$. We show that $\forall_K F \dsub \gamma(x, y)$. For any $\vfi : \PSub{KV}{W}$, we have $E \dsub{\gamma \vfi}(x \psub \vfi, y \psub \vfi)$, i.e. $\fpsh K E \dsub{\fpshadj K(\gamma \vfi)}(\fpshadj K(x \psub \vfi), \fpshadj K(y \psub \vfi))$, implying $F \dsub{\fpshadj K(\gamma \vfi)}(\fpshadj K(x \psub \vfi), \fpshadj K(y \psub \vfi))$.
		
		\item[$\Leftarrow$] Assume that $E \subseteq \forall_{\fpsh K} F$. Pick $V \Dsez_{\widehat \catV} \fpshadj K(x), \fpshadj K(y) : \fpsh K T \dsub{\fpshadj K(\gamma)}$ such that $\fpsh K E \dsub{\fpshadj K(\gamma)}(\fpshadj K(x), \fpshadj K(y))$, i.e.\ $E \dsub \gamma(x, y)$. Then $\forall_K F \dsub \gamma(x, y)$. Instantiating $\vfi = \id : \PSub{KV}{KV}$, we have $F \dsub{\fpshadj K(\gamma)}(\fpshadj K(x), \fpshadj K(y))$. \qedhere
	\end{itemize}
\end{proof}
\begin{definition}
	If $E \eqrel \Gamma$, then we define the context $\Gamma/E$ by setting $(\DSub{W}{\Gamma/E}) = (\DSub{W}{\Gamma})/E_W$ and $\overline\gamma \circ \vfi = \overline{\gamma \circ \vfi}$, which is well-defined by virtue of the second bullet in the definition of an equivalence relation.
	
	If $\Gamma \sez E \eqrel T$, then we define $\Gamma \sez T/E \type$ by setting $(T/E) \dsub \gamma = T \dsub \gamma/E \dsub{\gamma}$ and $\overline{x} \psub \vfi = \overline{x \psub \vfi}$.
\end{definition}
One easily checks that $(T/E)[\sigma] = T[\sigma]/E[\sigma]$.

\subsection{Discretizing contexts and the functor $\cohpi$}\label{sec:def-cohpi}
In this section, our aim is to construct for any context $\Gamma$, a discrete context $\quotshp \Gamma$ with a substitution $\inquotshp : \Gamma \to \quotshp \Gamma$ such that any substitution $\tau : \Gamma \to \Gamma'$ to a discrete context $\Gamma'$, factors uniquely over $\inquotshp$. The effect of applying $\quotshp$ will be that we are contracting every path to a point (and more generally, that we are contracting every cell in all its path dimensions). When we postcompose with $\cohfget$, we obtain our desired left adjoint $\cohpi = \cohfget \quotshp$ of $\cohdisc$.

\subsubsection{Discretizing contexts}
Our approach is quite straightforward: we simply divide out the least equivalence relation that makes the quotient discrete. Recall that a context $\Gamma$ is discrete iff for every $\gamma : \DSub{(W, \ctxpath{\var i})}{\Gamma}$, we have that $\gamma = \gamma (0 / \var i, \facewkn{\var i})$.
\begin{definition}
	Let the \textbf{shape equivalence relation} $\sheq$ on $\Gamma$ be the least equivalence relation such that $\sheq(\gamma, \gamma(0 / \var i, \facewkn{\var i}))$ for any $\gamma : \DSub{(W, \ctxpath{\var i})}{\Gamma}$. Then we define the \textbf{shape quotient} of $\Gamma$ as $\quotshp \Gamma = \Gamma/\sheq$. Given a substitution $\sigma : \Gamma \to \Gamma'$, we define $\quotshp \sigma : \quotshp \Gamma \to \quotshp \Gamma'$ by setting $\quotshp \sigma \circ \overline \gamma = \overline{\sigma \circ \gamma}$. This constitutes a functor $\quotshp : \widehat{\bpcubecat} \to \widehat{\bpcubecat}$.
	
	We define a natural transformation $\inquotshp : \Id \to \quotshp$ by $\inquotshp \circ \gamma := \overline \gamma$.
\end{definition}
It is easy to see that any substitution $\tau$ from $\Gamma$ into a discrete context, factors uniquely over $\inquotshp : \Gamma \to \quotshp \Gamma$. In particular, $\quotshp \sigma$ is well-defined.
\begin{lemma}\label{thm:kappa-injective}
	For any context $\Gamma \in \widehat \bpcubecat$, the substitution $\kappa : \flat \Gamma \to \Gamma$ is an injective presheaf map, meaning that $\kappa \circ \loch : (\DSub{W}{\flat \Theta}) \to (\DSub{W}{\Theta})$ is injective for every $W$.
\end{lemma}
\begin{proof}
	Given $\fpshadj \shp(\gamma) : \DSub{W}{\flat \Gamma}$, we have $\kappa \circ \fpshadj \shp(\gamma) = \gamma \varsigma$, where $\varsigma : \PSub{W}{\shp W}$ is easily seen to have a right inverse.
\end{proof}
\begin{lemma}
	Any substitution $\tau : \Gamma \to \Theta$ from a discrete context $\Gamma$ to any context $\Theta$, factors uniquely over $\kappa : \flat \Theta \to \Theta$. Hence, $\quotshp$ is left adjoint to $\flat$.
\end{lemma}
\begin{proof}
	\begin{description}
		\item[Existence.] Pick $\gamma : \DSub{W}{\Gamma}$. Since $\Gamma$ is discrete, $\gamma$ factors over $\varsigma : \PSub{W}{\shp W}$ as $\gamma = \gamma' \varsigma$. Then we have $\tau \gamma' : \DSub{\shp W}{\Theta}$ and hence $\fpshadj \shp(\tau \gamma') : \DSub{W}{\flat \Theta}$. So we define $\tau' : \Gamma \to \flat \Theta$ by setting $\tau' \gamma' \varsigma = \fpshadj \shp(\tau \gamma')$. To see that this is natural:
		\begin{equation}
			\tau' \gamma' \varsigma \vfi = \tau' \gamma' (\shp \vfi) \varsigma = \fpshadj \shp(\tau \gamma' (\shp \vfi)) = \fpshadj \shp(\tau \gamma') \vfi.
		\end{equation}
		
		\item[Uniqueness.] This follows from the fact that $\kappa : \flat \Theta \to \Theta$ is an injective presheaf map, see \cref{thm:kappa-injective}.
		
		\item[Adjointness.] Substitutions $\quotshp \Gamma \to \Theta$ factor uniquely (and thus naturally) over $\kappa : \flat \Theta \to \Theta$ and are thus in natural correspondence with substitutions $\quotshp \Gamma \to \flat \Theta$. On the other hand, substitutions $\Gamma \to \flat \Theta$ factor uniquely (and thus naturally) over $\inquotshp : \Gamma \to \quotshp \Gamma$ and are thus also in natural correspondence with substitutions $\quotshp \Gamma \to \flat \Theta$. This proves the adjunction. \qedhere
	\end{description}
\end{proof}
Since $\kappa$ is an injective presheaf map and $\inquotshp$ is clearly a surjective presheaf map, we use the following notations. If $\sigma : \Gamma \to \flat \Theta$, then we write $\kappa \sigma \inquotshp\inv$ for $\alpha\inv_{\quotshp \dashv \flat}(\sigma) : \quotshp \Gamma \to \Theta$. Conversely, if $\tau : \quotshp \Gamma \to \Theta$, we write $\kappa\inv \tau \inquotshp$ for $\alpha_{\quotshp \dashv \flat}(\tau) : \Gamma \to \flat \Theta$. Thus, we have unit $\kappa\inv \inquotshp : \Id \to \flat \quotshp$ and co-unit $\kappa \inquotshp \inv : \quotshp \flat \to \Id$.

\subsubsection{The functor $\cohpi$}
As $\quotshp$ is left adjoint to $\flat$, we could define $\shp$ as $\quotshp$. However, this does not give us a decomposiiton $\shp = \cohdisc \cohpi$ or the property $\flat \shp = \shp$. Instead, we define $\cohpi := \cohfget \quotshp \dashv \flat \cohcodisc = \cohdisc$. By consequence, we have $\shp = \cohdisc \cohpi = \flat \quotshp \dashv \flat \sharp = \flat$. Since both $\shp$ and $\quotshp$ are now left adjoint to $\flat$, we have $\kappa \quotshp : \shp \cong \quotshp$.
By \cref{thm:uniqueness-of-nattrans-psh}, we have $\varsigma = (\kappa \quotshp)\inv \inquotshp : \Id \to \shp$ and $\bar \varsigma = \cohfget \varsigma \cohdisc : \Id \to \bar \shp$.
We will maximally avoid to inspect the definition of $\cohpi$; hence we will avoid explicit use of $\quotshp$ and $\inquotshp$.

\subsection{Discretizing types}
If $\shp$ were a morphism of CwFs, then from a type $\Gamma \sez T \type$, we could define a type $\Gamma \sez (\shp T)[\varsigma] \dtype$, but unfortunately this is meaningless. Because we need that operation nonetheless, we will define it explicitly in this section. The approach is the same as for contexts: we simply obtain $\Gamma \sez \quotshp T \dtype$ from $T$ by dividing out the least equivalence relation $\sheq^T$ that makes $T$ discrete. The main obstacle is that we want this operation to commute with substitution, i.e. that $\sheq^T$ commutes with substitution. Here, once more, we will need the existence of coherence squares, as is evident from the following example, where we fail to prove the same result for the category of reflexive graphs $\widehat{\RGcat}$.

\subsubsection{The shape equivalence relation}
Before we proceed $\widehat \bpcubecat$, we will try to define the shape operation in $\widehat \RGcat$.
\begin{example}\label{eg:shape-equivalence-relation}
	Given any type $\Gamma \sez T \type$, the \textbf{shape equivalence relation} $\sheq^T$ is the smallest equivalence relation such that $\sheq^T \dsub{\gamma (\facewkn{\var i})}(p, p \psub{0 / \var i, \facewkn{\var i}})$ for every edge $(W, \ctxedge{\var i}) \Dsez p : T \dsub{\gamma (\facewkn{\var i})}$.
	
	Again, dividing out $\sheq^T$ is precisely what is needed to make $T$ discrete. We now try to prove the following (false) claim: The shape equivalence relation respects substitution: $\sheq^T[\sigma] = \sheq^{T[\sigma]}$.
	\begin{proof}[Non-proof]
		Pick a substitution $\sigma : \Delta \to \Gamma$ and a type $\Gamma \sez T \type$. We try to prove both inclusions.
		\begin{itemize}
			\item[$\supseteq$] It suffices to show that $\sheq^T[\sigma]$ satisfies the defining property of $\sheq^{T[\sigma]}$. Pick an edge $(\ctxedge{\var i}) \Dsez p : T[\sigma] \dsub{\delta (\facewkn{\var i})}$. We have to show that $\sheq^T[\sigma]\dsub{\delta (\facewkn{\var i})} (p, p \psub{0 / \var i, \facewkn{\var i}})$. After composing $\sigma$ and $\delta (\facewkn{\var i})$, this follows immediately from the definition of $\sheq^T$.
			
			\item[$\subseteq$] We will try to prove the equivalent statement that $\sheq^T \subseteq \forall_\sigma \sheq^{T[\sigma]}$. It suffices to show that the right hand side satisfies the defining property of $\sheq^T$. Pick an edge $(\ctxedge{\var i}) \Dsez p : T \dsub{\gamma (\facewkn{\var i})}$. In order to show that $\forall_\sigma \sheq^{T[\sigma]} \dsub{\gamma (\facewkn{\var i})} (p, p \psub{0 / \var i, \facewkn{\var i}})$, we need to show for every $\vfi : \PSub{V}{(\ctxedge{\var i})}$ and every $\delta : \DSub{V}{\Delta}$ such that $\sigma \delta = \gamma (\facewkn{\var i}) \vfi$, that $\sheq^{T[\sigma]}\dsub{\delta}(p \psub \vfi, p \psub{0 / \var i, \facewkn{\var i}} \psub \vfi)$. In the case where $V = (\ctxedge{\var i})$, a problem arises because we do not know that $\delta$ is degenerate. In fact, we can give a counterexample. \qedhere
		\end{itemize}
	\end{proof}
	\begin{proof}[Counterexample]
		Let $\Gamma \cong ()$ and $\Delta \cong \yoneda(\ctxedge{\var i})$: a reflexive graph with two nodes $\delta, \delta' : \DSub{()}{\Delta}$ and a single non-trivial edge $\delta_1$ (\cref{fig:shape-equivalence-relation}). Let $\sigma$ be the unique substitution $\Delta \to \Gamma$. Consider the type $\Gamma \sez T \type$ consisting of two nodes $x$ and $y$ connected by two non-trivial edges $p$ and $q$. We get the setup shown in \cref{fig:shape-equivalence-relation}, which we briefly discuss here.
		
		Both $\delta$ and $\delta'$ are mapped to $\gamma$ under $\sigma$; hence in $T[\sigma]$, they both get a copy of $x$ and $y$. The degenerate edges $\delta \psub{\facewkn{\var i}}$ and $\delta'\psub{\facewkn{\var i}}$, as well as the edge $\delta_1$, are mapped to $\gamma \psub{\facewkn{\var i}}$. Hence in $T[\sigma]$, both $p$ and $q$ are tripled. The degenerate edges, too, are tripled, but you see only one copy of them, as the other two are still degenerate.
		
		All four edges in $T$ live above the degenerate edge $\gamma \psub{\facewkn{\var i}}$. Hence, when dividing out $\sheq^T$, they are all contracted to the degenerate edge at their source. Only a point remains.
		
		In $T[\sigma]$, only the vertical edges and the constant ones, live above degenerate edges in $\Delta$. Hence, only those are contracted. The horizontal and diagonal edges are preserved.
		
		In $(T/\sheq^T)[\sigma]$ (which is easily checked to be equal to $T[\sigma]/\sheq^T[\sigma]$), by contrast, we only have a single horizontal edge. Indeed: we get two copies of $\overline x$, and three copies of its constant edge, two of which are still degenerate.
		
		This is an example where $\sheq^T[\sigma] \neq \sheq^{T[\sigma]}$.
	\end{proof}
	The situation would have been different, had we had coherence squares. Indeed, in that case, we would have constant squares on $p$ and $q$ in $T$, living above the constant square at $\gamma$. These would produce squares filling up the front and back of $T[\sigma]$, living above the constant square at $\delta_1$. We could make $\sheq^X$ contract not just edges above degenerate edges, but also squares living above (partially) degenerate squares. Then both filling squares, as well as their diagonals, would be contracted and we would end up with just a single horizontal edge in $T[\sigma]/\sheq^{T[\sigma]}$.
\end{example}
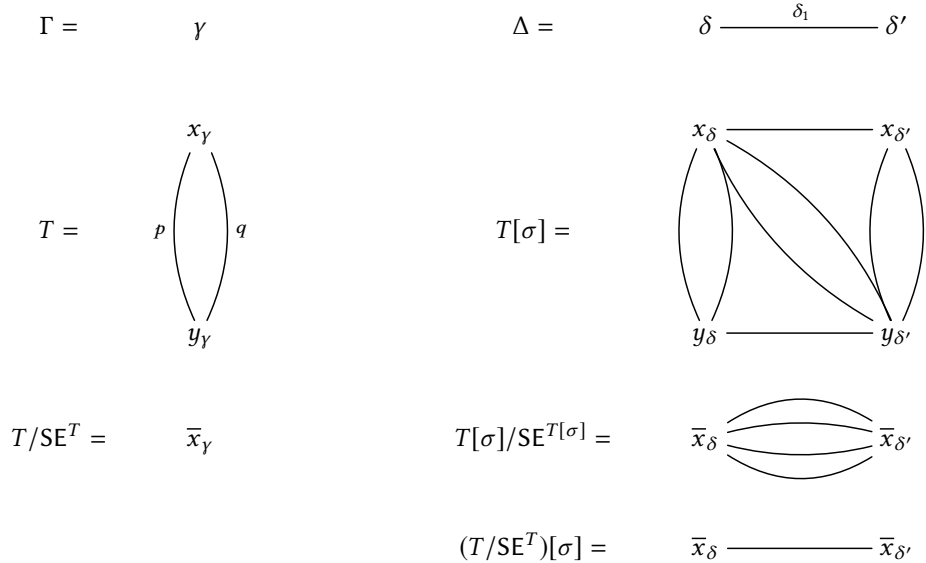
\begin{figure}[htb]
	\begin{equation*}
		\xymatrix{
				\Gamma =
				& {\gamma}
				&&& \Delta =
				& {\delta} \ar@{-}[rr]^{\delta_1}
				&& {\delta'}
				\\
				& x_\gamma \ar@{-}@/_{1em}/[dd]_p \ar@{-}@/^{1em}/[dd]^q
				&&& & x_\delta \ar@{-}[rr]
					\ar@{-}@/_{1em}/[dd] \ar@{-}@/^{1em}/[dd]
					\ar@{-}@/_{1em}/[ddrr] \ar@{-}@/^{1em}/[ddrr]
				&& x_{\delta'} \ar@{-}@/_{1em}/[dd] \ar@{-}@/^{1em}/[dd]
				\\
				T = & &&& T[\sigma] =
				\\
				& y_\gamma
				&&& & y_\delta \ar@{-}[rr]
				&& y_{\delta'}
				\\
				T/\sheq^T =
				& \overline x_\gamma
				&&& T[\sigma]/\sheq^{T[\sigma]} =
				& \overline{x}_\delta
					\ar@{-}@/^{1.5em}/[rr]
					\ar@{-}@/^/[rr]
					\ar@{-}@/_/[rr]
					\ar@{-}@/_{1.5em}/[rr]
				&& \overline{x}_{\delta'}
				\\
				& &&& (T/\sheq^T)[\sigma] =
				& {\overline x_\delta} \ar@{-}[rr]
				&& \overline x_{\delta'}
		}
	\end{equation*}
	\caption{Setup from the counterexample in \cref{eg:shape-equivalence-relation}. Degenerate edges are not shown, and nodes of types are indexed with the context nodes they live above, in order to distinguish duplicates.}
	\label{fig:shape-equivalence-relation}
\end{figure}
\begin{remark}\label{remark:annotated-restriction}
	Note, in \cref{fig:shape-equivalence-relation}, that the horizontal edges of $T[\sigma]$ arise from reflexive edges in $T$, yet they are themselves not reflexive. Therefore, it is important to distinguish between $x \psub{\facewkn{\var i}}^T$ and $x \psub{\facewkn{\var i}}^{T[\sigma]}$. Every edge that can be written as $x \psub{\facewkn{\var i}}^{T[\sigma]}$ can also be written as $x \psub{\facewkn{\var i}}^T$, but the converse does not hold as exhibited by the horizontal edges.
\end{remark}
\begin{definition}
	Given any type $\Gamma \sez T \type$, the \textbf{shape equivalence relation} $\sheq^T$ is the smallest equivalence relation on $T$ such that for any $(W, \ctxpath{\var i}) \Dsez p : T \dsub{\gamma(\facewkn{\var i})}$, we have $\sheq^T(p, p \psub{0/\var i, \facewkn{\var i}})$.
\end{definition}
\begin{lemma}
	The shape equivalence relation respects substitution: $\sheq^T[\sigma] = \sheq^{T[\sigma]}$.
\end{lemma}
\begin{proof}
	Pick a substitution $\sigma : \Delta \to \Gamma$ and a type $\Gamma \sez T \type$. We prove both inclusions.
	\begin{itemize}
		\item[$\supseteq$] It suffices to show that $\sheq^T[\sigma]$ satisfies the defining property of $\sheq^{T[\sigma]}$. Pick a path $(W, \ctxpath{\var i}) \Dsez p : T[\sigma] \dsub{\delta (\facewkn{\var i})}$. We have to show that $\sheq^T[\sigma]\dsub{\delta (\facewkn{\var i})} (p, p \psub{0 / \var i, \facewkn{\var i}})$. After composing $\sigma$ and $\delta (\facewkn{\var i})$, this follows immediately from the definition of $\sheq^T$.
		
		\item[$\subseteq$] We prove the equivalent statement that $\sheq^T \subseteq \forall_\sigma \sheq^{T[\sigma]}$. It suffices to show that the right hand side satisfies the defining property of $\sheq^T$. Pick a path $(W, \ctxpath{\var i}) \Dsez p : T \dsub{\gamma (\facewkn{\var i})}$. In order to show that $\forall_\sigma \sheq^{T[\sigma]} \dsub{\gamma (\facewkn{\var i})} (p, p \psub{0 / \var i, \facewkn{\var i}})$, we need to show for every $\vfi : \PSub{V}{(W, \ctxpath{\var i})}$ and every $\delta : \DSub{V}{\Delta}$ such that $\sigma \delta = \gamma (\facewkn{\var i}) \vfi$, that $\sheq^{T[\sigma]}\dsub{\delta}(p \psub \vfi, p \psub{0 / \var i, \facewkn{\var i}} \psub \vfi)$. We make a case distinction based on $\var i \psub \vfi$.
		\begin{description}
			\item[$\var i \psub \vfi = 0$] Then $\vfi = (0/\var i)\psi$ for some $\psi : \PSub V W$. Then we have to prove $\sheq^{T[\sigma]}\dsub \delta (p \psub{0/\var i} \psub \psi, p \psub{0/\var i} \psub \psi)$ which holds by reflexivity.
			
			\item[$\var i \psub \vfi = 1$] Then $\vfi = (1/\var i)\psi$ for some $\psi : \PSub V W$. Then we have to prove $\sheq^{T[\sigma]}\dsub \delta (p \psub{1/\var i} \psub \psi, p \psub{0/\var i} \psub \psi)$.
			Without loss of generality, we may assume that $\var i \not\in V$.
			Then we have a path $(V, \ctxpath{\var i}) \Dsez p \psub{\psi, \var i / \var i} : T\dsub{\gamma (\facewkn{\var i})(\psi, \var i/ \var i)}$. Now we have
			\begin{equation}
				\gamma (\facewkn{\var i})(\psi, \var i/ \var i) = \gamma (\facewkn{\var i}) (1/\var i) \psi (\facewkn{\var i}) = \sigma \delta (\facewkn{\var i}).
			\end{equation}
			Hence, we have $(V, \ctxpath{\var i}) \Dsez p \psub{\psi, \var i / \var i} : T[\sigma]\dsub{\delta(\facewkn{\var i})}$. Applying the definition of $\sheq^{T[\sigma]}$, we have $\sheq^{T[\sigma]} \dsub{\delta(\facewkn{\var i})} (p \psub{\psi, \var i / \var i}, p \psub{\psi, 0 / \var i, \facewkn{\var i}})$. Subsequently restricting by $(1/\var i) : \PSub{V}{(V, \ctxpath{\var i})}$ yields the desired result.
			
			\item[$\var i \psub \vfi \in V$] Without loss of generality, we may assume that $(W, \ctxpath{\var i})$ and $V$ are disjoint. Write $\var k = \var i \psub \vfi$ (and note that $\var k$ may be either a bridge or a path variable). Then $\vfi$ factors as $(\psi, \var i / \var i)(\var k / \var i) = (\var k / \var i)(\psi, \var k / \var k)$ for some $\psi : \PSub V W$. We have to prove $\sheq^{T[\sigma]} \dsub \delta(p \psub{\psi, \var i / \var i} \psub{\var k / \var i}, p \psub{\psi, \var 0 / \var i, \facewkn{\var i}} \psub{\var k/ \var i})$. This follows by restricting $\sheq^{T[\sigma]} \dsub{\delta(\facewkn{\var i})} (p \psub{\psi, \var i / \var i}, p \psub{\psi, 0 / \var i, \facewkn{\var i}})$, derived above, by $(\var k / \var i)$. \qedhere
		\end{description}
	\end{itemize}
\end{proof}
\begin{lemma}
	For any type $\Gamma \sez T \type$, we have $\sharp \sheq^T \subseteq \sheq^{\sharp T}$ and $\coshp \sheq^T = \sheq^{\coshp T}$.
\end{lemma}
This is in line with the intuition that $\sharp T$ has more paths than $T$, whereas $\coshp T$ has the same path relation as $T$.
\begin{proof}
	We first prove $\sharp \sheq^T \subseteq \sheq^{\sharp T}$.
	\begin{itemize}
		\item[$\subseteq$] We prove the equivalent statement that $\sheq^T \subseteq \forall_\flat \sheq^{\sharp T}$. It suffices to show that the right hand side satisfies the defining property of $\sheq^T$. Pick a path $(W, \ctxpath{\var i}) \Dsez p : T \dsub{\gamma (\facewkn{\var i})}$. In order to show that $\forall_\flat \sheq^{\sharp T} \dsub{\gamma (\facewkn{\var i})} (p, p \psub{0 / \var i, \facewkn{\var i}})$, we need to show for every $\vfi : \PSub{\flat V}{(W, \ctxpath{\var i})}$ that $\sheq^{\sharp T}\dsub{\fpshadj \flat (\gamma (\facewkn{\var i}) \vfi)}(\fpshadj \flat(p \psub \vfi), \fpshadj \flat(p \psub{0 / \var i, \facewkn{\var i}} \psub \vfi))$. We make a case distinction based on $\var i \psub \vfi$.
		\begin{description}
			\item[$\var i \psub \vfi = 0$] Then $\vfi = (0/\var i)\psi$ for some $\psi : \PSub {\flat V} W$. Then we have to prove $\sheq^{\sharp T}\dsub{\fpshadj \flat (\gamma \psi)}(\fpshadj \flat(p \psub{0/\var i} \psub \psi), \fpshadj \flat(p \psub{0 / \var i} \psub \psi))$ which holds by reflexivity.
			
			\item[$\var i \psub \vfi = 1$] Then $\vfi = (1/\var i)\psi$ for some $\psi : \PSub {\flat V} W$. Then we have to prove $\sheq^{\sharp T}\dsub{\fpshadj \flat (\gamma \psi)}(\fpshadj \flat(p \psub{1/\var i} \psub \psi), \fpshadj \flat(p \psub{0 / \var i} \psub \psi))$.
			Without loss of generality, we may assume that $\var i \not\in V$.
			Then we have a bridge $(\flat V, \ctxbrid{\var i}) \Dsez p \psub{\psi, \var i / \var i} : T\dsub{\gamma \psi}$ and hence a path $(V, \ctxpath{\var i}) \Dsez \fpshadj \flat(p \psub{\psi, \var i / \var i}) : (\sharp T) \dsub{\fpshadj \flat(\gamma \psi)}$. Applying the definition of $\sheq^{\sharp T}$, we have $\sheq^{\sharp T} \dsub{\fpshadj \flat(\gamma \psi)} (\fpshadj \flat(p \psub{\psi, \var i / \var i}), \fpshadj \flat (p \psub{\psi, 0 / \var i, \facewkn{\var i}}))$. Subsequently restricting by $(1/\var i) : \PSub{V}{(V, \ctxpath{\var i})}$ yields the desired result.
			
			\item[$\var i \psub \vfi \in V$] Analogous.
		\end{description}
	\end{itemize}
	We now prove $\coshp \sheq^T = \sheq^{\coshp T}$ by proving both inclusions.
	\begin{itemize}
		\item[$\supseteq$] It suffices to show that $\coshp \sheq^T$ satisfies the defining property of $\sheq^{\coshp T}$. Pick a path $(W, \ctxpath{\var i}) \Dsez \fpshadj \sharp(p) : \coshp T \dsub{\fpshadj \sharp(\gamma) \circ (\facewkn{\var i})}$. We have to show that $\coshp \sheq^T \dsub{\fpshadj \sharp(\gamma) \circ (\facewkn{\var i})} (\fpshadj \sharp(p), \fpshadj \sharp(p) \psub{0/\var i, \facewkn{\var i}})$, i.e. $\sheq^T \dsub{\gamma (\facewkn{\var i})} (p, p \psub{0/\var i, \facewkn{\var i}})$. Note that $p$ has type $(\sharp W, \ctxpath{\var i}) \Dsez p : T \dsub{\gamma (\facewkn{\var i})}$, i.e.\ it is a path. So this follows immediately from the definition of $\sheq^T$. (We could not prove the inclusion $\sharp \sheq^T \supseteq \sheq^{\sharp T}$ because we would have to apply $\flat$ to the primitive context, finding that $p$ is only a bridge.)
		\item[$\subseteq$] The proof for $\sharp$ can be copied almost verbatim. \qedhere
	\end{itemize}
\end{proof}

\subsubsection{The shape of a type}
\begin{definition}
	Given a type $\Gamma \sez T \type$, we define the discrete type $\Gamma \sez \quotshp T \dtype$ as $\quotshp T = T / \sheq^T$.
\end{definition}
This definition commutes with substitution, as $(\quotshp T)[\sigma] = (T/\sheq^T)[\sigma] = T[\sigma]/\sheq^T[\sigma] = T[\sigma]/\sheq^{T[\sigma]} = \quotshp(T[\sigma])$.
\begin{proposition}\label{thm:hatinquotshp}
	Given $\Gamma \sez T \type$, we have
	\begin{equation}
		\begin{array}{l l l}
			\Gamma \sez \hatinquotshp : T \to \quotshp T, \\
			\sharp \Gamma \sez \sharp \hatinquotshp : \sharp T \to \sharp \quotshp T, &
			\qquad &
			\sharp \Gamma \sez (\hatinquotshp \sharp \hatinquotshp\inv) : \sharp \quotshp T \to \quotshp \sharp T, \\
			\coshp \Gamma \sez \coshp \hatinquotshp : \coshp T \to \coshp \quotshp T, &
			\qquad &
			\coshp \Gamma \sez \coshp \quotshp T = \quotshp \coshp T \type.
		\end{array}
	\end{equation}
	naturally in $\Gamma$.
	We have commutative diagrams
	\begin{equation}
		\xymatrix{
			\sharp T \ar[d]_{\sharp \hatinquotshp} \ar[r]^{\hatinquotshp}
			& {\quotshp \sharp T} \\
			{\sharp \quotshp T} \ar[ru]_{(\hatinquotshp \sharp \hatinquotshp\inv)}
		} \qquad
		\xymatrix{
			\coshp T \ar[d]_{\coshp \hatinquotshp} \ar[r]^{\hatinquotshp}
			& {\quotshp \coshp T} \\
			{\coshp \quotshp T} \ar@{=}[ru]
		}
	\end{equation}
	If $T$ is discrete, then $\hatinquotshp$, $\sharp \hatinquotshp$ and $\coshp \hatinquotshp$ are also invertible.
\end{proposition}
\begin{proof}
	Recall from \cref{sec:psh-pi-types} that a function $\Gamma \sez f : \Pi A B$ is fully determined if we know $W \Dsez f \dsub \gamma \cdot a : B \dsub{\gamma, a}$ for every $W$, $\gamma : \DSub W \Gamma$ and $W \Dsez a : A \dsub \gamma$.
	
	We set $\hatinquotshp \dsub \gamma \cdot t := \overline t$. This is well-defined because $\hatinquotshp \dsub{\gamma \vfi} \cdot (t \psub \vfi) = \overline{t \psub \vfi} = \overline t \psub \vfi = (\hatinquotshp \dsub \gamma \cdot t) \psub \vfi$.
	
	We define $\sharp \hatinquotshp := \lambda(\ftrtm{\sharp}{(\ap\,\hatinquotshp)})$ and $\coshp \hatinquotshp := \lambda(\ftrtm{\coshp}{(\ap\,\hatinquotshp)})$. Then we have (twice using the fact that labels can be ignored)
	\begin{align*}
		(\sharp \hatinquotshp) \dsub{\fpshadj \flat(\gamma)} \cdot \fpshadj \flat(t)
		&= \ftrtm{\sharp}{(\ap\,\hatinquotshp)} \dsub{\fpshadj \flat(\gamma), \fpshadj \flat(t)}
		= \ftrtm{\sharp}{(\ap\,\hatinquotshp)} \dsub{\fpshadj \flat(\gamma, t)}
		= \fpshadj \flat(\ap\,\hatinquotshp \dsub{\gamma, t}) \\
		&= \fpshadj \flat(\hatinquotshp \dsub \gamma \cdot t)
		= \fpshadj \flat(\overline t) = \overline{\fpshadj \flat(t)},
	\end{align*}
	i.e. $(\sharp \hatinquotshp) \dsub \gamma \cdot t = \overline t$. Similarly, we find $(\coshp \hatinquotshp) \dsub \gamma \cdot t = \overline t$.
	
	Note that $\sharp \quotshp T = \sharp(T/\sheq^T) = \sharp T / \sharp \sheq^T$ and $\quotshp \sharp T = \sharp T / \sheq^{\sharp T}$. Since $\sharp \sheq^T \subseteq \sheq^{\sharp T}$, we can define $(\hatinquotshp \sharp \hatinquotshp\inv) \dsub \gamma \cdot \overline t := \overline t$. Then the first commuting diagram is clear.
	
	Also note that $\coshp \quotshp T = \coshp(T/\sheq^T) = \coshp T/\coshp \sheq^T = \coshp T/\sheq^{\coshp T} = \quotshp \coshp T$. The second commuting diagram is then also clear.
	
	Now suppose that $T$ is discrete. We show that $\tmshp\loch$ is an isomorphism by showing that every equivalence class of $\sheq^T$ is a singleton. This is equivalent to saying that $\sheq^T$ is the equality relation. Clearly, the equality relation is the weakest of all equivalence relations, so it suffices to show that the equality relation satisfies the defining property of $\sheq^T$. But that is precisely the statement that $T$ is discrete.
	
	Furthermore, if $\sheq^T$ is the equality relation, then so are $\sharp \sheq^T$ and $\coshp \sheq^T$. Hence, $\sharp \hatinquotshp$ and $\coshp \hatinquotshp$ will also be invertible.
\end{proof}
\begin{remark}
	Substitutions can be applied to the functions $\sharp \hatinquotshp$ and $\coshp \hatinquotshp$ from \cref{thm:hatinquotshp}, moving them to non-$\sharp$ or non-$\coshp$ contexts. We will omit those substitutions, writing e.g. $\Delta \sez \sharp \hatinquotshp : (\sharp T)[\sigma] \to (\sharp \quotshp T)[\sigma]$.
\end{remark}
\begin{lemma}\label{thm:elim-quotshp}
	For discrete types $T$ living in the appropriate context, we have \emph{invertible} rules
	\begin{equation}
		\binference{
			\Gamma, \var x : \quotshp S \sez t : T
		}{
			\Gamma, \var x : S \sez
			t[\wknvar x, \hatinquotshp(\var x)/\var x]
			: T[\wknvar x, \hatinquotshp(\var x)/\var x]
		}{}.
	\end{equation}
	\begin{equation}
		\binference{
			\Gamma, \ftrtm \sharp {\var x} : (\sharp \quotshp S) [\sigma] \sez t : T
		}{
			\Gamma, \ftrtm \sharp {\var x} : (\sharp S) [\sigma] \sez
			t[\wkn{\ftrtm \sharp {\var x}}, (\sharp \hatinquotshp)(\ftrtm \sharp {\var x})/\ftrtm \sharp {\var x}]
			: T[\wkn{\ftrtm \sharp {\var x}}, (\sharp \hatinquotshp)(\ftrtm \sharp {\var x})/\ftrtm \sharp {\var x}]
		}{}.
	\end{equation}
	\begin{equation}
		\binference{
			\Gamma, \ftrtm \coshp {\var x} : (\coshp \quotshp S) [\sigma] \sez t : T
		}{
			\Gamma, \ftrtm \coshp {\var x} : (\coshp S) [\sigma] \sez
			t[\wkn{\ftrtm \coshp {\var x}}, (\coshp \hatinquotshp)(\ftrtm \coshp {\var x})/\ftrtm \sharp {\var x}]
			: T[\wkn{\ftrtm \coshp {\var x}}, (\coshp \hatinquotshp)(\ftrtm \coshp {\var x})/\ftrtm \coshp {\var x}]
		}{}.
	\end{equation}
	The downward direction is each time a straightforward instance of substitution and hence natural in $\Gamma$. The inverse is then automatically also natural in $\Gamma$.
\end{lemma}
\begin{proof}
	\begin{description}
		\item[Rule 1] To show that the first rule is invertible, we pick a term $\Gamma, \var x : S \sez u : T[\wknvar x, \hatinquotshp(\var x)/\var x]$ and show that it factors over $(\wknvar x, \hatinquotshp(\var x)/ \var x) : (\Gamma, \var x : S) \to (\Gamma, \var x : \quotshp S)$. So pick a path $(W, \ctxpath{\var i}) \Dsez s : S \dsub{\gamma(\facewkn{\var i})}$ that becomes degenerate in $\quotshp S$. We have to show that $u \dsub{\gamma (\facewkn{\var i}), s} = u \dsub{\gamma (\facewkn{\var i}), s \psub{0/\var i, \facewkn{\var i}}}$. Note that both live in $T \dsub{\gamma (\facewkn{\var i}), \hatinquotshp(s)}$.
	
		Now, the defining substitution $(\gamma (\facewkn{\var i}), \hatinquotshp(s))$ is degenerate in $\var i$ because this is obvious for the first component and then degeneracy of the second component follows from discreteness of $\quotshp S$. Hence, $u \dsub{\gamma (\facewkn{\var i}), s}$ is degenerate as a defining term of type $T$, meaning that
		\begin{equation}
			u \dsub{\gamma (\facewkn{\var i}), s} = u \dsub{\gamma (\facewkn{\var i}), s} \psub{0/\var i, \facewkn{\var i}} = u \dsub{\gamma (\facewkn{\var i}), s \psub{0/\var i, \facewkn{\var i}}}.
		\end{equation}
		\item[Rule 2] Note that $(\sharp \quotshp S)[\sigma] = (\sharp(S/\sheq^S))[\sigma] = (\sharp S/\sharp \sheq^S)[\sigma] = (\sharp S)[\sigma]/(\sharp \sheq^S)[\sigma]$.
		To show that the second rule is invertible, we pick a term
		$\Gamma, \ftrtm \sharp {\var x} : (\sharp S) [\sigma] \sez u : T[\wkn{\ftrtm \sharp {\var x}}, (\sharp \hatinquotshp)(\ftrtm \sharp {\var x})/\ftrtm \sharp {\var x}]$ and show that it factors over $(\wkn{\ftrtm \sharp {\var x}}, (\sharp \hatinquotshp)(\ftrtm \sharp {\var x})/\ftrtm \sharp {\var x}) : (\Gamma, \ftrtm \sharp {\var x} : (\sharp S) [\sigma]) \to (\Gamma, \ftrtm \sharp {\var x} : (\sharp \quotshp S) [\sigma])$. To that end, we need to show that whenever $(\sharp \sheq^{S}) [\sigma] \dsub \gamma(r, s)$, we also have $u \dsub{\gamma, r} = u \dsub{\gamma, s}$. Let us write $U \dsub \gamma(r, s)$ for $u \dsub{\gamma, r} = u \dsub{\gamma, s}$. This is easily seen to be an equivalence relation on $(\sharp S)[\sigma]$. So we need to prove $(\sharp \sheq^{S}) [\sigma] \subseteq U$, or equivalently $\sheq^S \subseteq \forall_\flat \forall_\sigma U$.
		
		Let $\Delta$ be the context of $S$, i.e. $\Delta \sez S \type$ and $\sigma : \Gamma \to \sharp \Delta$. It is sufficient to show that $\forall_\flat \forall_\sigma U$ satisfies the defining property of $\sheq^S$. So pick a path $(W, \ctxpath{\var i}) \Dsez p : S \dsub{\delta(\facewkn{\var i})}$. Write $q = p \psub{0/\var i, \facewkn{\var i}}$. We have to show $\forall_\flat \forall_\sigma U \dsub{\delta(\facewkn{\var i})}(p, q)$. So pick $\vfi : \PSub{\flat V}{(W, \ctxpath{\var i})}$; then we have to show $\forall_\sigma U \dsub{\fpshadj \flat(\delta (\facewkn{\var i}) \vfi)} (\fpshadj \flat(p \psub \vfi), \fpshadj \flat(q \psub \vfi))$.
		
		Without loss of generality, we may assume that $V$ and $(W, \ctxpath{\var i})$ are disjoint. Write $k = \var i \psub \vfi \in V \uplus \accol{0, 1}$. Then $\vfi$ factors as $(\psi, \var i^\IB / \var i^\IP)(k / \var i^\IB)$ for some $\psi : \PSub{\flat V}{W}$. Because $\forall_\sigma U$ respects restriction by $(k / \var i^\IP)$ and because $\flat (k / \var i^\IP) = (k / \var i^\IB)$, it is then sufficient to show that
		\begin{equation}
			\forall_\sigma U \dsub{\fpshadj \flat(\delta (\facewkn{\var i}) (\psi, \var i^\IB / \var i^\IP))} (\fpshadj \flat (p \psub{\psi, \var i^\IB / \var i^\IP}), \fpshadj \flat (q \psub{\psi, \var i^\IB / \var i^\IP}))
		\end{equation}
		which simplifies to
		\begin{equation}
			\forall_\sigma U\dsub{\fpshadj \flat(\delta \psi)(\facewkn{\var i^\IP})}(\fpshadj \flat(p \psub{\psi, \var i^\IB / \var i^\IP}), \fpshadj \flat(q \psub{\psi, \var i^\IB / \var i^\IP})).
		\end{equation}
		Write
		\begin{align*}
			\delta' &:= \fpshadj \flat(\delta \psi) : \DSub{V}{\sharp \Delta}, \\
			(V, \ctxpath{\var i}) \Dsez p' &:= \fpshadj \flat(p \psub{\psi, \var i^\IB / \var i^\IP}) : (\sharp S) \dsub{\delta' (\facewkn{\var i})}, \\
			(V, \ctxpath{\var i}) \Dsez q' &:= \fpshadj \flat(q \psub{\psi, \var i^\IB / \var i^\IP}) : (\sharp S) \dsub{\delta' (\facewkn{\var i})},
		\end{align*}
		which satisfies
		\begin{align}
			(V, \ctxpath{\var i}) \Dsez \overline{p'} &= \overline{q'} : (\sharp \quotshp S) \dsub{\delta' (\facewkn{\var i})}, \label{eq:pf-left-quotshp-1}\\
			(V, \ctxpath{\var i}) \Dsez q' &= p' \psub{0 / \var i, \facewkn{\var i}} : (\sharp \quotshp S) \dsub{\delta' (\facewkn{\var i})}. \label{eq:pf-left-quotshp-2}
		\end{align}
		Then we can further simplify to $\forall_\sigma U \dsub{\delta' (\facewkn{\var i})}(p', q')$.
		
		So pick $\chi : \PSub{Y}{(V, \ctxpath{\var i})}$ and $\gamma : \DSub Y \Gamma$ so that $\sigma \gamma = \delta' (\facewkn{\var i}) \chi$. We have to prove $U \dsub{\gamma}(p' \psub \chi, q' \psub \chi)$. Again, without loss of generality, we may assume that $Y$ and $(V, \ctxpath{\var i})$ are disjoint. Then again, $\chi$ factors as $(\omega, \var i / \var i)(j / \var i)$ for some $\omega : \PSub Y V$, where $j = \var i \psub \chi$. We claim that it then suffices to show that $U \dsub{\gamma (\facewkn{\var i})}(p' \psub{\omega, \var i / \var i}, q' \psub{\omega, \var i / \var i})$. First, note that this is well-typed, i.e.
		\begin{equation}
			(Y, \ctxpath{\var i}) \Dsez p' \psub{\omega, \var i / \var i}, q' \psub{\omega, \var i / \var i} : (\sharp S)[\sigma]\dsub{\gamma(\facewkn{\var i})}
		\end{equation}
		because $\delta'(\facewkn{\var i})(\omega, \var i/\var i) = \delta'(\facewkn{\var i})(\omega, \var i/\var i)(j / \var i)(\facewkn{\var i}) = \delta'(\facewkn{\var i})\chi(\facewkn{\var i}) = \gamma(\facewkn{\var i})$. Second, if we further restrict the anticipated result by $(j / \var i)$, then we do obtain $U \dsub \gamma(p' \psub \chi, q' \psub \chi)$.
		
		So it remains to prove that $U \dsub{\gamma (\facewkn{\var i})}(p' \psub{\omega, \var i / \var i}, q' \psub{\omega, \var i / \var i})$, i.e.
		\begin{equation}
			(Y, \ctxpath{\var i}) \Dsez u \dsub{\gamma (\facewkn{\var i}), p' \psub{\omega, \var i / \var i}} = u \dsub{\gamma (\facewkn{\var i}), q' \psub{\omega, \var i / \var i}} : T\dsub{\gamma(\facewkn{\var i}), \overline{p' \psub{\omega, \var i / \var i}}},
		\end{equation}
		which is well-typed by \cref{eq:pf-left-quotshp-1}. The combination of \cref{eq:pf-left-quotshp-1} and \cref{eq:pf-left-quotshp-2} tells us that $\overline{p' \psub{\omega, \var i / \var i}}$ is degenerate in $\var i$. Hence, by discreteness of $T$, we have
		\begin{equation}
			u \dsub{\gamma (\facewkn{\var i}), p' \psub{\omega, \var i / \var i}}
			= u \dsub{\gamma (\facewkn{\var i}), p' \psub{\omega, \var i / \var i}} \psub{0/\var i, \facewkn{\var i}}
			= u \dsub{\gamma (\facewkn{\var i}), q' \psub{\omega, \var i / \var i}}.
		\end{equation}
		
		\item[Rule 3] Since $(\coshp \quotshp S)[\sigma] = (\quotshp \coshp S)[\sigma] = \quotshp((\coshp S)[\sigma])$, and $\coshp \hatinquotshp = \hatinquotshp$, the third rule is a special case of the first rule. \qedhere
	\end{description}
\end{proof}

\section{Universes of discrete types}
In \cref{sec:uniNDD} we give a straightforward definition of a sequence of universes that classify discrete types. Unfortunately, these universes are themselves not discrete, so that they do not contain their lower-level counterparts. In \cref{sec:uniDD-discussion} we discuss the problem and define a hierarchy of discrete universes of discrete types.
As of this point, we will write $\uniPsh_\ell$ for the standard presheaf universe $\uni \ell$.

\subsection{Non-discrete universes of discrete types}\label{sec:uniNDD}
In any presheaf model, we have a hierarchy of universes $\uniPsh_{\ell}$ such that
\begin{equation}
	\inference{\Gamma \ctx}{\Gamma \sez \uniPsh_\ell \type_{\ell+1}}{}, \qquad
	\binference{\Gamma \sez A : \uniPsh_\ell}{\Gamma \sez \El\,A \type_\ell}{}.
\end{equation}
In this section, we will devise a sequence of universes $\uniNDD_\ell$ such that
\begin{equation}
	\inference{\Gamma \ctx}{\Gamma \sez \uniNDD \type_{\ell+1}}{}, \qquad
	\binference{\Gamma \sez A : \uniNDD_\ell}{\Gamma \sez \El\,A \dtype_\ell}{},
\end{equation}
that is: $\uniNDD_\ell$ classifies discrete types of level $\ell$, but it is itself non-discrete. In \cref{ch:discrete-universe-of-discrete-types}, we will devise a universe that is itself discrete, and that in an unusual way classifies all discrete types.
\begin{proposition}
	The CwF $\widehat{\bpcubecat}$ supports a universe for $\DTy_\ell$, the functor that maps a context $\Gamma$ to its set of discrete level $\ell$ types $\Gamma \sez T \dtype_\ell$.
\end{proposition}
\begin{proof}
	Given $\gamma : \DSub W \Gamma$, we define $\uniNDD_\ell \dsub \gamma := \set{\dtycode T}{\yoneda W \sez T \dtype_\ell}$. This makes $\uniNDD_\ell$ a dependent subpresheaf of $\uniPsh_\ell$. We use the same construction for encoding and decoding types (see \cref{thm:unipsh} on page \pageref{thm:unipsh}). The only thing we have to show is that a type $(\Gamma \sez T \type_\ell)$ is discrete if and only if its encoding $(\Gamma \sez \tycode T : \uniPsh_\ell)$ is a term of $\uniNDD_\ell$.
	\begin{itemize}
		\item[$\Rightarrow$] If $\Gamma \sez T \dtype_\ell$, then $\tycode T \dsub \gamma = \dtycode{T [\gamma]}$, and clearly $\yoneda W \sez T [\gamma] \dtype$ is discrete.
		
		\item[$\Leftarrow$] Assume $\Gamma \sez A : \uniNDD_\ell$. We show that $\Gamma \sez \El\,A \type$ is a discrete type, so pick a path $(W, \ctxpath{\var i}) \Dsez p : (\El\,A) \dsub{\gamma (\facewkn{\var i})}$. We need to show that $p = p \psub{0/\var i, \facewkn{\var i}}^{\El\,A}$. We have
		\begin{align*}
			p \psub{0/\var i, \facewkn{\var i}}^{\El\,A}
			&= p \psub{0/\var i, \facewkn{\var i}}^{\dEl(A \dsub{\gamma (\facewkn{\var i})})}
			= p \psub{0/\var i, \facewkn{\var i}}^{\dEl(A \dsub{\gamma}) [\facewkn{\var i}]},
		\end{align*}
		and so we need to prove $(W, \ctxpath{\var i}) \Dsez p = p \psub{0/\var i, \facewkn{\var i}} : \dEl(A \dsub{\gamma}) \dsub{\facewkn{\var i}}$. But $\dEl(A \dsub \gamma)$ is discrete by construction of $\uniNDD_\ell$ and $(\facewkn{\var i}) : \DSub{(W, \ctxpath{\var i})}{\yoneda W}$ is degenerate in $\var i$, so that this equality indeed holds. \qedhere
	\end{itemize}
\end{proof}

\subsection{Discrete universes of discrete types}\label{sec:uniDD-discussion}
Let us have a look at the structure of $\uniNDD$ (ignoring universe levels for a moment):
\begin{itemize}
	\item A \textbf{point} in $\uniNDD$ is a discrete type $\yoneda() \sez T \dtype$. Since $\yoneda()$ is the empty context, this effectively means that points in $\uniNDD$ are discrete closed types, as one would expect. Differently put, for every shape $W$, there is only one cube $\bullet : \DSub{W}{\yoneda()}$ and thus all $W$-shaped cubes $W \Dsez t : T \dsub{\bullet}$ have the same status; essentially $T$ has the structure of a non-dependent presheaf.
	\item A \textbf{path} in $\uniNDD$ is a discrete type $\yoneda(\ctxpath{\var i}) \sez T \dtype$. For every shape $W$, the presheaf $\yoneda(\ctxpath{\var i})$ contains fully degenerate $W$-cubes $(\facewkn{W}, 0/\var i), (\facewkn{W}, 1/\var i) : \DSub{W}{\yoneda(\ctxpath{\var i})}$. As these cubes are fully degenerate, all $W$-cubes of $T$ above them, must also be degenerate in all path dimensions (as $T$ is discrete). So $T$ contains two discrete, closed types $T[0/\var i]$ and $T[1/\var i]$.
	
	Moreover, for every shape $W \not\ni \var i$, we have a cube $(\facewkn W) : \DSub{(W, \ctxpath{\var i})}{\yoneda(\ctxpath{\var i})}$ that is degenerate in all dimensions but $\var i$. We can think of this as the constant cube on the path $\id : \DSub{(\ctxpath{\var i})}{\yoneda(\ctxpath{\var i})}$. Above it live heterogeneous higher paths (degenerate in all path dimensions but $\var i$) that connect a $W$-cube of $A$ with a $W$-cube of $B$. We get a similar setup of heterogeneous higher bridges from $(\facewkn W, \var i^\IB / \var i^\IP) : \DSub{(W, \ctxbrid{\var i})}{\yoneda(\ctxpath{\var i})}$. Finally, the face map $(\var i^\IB/\var i^\IP) : \PSub{(W, \ctxbrid{\var i})}{(W, \ctxpath{\var i})}$ allows us to find under every heterogeneous path, a heterogeneous bridge.
	
	Thus, bluntly put, a path from $A$ to $B$ in $\uniNDD$ consists of:
	\begin{itemize}
		\item A (discrete) notion of heterogeneous paths with source in $A$ and target in $B$,
		\item A (discrete) notion of heterogeneous bridges with source in $A$ and target in $B$,
		\item An operation that gives us a heterogeneous bridge under every heterogeneous path.
	\end{itemize}
	
	\item A \textbf{bridge} in $\uniNDD$ is a discrete type $\yoneda(\ctxbrid{\var i}) \sez T \type$. The presheaf $\yoneda(\ctxbrid{\var i})$ has everything that $\yoneda(\ctxpath{\var i})$ has, except for the interesting path. A similar analysis as above, shows that a bridge from $A$ to $B$ in $\uniNDD$ is quite simply a (discrete) notion of heterogeneous bridges from $A$ to $B$.
\end{itemize}
Now let us think a moment about what we want:
\begin{itemize}
	\item The \textbf{points} seem to be all right: we want them to be discrete closed types.
	\item A \textbf{path} in the universe should always be degenerate, if we want the universe to be a discrete closed type.
	\item In order to understand what a \textbf{bridge} should be, let us have a look at parametric functions. A function $f : \forall(X : \uni{}).\El\,X$ (which we know does not exist, but this choice of type keeps the example simple) is supposed to map related types $X$ and $Y$ to heterogeneously equal values $fX : \El\,X$ and $fY : \El\,Y$. Since bridges were invented as an abstraction of relations, and paths as some sort of pre-equality, we can reformulate this: The function $f$ should map bridges from $X$ to $Y$ to heterogeneous paths from $fX$ to $fY$. Well, then a bridge from $X$ to $Y$ will certainly have to provide a notion of heterogeneous paths between $\El\,X$ and $\El\,Y$!
	
	On the other hand, consider the (non-parametric) type $\Sigma(X : \uni{}).\El\,X$. What is a bridge between $(X, x)$ and $(Y, y)$ in this type? We should expect it to be a bridge from $X$ to $Y$ and a heterogeneous bridge from $x$ to $y$. This shows that bridges in the universe should also provide a notion of bridges.
\end{itemize}
To conclude: we want $\uni{}$ to be a type whose paths are constant, and whose bridges are the paths from $\uniNDD$, i.e. terms $(\ctxbrid{\vec{\var j}}, \ctxpath{\vec{\var i}}) \Dsez A : \uniDD$ should correspond to terms $(\ctxpath{\vec {\var j}}) \Dsez A' : \uniNDD$. So we define it that way:
\begin{definition}
	We define the \textbf{discrete universe of discrete level $\ell$ types} $\sez \uniDD_\ell \dtype_{\ell+1}$ as $\uniDD_\ell = \cohdisc \cohpaths \uniNDD_\ell = \flat \coshp \uniNDD_\ell$.
\end{definition}
Note that $\sharp \shp \dashv \flat \coshp$ and that $\sharp \shp (\ctxbrid{\vec{\var j}}, \ctxpath{\vec{\var i}}) = (\ctxpath{\vec {\var j}})$.

There is a minor issue with the above definition: we want $\uniDD_\ell$ to exist in any context. We can simply define $\Gamma \sez \uniDD_\ell \dtype_{\ell+1}$ as $\uniDD_\ell = (\flat \coshp \uniNDD_\ell)[\bullet]$. Note that $\uniNDD_\ell = \uniNDD_\ell[\bullet]$, so this does not destroy any information.

The universes $\uniPsh_\ell$ and $\uniNDD_\ell$ have a decoding operation $\El$ and an inverse encoding operation $\tycode \loch$ that allow us to turn terms of the universe into types and vice versa. Moreover, the operators for $\uniNDD_\ell$ are simply those of $\uniPsh_\ell$ restricted to $\uniNDD_\ell$ (for $\El$) or to discrete types (for $\tycode \loch$). For $\uniDD$, the situation is different:
\begin{proposition}
	We have mutually inverse rules
	\begin{equation}
		\inference{
			\Gamma \sez A : \uniDD_\ell
		}{\sharp \shp \Gamma \sez \ElDD~A \dtype_\ell}{} \qquad
		\inference{
			\sharp \shp \Gamma \sez T \dtype_\ell
		}{\Gamma \sez \tycodeDD T : \uniDD_\ell}{}
	\end{equation}
	that are natural in $\Gamma$, i.e. $(\ElDD\,A)[\sharp \shp \sigma] = \ElDD(A[\sigma])$.
	Moreover, $(\ElDD\,A)[\iota \varsigma] = \El\,\vartheta(\kappa(A))$.
\end{proposition}
\begin{proof}
	We use that $\uniDD_\ell = \flat \coshp \uniNDD_\ell$.

	We set $\ElDD~A = \El~\alpha_{\sharp \dashv \coshp}\inv(\alpha_{\shp \dashv \flat}\inv(A)) = \El~\vartheta(\kappa(A[\varsigma\inv])[\iota]\inv)$. Then the inverse is given by $\tycodeDD T = \kappa\inv(\vartheta\inv(\tycode T)[\iota])[\varsigma]$.
\end{proof}

\chapter{Semantics of ParamDTT}\label{ch:paramdtt}
In this chapter, we finally interpret the inference rules of ParamDTT in the category with families $\widehat{\bpcubecat}$ of bridge/path cubical sets. We start with some auxiliary lemmas, then give the meta-type of the interpretation function, followed by interpretations for the core typing rules, the typing rules related to internal parametricity, and the typing rules related to $\Nat$ and $\Size$.

\section{Some lemmas}\label{sec:uniDD-lemmas}
\begin{lemma}
	For discrete types $T$ in the relevant contexts, we have invertible rules:
	\begin{equation}
		\binference{\shp \Gamma \sez t : T}{\Gamma \sez t[\varsigma] : T[\varsigma]}{}, \qquad
		\binference{\sharp \shp \Gamma \sez t : T}{\sharp \Gamma \sez t[\sharp \varsigma] : T[\sharp \varsigma]}{}.
	\end{equation}
\end{lemma}
\begin{proof}
	\begin{description}
		\item[Rule 1] Recall that we have $\kappa \quotshp : \shp \cong \quotshp$ and $\varsigma = (\kappa \quotshp)\inv \inquotshp$. Thus, it is sufficient to prove
		\begin{equation}
			\binference{\quotshp \Gamma \sez t' : T'}{\Gamma \sez t'[\inquotshp] : T'[\inquotshp]}{},
		\end{equation}
		after which we can pick $T' = T[(\kappa \quotshp)\inv]$ and $t' = t[(\kappa \quotshp)\inv]$. A proof of this is analogous to but simpler than the proof of the first rule in \cref{thm:elim-quotshp}.
		
		\item[Rule 2] Since $\sharp \flat = \sharp$ and $\sharp \kappa = \id$, we have $\sharp \shp = \sharp \quotshp$ and $\sharp \varsigma = \sharp \inquotshp$. Thus, we need to prove
		\begin{equation}
			\binference{\sharp \quotshp \Gamma \sez t : T}{\sharp \Gamma \sez t[\sharp \inquotshp] : T[\sharp \inquotshp]}{}.
		\end{equation}
		 A proof of this is analogous to but simpler than the proof of the second rule in \cref{thm:elim-quotshp}. \qedhere
	\end{description}
\end{proof}
\begin{lemma}
	For discrete types $\shp\Gamma \sez T \dtype$, we have an invertible substitution
	\begin{equation}
		(\shp \pi, \xi[\varsigma]\inv) : \shp(\Gamma.T[\varsigma]) \cong (\shp \Gamma).T.
	\end{equation}
	We will abbreviate it as $\subext \varsigma\inv$ and the inverse as $\subext \varsigma$. We have $\subext \varsigma \circ \varsigma \subext = \varsigma$.
\end{lemma}
\begin{proof}
	We have a commutative diagram
	\begin{equation}
		\xymatrix{
			\shp(\Gamma.T[\varsigma])
				\ar[rr]^{(\shp \pi, \xi[\varsigma]\inv)}
			&& (\shp \Gamma).T
			\\
			& \Gamma.T[\varsigma]
				\ar[lu]_{\varsigma}
				\ar[ru]^{\varsigma \subext}
				\ar[ld]^{\inquotshp}
				\ar[rd]_{\inquotshp \subext}
			\\
			{\quotshp(\Gamma.T[\varsigma])}
				\ar[rr]^{(\quotshp \pi, \xi[\inquotshp]\inv)}
				\ar[uu]^{(\kappa \quotshp)\inv}_{\wr}
			&& (\quotshp \Gamma).T[(\kappa \quotshp)\inv].
				\ar[uu]_{(\kappa \quotshp)\inv \subext}^{\wr}
		}
	\end{equation}
	The left and right triangles commute because $\varsigma = (\kappa \quotshp)\inv \inquotshp$. The upper triangle commutes because $(\shp \pi, \xi[\varsigma]\inv) \varsigma = (\shp \pi \circ \varsigma, \xi[\varsigma]\inv[\varsigma]) = (\varsigma \pi, \xi) = \varsigma \subext$. The lower one commutes by similar reasoning. The square commutes because
	\begin{align*}
		\kappa \quotshp \subext \circ (\shp \pi, \xi[\varsigma]\inv)
		&= \kappa \quotshp \subext \circ (\flat \quotshp \pi, \xi[\inquotshp]\inv[\kappa \quotshp])
		= (\kappa \quotshp \circ \flat \quotshp \pi, \xi[\inquotshp]\inv[\kappa \quotshp]) \\
		&= (\quotshp \pi \circ \kappa \quotshp, \xi[\inquotshp]\inv[\kappa \quotshp])
		= (\quotshp \pi, \xi[\inquotshp]\inv) \circ \kappa \quotshp.
	\end{align*}
	So in order to prove the theorem, it is sufficient to show that the lower arrow is invertible. It maps $\overline{(\gamma, t)} : \DSub W {\quotshp(\Gamma.T[\varsigma])}$ to
	\begin{equation}
		(\quotshp \pi, \xi[\inquotshp]\inv) \circ \overline{(\gamma, t)}
		= (\quotshp \pi \circ \overline{(\gamma, t)}, \xi[\inquotshp]\inv \dsub{\overline{(\gamma, t)}})
		= (\overline \gamma, \xi \dsub{\gamma, t}) = (\overline \gamma, t) : \DSub{W}{(\quotshp \Gamma).T[(\kappa \quotshp)\inv]}.
	\end{equation}
	So we have to show that we can do the converse. So we have to show that if we have $\gamma, \gamma' : \DSub W \Gamma$ such that $\sheq^\Gamma_W(\gamma, \gamma')$ (i.e.\ $\overline \gamma = \overline{\gamma'}$), and $t : T[\varsigma]\dsub{\gamma} = T[(\kappa \quotshp)\inv] \dsub{\overline \gamma}$, then $\overline{(\gamma, t)} = \overline{(\gamma', t)}$. We will prove a stronger statement, namely that $\sheq^\Gamma \subseteq E$, where we say $E_W(\gamma, \gamma')$ when $\sheq^\Gamma(\gamma, \gamma')$ and for every $\vfi : \PSub{V}{W}$ and every $t : T[(\kappa \quotshp)\inv]\dsub{\overline{\gamma \vfi}}$, we have $\overline{(\gamma\vfi, t)} = \overline{(\gamma'\vfi, t)}$. Because $E$ is an equivalence relation on $\Gamma$, it suffices to prove that $E$ satisfies the defining property of $\Gamma$.
	
	So pick a path $\gamma : \DSub{(W, \ctxpath{\var i})}{\Gamma}$. We have to prove $E(\gamma, \gamma (0/\var i, \facewkn{\var i}))$. Pick some $\vfi : \PSub{V}{(W, \ctxpath{\var i})}$. As usual, we can decompose $\vfi = (\psi, \var i/\var i)(k/\var i)$ for some $\psi : \PSub V W$ and $k \in V \uplus \accol{0, 1}$. Pick $t : T[(\kappa \quotshp)\inv]\dsub{\overline{\gamma \vfi}}$. We have
	\begin{equation}
		\overline{(\gamma(\psi, \var i/\var i), t \psub{\facewkn{\var i}})}
		= \overline{(\gamma(\psi, \var i/\var i), t \psub{\facewkn{\var i}})} (0/\var i, \facewkn{\var i})
		= \overline{(\gamma (0/\var i, \facewkn{\var i}) (\psi, \var i/\var i), t \psub{\facewkn{\var i}})}
	\end{equation}
	where the first equality holds by definition of $\sheq$, and the second one follows from calculating with substitutions. Restricting by $(k/\var i)$ yields the desired result.
\end{proof}

\section{Meta-type of the interpretation function}
\textbf{Contexts} $\Gamma \ctx$ are interpreted to bridge/path cubical sets $\interp \Gamma \ctx$.

\textbf{Types} $\Gamma \judty{T}$ are interpreted to discrete types $\sharp \interp \Gamma \sez \interp T_\Ty \dtype$.

\textbf{Terms} $\Gamma \judty t T$ are interpreted as terms $\interp \Gamma \sez \interp t : \interp T [\iota]$.

\textbf{Definitional equality} is interpreted as equality of interpretations.

In the paper, the promotion of an element of the universe to a type, is not reflected syntactically. For that reason, we need a different interpretation function for types and for terms. However, to keep things simpler here, we will add a syntactical reminder $\El$ of the term-to-type promotion, allowing us to omit the index $\Ty$.

\section{Core typing rules}
\subsection{Contexts}
Context formation rules are interpreted as follows:
\begin{equation}
	\interp{
		\inference{}{\ctx}{c-em}	
	} =
	\inference{}{\ctx}{}
\end{equation}
\begin{equation}
	\interp{
		\inference{\Gamma \judty T}{\Gamma, \ctxvar \mu x T \ctx}{c-ext}
	} =
	\inference{
		\inference{
			\inference{
				\sharp \interp \Gamma \sez \interp T \dtype \qquad
				\mu \in \accol{\Id, \sharp, \coshp}
			}{\sharp \interp \Gamma \sez \mu \interp T \type}{}
		}{\interp \Gamma \sez (\mu \interp T)[\iota] \type}{}
	}{\interp \Gamma, \ftrvar \mu x : (\mu \interp T)[\iota] \ctx}{}
\end{equation}
The variable rule
\begin{equation}
	\interp{
	\inference{
		\Gamma \ctx \quad
		(\ctxvar \mu x T) \in \Gamma \quad
		\mu \leq \idmod
	}{\Gamma \judtm x T}{t-var}
	}
\end{equation}
is interpreted through a combination of weakening and the following rules:
\begin{equation}
	\inference{
		\interp \Gamma, \var x : \interp T [\iota] \ctx
	}{\interp \Gamma, \var x : \interp T [\iota] \sez \var x : \interp T [\iota \wknvar x]}{}
	, \qquad
	\inference{
		\interp \Gamma, \var x : (\coshp \interp T) [\iota] \ctx
	}{\interp \Gamma, \var x : (\coshp \interp T) [\iota] \sez \vartheta(\var x) : \interp T [\iota \wknvar x]}{}
\end{equation}
The second one would normally have type $\interp T [\vartheta \iota \wknvar x]$, but we have $\vartheta \sharp = \id$, so we may remove $\vartheta : \sharp \interp \Gamma \to \sharp \interp \Gamma$.
\begin{lemma}\label{thm:leftflat}\label{thm:leftflat2}
	For any syntactic context $\Gamma$, we have $\sharp \interp \Gamma = \sharp \interp{\sharp \setminus \Gamma}$ and equivalently $\flat \interp \Gamma = \flat \interp{\sharp \setminus \Gamma}$. The substitution $\kappa : \flat \interp \Gamma \cong \interp{\sharp \setminus \Gamma}$ is an isomorphism.
\end{lemma}
\begin{proof}
	We prove this by induction on the length of the context.
	\begin{description}
		\item[Empty context] We have $\sharp \setminus () = ()$; hence $\flat \interp{\sharp \setminus ()} = \flat \interp{()} = \flat () = ()$. Then $\kappa : () \to ()$ is the only substitution of that type and it is indeed an isomorphism.
		
		\item[Pointwise extension] We have
		\begin{align*}
			\flat \interp{\Gamma, \ctxptw x T}
			&= \flat (\interp \Gamma, \ftrvar \coshp x : (\coshp \interp T)[\iota])
			= \flat \interp \Gamma, \ftrvar {\flat\coshp} x : \flat \coshp \interp T, \\
			\flat \interp{\sharp \setminus (\Gamma, \ctxptw x T)}
			&= \flat (\interp{\sharp \setminus \Gamma}, \ftrvar \coshp x : (\coshp \interp T)[\iota])
			= \flat \interp{\sharp \setminus \Gamma}, \ftrvar {\flat\coshp} x : \flat \coshp \interp T,
		\end{align*}
		which is equal by virtue of the induction hypothesis. The context
		\begin{equation}
			\interp{\sharp \setminus (\Gamma, \ctxptw x T)}
			 = \interp{\sharp \setminus \Gamma}, \ftrvar \coshp x : (\coshp \interp T)[\iota]
		\end{equation}
		is discrete because $\interp{\sharp \setminus \Gamma} \cong \flat \interp \Gamma$ by the induction hypothesis, $\interp T$ is discrete and $\coshp$ preserves discreteness (\cref{thm:coshp-preserves-discreteness}). Hence, $\kappa$ is an isomorphism for this context.
		
		\item[Continuous extension] We have
		\begin{align*}
			\flat \interp{\Gamma, \ctxctu x T}
			&= \flat (\interp \Gamma, \var x : \interp T[\iota])
			= \flat \interp \Gamma, \ftrvar \flat x : \flat \interp T, \\
			\flat \interp{\sharp \setminus (\Gamma, \ctxctu x T)}
			&= \flat (\interp{\sharp \setminus \Gamma}, \var x : \interp T[\iota])
			= \flat \interp{\sharp \setminus \Gamma}, \ftrvar \flat x : \flat \interp T,
		\end{align*}
		which is equal by virtue of the induction hypothesis. The context
		\begin{equation}
			\interp{\sharp \setminus (\Gamma, \ctxctu x T)}
			 = \interp{\sharp \setminus \Gamma}, \var x : \interp T[\iota]
		\end{equation}
		is discrete because $\interp{\sharp \setminus \Gamma} \cong \flat \interp \Gamma$ by the induction hypothesis and $\interp T$ is discrete. Hence, $\kappa$ is an isomorphism for this context.
		
		\item[Parametric extension] We have
		\begin{align*}
			\flat \interp{\Gamma, \ctxpar x T}
			&= \flat (\interp \Gamma, \ftrvar \sharp x : (\sharp \interp T)[\iota])
			= \flat \interp \Gamma, \ftrvar{\flat\sharp} x : \flat \sharp \interp T,
			= \flat \interp \Gamma, \ftrvar \flat x : \flat \interp T, \\
			\flat \interp{\sharp \setminus (\Gamma, \ctxpar x T)}
			&= \flat \interp{(\sharp \setminus \Gamma), \ctxctu x T)}
			= \flat (\interp{\sharp \setminus \Gamma}, \var x : \interp T[\iota])
			= \flat \interp{\sharp \setminus \Gamma}, \ftrvar \flat x : \flat \interp T,
		\end{align*}
		which is equal by virtue of the induction hypothesis. The context
		\begin{equation}
			\interp{\sharp \setminus (\Gamma, \ctxpar x T)}
			= \interp{(\sharp \setminus \Gamma), \ctxctu x T)}
			 = \interp{\sharp \setminus \Gamma}, \var x : \interp T[\iota]
		\end{equation}
		is discrete because $\interp{\sharp \setminus \Gamma} \cong \flat \interp \Gamma$ by the induction hypothesis and $\interp T$ is discrete. Hence, $\kappa$ is an isomorphism for this context.
		
		\item[Interval extensions] These will be special cases of the above.
		\item[Face predicate extension] See the addendum in \cref{sec:face-predicates}. \qedhere
	\end{description}
\end{proof}
\begin{lemma}\label{thm:leftsharp}\label{thm:leftsharp2}
	For any syntactic context $\Gamma$, we have $\interp{\coshp \setminus \Gamma} = \sharp \interp \Gamma$.
\end{lemma}
\begin{proof}
	We prove this by induction on the length of the context.
	\begin{description}
		\item[Empty context] We have $\coshp \setminus () = ()$; hence $\interp{\coshp \setminus ()} = \interp{()} = () = \sharp ()$.
		
		\item[Pointwise extension] We have
		\begin{align*}
			\sharp \interp{\Gamma, \ctxptw x T}
			&= \sharp (\interp \Gamma, \ftrvar \coshp x : (\coshp \interp T)[\iota])
			= \sharp \interp \Gamma, \ftrvar{\sharp \coshp}{x} : \sharp \coshp \interp T
			= \sharp \interp \Gamma, \ftrvar{\coshp}{x} : \coshp \interp T, \\
			\interp{\coshp \setminus (\Gamma, \ctxptw x T)}
			&= \interp{(\coshp \setminus \Gamma), \ctxptw x T}
			= \interp{\coshp \setminus \Gamma}, \ftrvar \coshp x : (\coshp \interp T)[\iota]
			= \sharp \interp{\Gamma}, \ftrvar \coshp x : \coshp \interp T,
		\end{align*}
		where in the last step we used the induction hypothesis and the fact that $\iota \sharp = \id : \sharp \interp \Gamma \to \sharp \interp \Gamma$.
		
		\item[Continuous extension] We have
		\begin{align*}
			\sharp \interp{\Gamma, \ctxctu x T}
			&= \sharp (\interp \Gamma, \var x : \interp T[\iota])
			= \sharp \interp \Gamma, \ftrvar \sharp x : \sharp \interp T, \\
			\interp{\coshp \setminus (\Gamma, \ctxctu x T)}
			= \interp{(\coshp \setminus \Gamma), \ctxpar x T}
			&= \interp{\coshp \setminus \Gamma}, \ftrvar \sharp x : (\sharp \interp T)[\iota]
			= \sharp \interp \Gamma, \ftrvar \sharp x : \sharp \interp T.
		\end{align*}
		
		\item[Parametric extension] We have
		\begin{align*}
			\sharp \interp{\Gamma, \ctxpar x T}
			&= \sharp(\interp \Gamma, \ftrvar \sharp x : (\sharp \interp T)[\iota])
			= \sharp \interp \Gamma, \ftrvar \sharp x : \sharp \interp T, \\
			\interp{\coshp \setminus (\Gamma, \ctxpar x T)}
			= \interp{(\coshp \setminus \Gamma), \ctxpar x T}
			&= \interp{\coshp \setminus \Gamma}, \ftrvar \sharp x : (\sharp \interp T)[\iota]
			= \sharp \interp \Gamma, \ftrvar \sharp x : \sharp \interp T.
		\end{align*}
		
		\item[Interval extensions] These will be special cases of the above.
		
		\item[Face predicate extension] See the addendum in \cref{sec:face-predicates}. \qedhere
	\end{description}
\end{proof}

\subsection{Universes}
We have
\begin{equation}
	\interp{
		\inference{
			\Gamma \ctx \qquad
			\ell \in \IN
		}{\Gamma \judtm{\uni \ell}{\El\,\uni{\ell+1}}}{t-Uni}
	} =
	\inference{
		\interp \Gamma \ctx \qquad
		\ell \in \IN
	}{\interp \Gamma \sez \tycodeDD{\uniDD_\ell} : \uniDD_{\ell + 1}}{}
\end{equation}
\begin{equation}
	\interp{
		\inference{
			\Gamma \judtm T {\El\,\uni k} \qquad
			k \leq \ell \in \IN
		}{\Gamma \judtm{T}{\El\,\uni \ell}}{t-lift}
	} =
	\inference{
		\interp \Gamma \sez \interp T : \uniDD_k \qquad
		k \leq \ell \in \IN
	}{\interp \Gamma \sez \interp T : \uniDD_\ell}{}
\end{equation}
\begin{equation}
	\interp{
		\inference{\leftflat\Gamma \judtm A {\uni{\ell}}}{\Gamma \judty{\El\,A}}{ty}
	} =
	\inference{
		\inference{
			\interp{\sharp \setminus \Gamma} \sez \interp A : \uniDD_\ell
		}{\sharp \shp \interp{\sharp \setminus \Gamma} \sez \ElDD \interp A \dtype}{}
	}{\sharp \interp{\sharp \setminus \Gamma} = \sharp \interp \Gamma \sez (\ElDD \interp A)[\sharp \varsigma] \dtype}{}
\end{equation}
In particular, we have
\begin{equation}
	\interp{\El\,\uni \ell} = (\ElDD \interp{\uni \ell})[\sharp \varsigma]
	= (\ElDD \tycodeDD{\uniDD_\ell})[\sharp \varsigma]
	= \uniDD_\ell[\sharp \varsigma] = \uniDD_\ell,
\end{equation}
so that it is justified that we simply put $\uniDD_k$ on several occasions where we should have used $\interp{\El\,\uni k}$.
\begin{remark}
	In the paper, we defined $\interp{\El\,A}$ as $\El~\vartheta(\ftrtm{\sharp}{\interp A})$. Note that we have
	\begin{align*}
		(\ElDD \interp A)[\sharp \varsigma]
		&= \El\,\vartheta(\kappa(\interp A[\varsigma]\inv)[\iota]\inv)[\sharp \varsigma]
		= \El\,\vartheta(\kappa(\interp A[\varsigma]\inv)[\varsigma][\iota]\inv)
		= \El\,\vartheta(\kappa(\interp A)[\iota]\inv).
	\end{align*}
	Moreover, $\sharp \interp \Gamma \sez \kappa(\interp A)[\iota]\inv = \ftrtm{\sharp}{\interp A} : \coshp \uniNDD_\ell$ because $(\ftrtm{\sharp}{\interp A})[\iota] = \iota(\interp A) = \kappa(\interp A)$ where the last step uses $\iota = \kappa : \flat \coshp \to \coshp$.
\end{remark}

\subsection{Substitution}
We have syntactic substitution rule
\begin{equation}
	\inference{
		\Gamma, \ctxvar \mu x T, \Delta \sez J \qquad
		\mu \setminus \Gamma \judtm t T
	}{\Gamma, \Delta[\varclr t/\varclr x] \sez J[\varclr t/\varclr x]}{subst}
\end{equation}
which can be shown to be admissible by induction on the derivation of $J$. The idea behind this is a combination of the general idea of substitution, and the fact that we use $\mu \setminus \Gamma \judtm t T$ to express something that would more intuitively look like $\Gamma \sez \ctxvar \mu t T$. In fact, we have the following result:
\begin{lemma}
	For $\mu \in \accol{\coshp, \idmod, \sharp}$, we have
	\begin{equation}
		\inference{
			\interp{\mu \setminus \Gamma} \sez t : T [\iota]
		}{\interp \Gamma \sez \forsub \mu t : (\mu T)[\iota]}{}
	\end{equation}
	where $\forsub \idmod t = t$, $\forsub \sharp t = (\ftrtm \sharp t)[\iota]$ and $\forsub \coshp t = (\ftrtm \coshp t)[\iota]$.
\end{lemma}
\begin{proof}
	The idea is that $\mu \setminus \loch$ is left adjoint to $\mu \circ \loch$. In the model, we see this formally as $\interp{\coshp \setminus \Gamma} = \sharp \interp \Gamma$, $\interp{\idmod \setminus \Gamma} = \interp \Gamma$ and $\interp{\sharp \setminus \Gamma} \cong \flat \interp \Gamma$. So in each case, we can simply use the adjunction in the model.
	\begin{description}
		\item[$\mu = \idmod$] Then $\interp \Gamma \sez t : T[\iota]$ by the premise.
		
		\item[$\mu = \sharp$] Then we have
		\begin{equation}
			\inference{
			\inference{
				\interp{\sharp \setminus \Gamma} \sez t : T [\iota]
			}{\flat \interp \Gamma \sez t [\kappa] : T [\iota \kappa]}{}
			}{\interp \Gamma \sez \iota(t [\kappa])[\kappa]\inv : (\sharp T)[\iota]}{},
		\end{equation}
		i.e.\ we first apply the isomorphism $\kappa : \flat \interp \Gamma \cong \interp{\sharp \setminus \Gamma}$ and then the adjunction $\iota(\loch)[\kappa]\inv : \flat \dashv \sharp$. The second step is well-typed because $\iota \circ (\iota \kappa) / \kappa = \iota : \interp \Gamma \to \sharp \interp \Gamma = \sharp \interp{\sharp \setminus \Gamma}$, or more meaningfully
		\begin{equation}
			\iota \sharp \interp \Gamma \circ (\iota \interp{\sharp \setminus \Gamma} \circ \kappa \interp{\sharp \setminus \Gamma}) / \kappa \interp{\Gamma} = \iota \interp \Gamma.
		\end{equation}
		Indeed, $\iota \sharp = \id$ and $(\iota \kappa) \interp{\sharp \setminus \Gamma} = \iota \flat \interp{\sharp \setminus \Gamma} = \iota \flat \interp \Gamma = \iota \interp \Gamma \circ \kappa \interp \Gamma$.
		
		However, the resulting term is a bit obscure. We can instead do
		\begin{equation}
			\inference{
			\inference{
				\interp{\sharp \setminus \Gamma} \sez t : T[\iota]
			}{\sharp \interp{\Gamma} \sez \ftrtm{\sharp}{t} : \sharp T}{}
			}{\interp \Gamma \sez (\ftrtm{\sharp}{t})[\iota] : (\sharp T)[\iota]}{}.
		\end{equation}
		In the first step, we used that $\sharp \interp{\sharp \setminus \Gamma} = \sharp \interp \Gamma$ and $\sharp \iota = \id$. As it happens, $\iota(t [\kappa])[\kappa]\inv = (\ftrtm{\sharp}{t})[\iota]$, or more precisely
		\begin{equation}
			\iota(t [\kappa \interp{\sharp \setminus \Gamma}])[\kappa \interp \Gamma]\inv = (\ftrtm{\sharp}{t})[\iota \interp \Gamma]
		\end{equation}
		because
		\begin{equation}
			(\ftrtm{\sharp}{t})[\iota \interp \Gamma][\kappa \interp \Gamma]
			= (\ftrtm{\sharp}{t})[\iota \interp{\sharp \setminus \Gamma}][\kappa \interp{\sharp \setminus \Gamma}]
			= \iota(t)[\kappa \interp {\sharp \setminus \Gamma}]
			= \iota(t[\kappa \interp {\sharp \setminus \Gamma}]).
		\end{equation}
		So we conclude $\forsub \sharp t = (\ftrtm{\sharp}{t})[\iota]$.
		
		\item[$\mu = \coshp$] We can apply the adjunction $\vartheta\inv(\loch)[\iota] : \sharp \dashv \coshp$:
		\begin{equation}
			\inference{
				\sharp \interp \Gamma \sez t : T
			}{\interp \Gamma \sez \vartheta\inv(t)[\iota] : (\coshp T)[\iota]}{}
		\end{equation}
		In the premise, we can omit $[\iota]$ on $T$ because $\iota \sharp \interp \Gamma = \id$. The conclusion should normally have type $(\coshp T)[\vartheta \setminus \iota]$. However, $\vartheta \sharp \interp \Gamma = \id$, so we are left with just $\iota$. Note that $\vartheta\inv(t) = \ftrtm{\coshp}{t}$ because $\vartheta (\ftrtm{\coshp}{t}) = t[\vartheta]$ and $\vartheta \sharp \interp \Gamma = \id$.
		So we conclude $\forsub \coshp t = (\ftrtm{\coshp}{t})[\iota]$. \qedhere
	\end{description}
\end{proof}
We assume the following without proof:\footnote{One could argue that the presence of a conjecture in this technical report, implies that we have not proven soundness of the type system. However, in practice, any mathematical proof will wipe some tedious details under the carpet, when the added value of figuring them out is outweighed by the work required to do so. Note also that, would this conjecture be false, most of the model remains intact.}
\begin{conjecture}
	The interpretation of the substitution rule, corresponding to the syntactic admissibility proof, is given by the substitution $(\id, \forsub \mu t) : \interp \Gamma \to \interp{\Gamma, \ctxvar \mu x T}$.
\end{conjecture}
\begin{lemma}
	Let $\Gamma' = (\Gamma, \ctxvar \mu x {\El~A})$ be a syntactic context. Then we have $\mu \setminus \Gamma' \judtm{x}{\El~A}$. The interpretation of $x$ satisfies $\forsub \mu \interp x = \ftrvar \mu x$.
\end{lemma}
\begin{proof}
	For $\mu = \idmod$, this is trivial.
	
	For $\mu = \sharp$, we have
	\begin{equation}
		\forsub \sharp \interp x = \forsub \sharp \var x = (\ftrvar \sharp x)[\iota].
	\end{equation}
	Here, $\iota$ has type $\interp{\Gamma'} \to \sharp \interp{\Gamma'}$, but $\ftrvar \sharp x$ is already in a sharp type in $\interp{\Gamma'}$ and $\iota \sharp = \id$, so we can omit it and have $\forsub \sharp \interp x = \ftrvar \sharp x$.
	
	For $\mu = \coshp$, we have
	\begin{equation}
		\forsub \coshp \interp x = \forsub \coshp \vartheta(\ftrvar \coshp x) = \ftrtm{\coshp}{(\vartheta(\ftrvar \coshp x))}[\iota] = \ftrvar \coshp x [\iota] = \ftrvar \coshp x,
	\end{equation}
	because $\iota \coshp = \id$.
\end{proof}

\subsection{Definitional equality} As definitional equality is interpreted as equality, it is evidently consistent to assume that this is an equivalence relation and a congruence. The conversion rule is also obvious.

\subsection{Quantification}
\subsubsection{Continuous quantification}
In general, write $\subext \iota = (\wknvar x, \iota(\var x)/\ftrvar \sharp x) : (\Gamma, \var x : T) \to (\Gamma, \ftrvar \sharp x : \sharp T[\iota])$. If $\Gamma = \sharp \Delta$, then this becomes $\subext \iota : (\sharp \Delta, \var x : T) \to (\sharp \Delta, \ftrvar \sharp x : \sharp T) = \sharp (\Delta, \var x : T[\iota])$. For the continuous quantifiers, we have
\begin{align*}
	&\interp{
		\inference{
			\Gamma \judtm{A}{\El\,\uni \ell} \qquad
			\Gamma, \ctxctu x {\El\,A} \judtm{B}{\El\,\uni \ell}
		}{\Gamma \judtm{\prodctu x A.B}{\El\,\uni \ell}}{t-$\Pi$}\quad
	} = \nn\\ &
	\inference{
	\inference{
		\inference{
			\interp \Gamma \sez \interp A : \uniDD_\ell
		}{\sharp \shp \interp \Gamma \sez \ElDD \interp A \dtype_\ell}{}
		\qquad
		\inference{
		\inference{
		\inference{
			\interp \Gamma, \var x : (\ElDD \interp A)[\sharp \varsigma][\iota] \sez \interp B : \uniDD_\ell
		}{\sharp \shp \paren{\interp \Gamma, \var x : (\ElDD \interp A)[\sharp \varsigma][\iota]} \sez \ElDD \interp B \dtype_\ell}{}
		}{\sharp \paren{\shp \interp \Gamma, \var x : (\ElDD \interp A)[\iota]} \sez \ElDD \interp B[\sharp(\subext \varsigma)] \dtype_\ell}{}
		}{\sharp \shp \interp \Gamma, \var x : \ElDD \interp A \sez \ElDD \interp B[\sharp(\subext \varsigma)][\subext \iota] \dtype_\ell}{}
	}{\sharp \shp \Gamma \sez \Pi(\var x : \ElDD \interp A).\ElDD \interp B[\sharp(\subext \varsigma)][\subext \iota] \dtype_\ell}{}
	}{\Gamma \sez \tycodeDD{\Pi(\var x : \ElDD \interp A).(\ElDD \interp B[\sharp(\subext \varsigma)][\subext \iota])} : \uniDD_\ell}{}
\end{align*}
and similar for $\Sigma$. Note that
\begin{align}
	\interp{\El~\prodctu x A.B}
	&= \ElDD \interp{\prodctu x A.B} [\sharp \varsigma]
	= \Pi(\var x : \ElDD \interp A).(\ElDD \interp B[\sharp(\subext \varsigma)][\subext \iota])~[\sharp \varsigma] \\
	&= \Pi(\var x : \ElDD \interp A [\sharp \varsigma]).(\ElDD \interp B[\sharp(\subext \varsigma)][\subext \iota][\sharp \varsigma \subext]) \nn\\
	&= \Pi(\var x : \ElDD \interp A [\sharp \varsigma]).(\ElDD \interp B[\sharp(\subext \varsigma \circ \varsigma \subext)][\subext \iota]) \nn\\
	&= \Pi(\var x : \ElDD \interp A [\sharp \varsigma]).(\ElDD \interp B[\sharp \varsigma][\subext \iota])
	= \Pi(\var x : \interp{\El\,A}).(\interp{\El\,B} [\subext \iota]). \nn
\end{align}
If we further apply $[\iota]$, we find
\begin{equation}
	\interp{\El~\prodctu x A.B}[\iota] = \Pi(\var x : \interp{\El~A}[\iota]).(\interp{\El~B}[\subext \iota][\iota \subext]) = \Pi(\var x : \interp{\El~A}[\iota]).(\interp{\El~B}[\iota]).
\end{equation}
\subsubsection{Parametric quantification}
For the existential type $\Sigmapar$, we have
\begin{align*}
	& \interp{
		\inference{
			\Gamma \judtm{A}{\El\,\uni \ell} \qquad
			\Gamma, \ctxctu x {\El~A} \judtm{B}{\El\,\uni \ell}
		}{\Gamma \judtm{\sumpar x A.B}{\El\,\uni \ell}}{t-$\Sigma$}
	} = \nn\\ &
	\inference{
	\inference{
	\inference{
		\inference{
		\inference{
			\interp \Gamma \sez \interp A : \uniDD_\ell
		}{\sharp \shp \interp \Gamma \sez \ElDD \interp A \dtype_\ell}{}
		}{\sharp \shp \interp \Gamma \sez \sharp \ElDD \interp A \type_\ell}{}
		\qquad
		\inference{
		\inference{
			\interp \Gamma, \var x : (\ElDD \interp A)[\sharp \varsigma][\iota] \sez \interp B : \uniDD_\ell
		}{\sharp \shp \paren{\interp \Gamma, \var x : (\ElDD \interp A)[\sharp \varsigma][\iota]} \sez \ElDD \interp B \dtype_\ell}{}
		}{\sharp \shp \interp \Gamma, \ftrvar \sharp x : \sharp \ElDD \interp A \sez \ElDD \interp B [\sharp(\subext\varsigma)] \dtype_\ell}{}
	}{\sharp \shp \Gamma \sez \Sigma(\ftrvar \sharp x : \sharp \ElDD \interp A).(\ElDD \interp B [\sharp (\subext \varsigma)]) \type_\ell}{}
	}{\sharp \shp \Gamma \sez \quotshp \Sigma(\ftrvar \sharp x : \sharp \ElDD \interp A).(\ElDD \interp B [\sharp (\subext \varsigma)]) \dtype_\ell}{}
	}{\Gamma \sez \tycodeDD{\quotshp \Sigma(\ftrvar \sharp x : \sharp \ElDD \interp A).(\ElDD \interp B [\sharp (\subext \varsigma)])} : \uniDD_\ell}{}
\end{align*}
Then we have
\begin{align*}
	\interp{\El~\sumpar x A.B}
	&= \ElDD \interp{\sumpar x A.B} [\sharp \varsigma]
	= \paren{\quotshp \Sigma(\ftrvar \sharp x : \sharp \ElDD \interp A).(\ElDD \interp B [\sharp (\subext \varsigma)])} [\sharp \varsigma] \\
	&= \quotshp \Sigma(\ftrvar \sharp x : (\sharp \ElDD \interp A)[\sharp \varsigma]).(\ElDD \interp B [\sharp (\subext \varsigma)] [\sharp (\varsigma \subext)]) \\
	&= \quotshp \Sigma(\ftrvar \sharp x : \sharp(\ElDD \interp A[\sharp \varsigma])).(\ElDD \interp B [\sharp \varsigma])
	= \quotshp \Sigma(\ftrvar \sharp x : \sharp \interp{\El\,A}).\interp{\El\,B}.
\end{align*}
The universal type $\Pipar$ is defined similarly, but without throwing in $\quotshp$:
\begin{equation}
	\interp{\prodpar x A.B} = \tycodeDD{\Pi(\ftrvar \sharp x : \sharp \ElDD \interp A).(\ElDD \interp B [\sharp(\subext \varsigma)])}.
\end{equation}
Then we have
\begin{equation}
	\interp{\El~\prodpar x A.B} = \Pi(\ftrvar \sharp x : \sharp \interp{\El\,A}).\interp{\El\,B}.
\end{equation}
If we further apply $[\iota]$, we find
\begin{equation}
	\interp{\El~\prodpar x A.B} [\iota] = \Pi(\ftrvar \sharp x : (\sharp \interp{\El\,A})[\iota]).(\interp{\El\,B}[\iota \subext]) = \Pi(\ftrvar \sharp x : (\sharp \interp{\El\,A})[\iota]).(\interp{\El\,B}[\iota]).
\end{equation}
In the second step, we use that $\iota \subext = (\iota \wknvar x, \var x/\var x)$ is equal to $\iota = (\iota \wknvar x, \iota(\var x)/\var x)$ because $\var x$ has a sharp type and $\iota\sharp = \id$.

\subsubsection{Pointwise quantification}
We generalize the $\subext \nu = (\wknvar x, \nu(\var x)/\var x)$ notation to any natural transformation $\nu$ between morphisms of CwFs. We have
\begin{align*}
	&\interp{
		\inference{
			\Gamma \judtm{A}{\El\,\uni \ell} \qquad
			\Gamma, \ctxptw x {\El\,A} \judtm{B}{\El\,\uni \ell}
		}{\Gamma \judtm{\prodptw x A.B}{\El\,\uni \ell}}{t-$\Pi$}
	} = \nn\\ &
	\inference{
	\inference{
		\inference{
		\inference{
			\interp \Gamma \sez \interp A : \uniDD_\ell
		}{\sharp \shp \interp \Gamma \sez \ElDD \interp A \dtype_\ell}{}
		}{\sharp \shp \interp \Gamma \sez \coshp \ElDD \interp A \dtype_\ell}{}
		\qquad
		\inference{
		\inference{
			\interp \Gamma, \ftrvar \coshp x : (\coshp(\ElDD \interp A [\sharp \varsigma]))[\iota] \sez \interp B : \uniDD_\ell
		}{\sharp \shp \paren{\interp \Gamma, \ftrvar \coshp x : (\coshp(\ElDD \interp A [\sharp \varsigma]))[\iota]} \sez \ElDD \interp B \dtype_\ell}{}
		}{\sharp \shp \interp \Gamma, \ftrvar \coshp x : \coshp \ElDD \interp A \sez \ElDD \interp B [\sharp(\subext \varsigma)] \dtype_\ell}{}
	}{\sharp \shp \interp \Gamma \sez \Pi(\ftrvar \coshp x : \coshp \ElDD \interp A).(\ElDD \interp B [\sharp(\subext \varsigma)]) \dtype_\ell}{}
	}{\interp \Gamma \sez \tycodeDD{\Pi(\ftrvar \coshp x : \coshp \ElDD \interp A).(\ElDD \interp B [\sharp(\subext \varsigma)])} : \uniDD_\ell}{}
\end{align*}
and similar for $\Sigma$. The step where we use $\sharp(\subext \varsigma)$ is a bit obscure. We have
\begin{equation}
	\subext \varsigma :
	\paren{\shp \interp \Gamma, \ftrvar \coshp x : (\coshp(\ElDD \interp A))[\iota]}
	\to
	\shp \paren{\interp \Gamma, \ftrvar \coshp x : (\coshp(\ElDD \interp A [\sharp \varsigma]))[\iota]}
\end{equation}
because
\begin{equation}
	(\coshp(\ElDD \interp A))[\iota][\varsigma]
	= (\coshp(\ElDD \interp A))[\sharp \varsigma][\iota]
	= (\coshp(\ElDD \interp A [\sharp \varsigma]))[\iota].
\end{equation}
Now if we apply $\sharp$ on the domain of $\subext \varsigma$, then $\sharp \iota = \id$ disappears, and $\coshp$ absorbs $\sharp$ on its left.

We will have
\begin{equation}
	\interp{\El~\prodptw x A.B} = \Pi(\ftrvar \coshp x : \coshp \interp{\El~A}).\interp{\El~B}
\end{equation}
and similar for $\Sigma$. If we further apply $[\iota]$, we find
\begin{equation}
	\interp{\El~\prodptw x A.B}[\iota] = \Pi(\ftrvar \coshp x : (\coshp \interp{\El~A})[\iota]).(\interp{\El~B}[\iota \subext]) = \Pi(\ftrvar \coshp x : (\coshp \interp{\El~A})[\iota]).(\interp{\El~B}[\iota]),
\end{equation}
where $\iota \subext = \iota$ because $\var x$ has a type in the image of $\coshp$ and $\iota \coshp = \id$.

So in general, we see that
\begin{equation}
	\interp{\El~\prodvar \mu x A.B}[\iota] = \Pi(\ftrvar \mu x : (\mu \interp{\El~A})[\iota]).(\interp{\El~B}[\iota]).
\end{equation}

\subsection{Functions}
Abstraction is interpreted as
\begin{equation}
	\interp{
		\inference{
			\Gamma, \ctxvar \mu x {\El~A} \judtm b {\El~B}
		}{\Gamma \judtm{\lamannotvar \mu x A.b}{\El~\prodvar \mu x A.B}}{t-$\lambda$}
	} =
	\inference{
		\interp \Gamma, \ftrvar \mu x : (\mu \interp{\El~A})[\iota] \sez \interp b : \interp{\El~B}[\iota]
	}{\interp \Gamma \sez \lambda \var x . \interp b : \Pi(\ftrvar \mu x : (\mu \interp{\El~A})[\iota]).(\interp{\El~B}[\iota])}{}.
\end{equation}

Application is interpreted as
\begin{align*}
	&\interp{
		\inference{
			\Gamma \judtm{f}{\El~\prodvar \mu x A.B} \qquad
			\mu \setminus \Gamma \judtm{a}{\El~A}
		}{\Gamma \judtm{\apvar \mu f a}{\El~B[\varclr a / \varclr x]}}{t-ap}
	} = \nn\\ &
	\inference{
		\interp \Gamma \sez \interp f : \Pi(\ftrvar \mu x : (\mu \interp{\El~A})[\iota]).(\interp{\El~B}[\iota]) \qquad
		\inference{
			\interp{\mu \setminus \Gamma} \sez \interp a : \interp{\El~A}[\iota]
		}{\interp \Gamma \sez \forsub \mu {\interp a} : (\mu \interp{\El~A})[\iota]}{}
	}{\interp \Gamma \sez \interp f (\forsub \mu {\interp a}) : (\mu \interp{\El~A})[\iota][\id, \forsub \mu \interp a / \ftrvar \mu x]}{}
\end{align*}

The $\beta$-rule looks like this:
\begin{align*}
	& \interp{
		\inference{
			\Gamma, \ctxvar \mu x {\El~A} \judtm b {\El~B} \qquad
			\mu \setminus \Gamma \judtm{a}{\El~A}
		}{\Gamma \judtmeq{(\lamannotvar \mu x A.b)\apvar \mu{}{a}}{b[\varclr a / \varclr x]}{\El~B[\varclr a / \varclr x]}}{}
	} = \nn\\ &
	\inference{
		\interp \Gamma, \ftrvar \mu x : (\mu \interp{\El~A})[\iota] \sez \interp b : \interp{\El\,B}[\iota] \qquad
		\inference{
			\interp{\mu \setminus \Gamma} \sez \interp a : \interp{\El~A}[\iota]
		}{\interp \Gamma \sez \forsub \mu \interp a : (\mu \interp{\El~A})[\iota]}{}
	}{\interp \Gamma \sez (\lambda \ftrvar \mu x . \interp b) (\forsub \mu \interp a) = \interp b [\id, \forsub \mu \interp a / \ftrvar \mu x] : \interp{\El~B}[\iota][\id, \forsub \mu \interp a / \ftrvar \mu x]}{}
\end{align*}
and follows from \cref{def:pi-types}. The $\eta$-rule is:
\begin{align*}
	& \interp{
		\inference{
			\Gamma \judtm{f}{\El~\prodvar \mu x A.B}
		}{\Gamma \judtmeq{\lamannotvar \mu x A.\apvar \mu f x}{f}{\El~\prodvar \mu x A.B}}{}
	} = \nn\\ &
	\inference{
		\interp \Gamma \sez \interp f : \Pi(\ftrvar \mu x : \interp{\El~A}[\iota]).(\interp{\El~B}[\iota])
	}{\interp \Gamma \sez \lambda \ftrvar \mu x . (\interp f [\wkn{\ftrvar \mu x}]) (\forsub \mu \interp x) = \interp f : \Pi(\ftrvar \mu x : \interp{\El~A}[\iota]).(\interp{\El~B}[\iota])}{}
\end{align*}
This rule also follows from \cref{def:pi-types}, because $\forsub \mu \interp x = \ftrvar \mu x$.

\subsection{Pairs}
For pair formation, we have quite straightforwardly:
\begin{align*}
	& \interp{
		\inference{
			\Gamma \judty{\El~\sumvar \mu x A.B} \qquad
			\mu \setminus \Gamma \judtm a {\El~A} \qquad
			\Gamma \judtm b {\El~B[\varclr a/\varclr x]}
		}{\Gamma \judtm{\pairvar \mu a b}{\El~\sumvar \mu x A.B}}{t-pair}
	} = \nn\\ &
	\inference{
		\inference{
			\interp{\mu \setminus \Gamma} \sez \interp a : \interp{\El~A}[\iota]
		}{\interp{\Gamma} \sez \forsub \mu \interp a : (\mu \interp{\El~A})[\iota]}{}
		\qquad
		\interp \Gamma \sez \interp b : \interp{\El~B}[\iota][\id, \forsub \mu \interp a/\ftrvar \mu x]
	}{\interp \Gamma \sez (\forsub \mu \interp a, \interp b) : \Sigma(\ftrvar \mu x : (\mu \interp{\El~A})[\iota]).(\interp{\El~B}[\iota])}{}
\end{align*}
For $\mu = \sharp$, we have to further apply $\hatinquotshp$.

The type is included in the syntactical rule to ensure that $B$ is actually a well-defined type. We also need to know that in the model, but we did not write it explicitly in the interpretation. Note that we are not in fact using the interpretation of the existence of the $\Sigma$-type; rather, we use that if the $\Sigma$-type exists, then admissibly $B$ is a type.

\subsubsection{Projections for continuous pairs}
Instead of interpreting the eliminator, we interpret the first and second projections:
\begin{equation}
	\interp{
		\inference{
			\Gamma \judtm{p}{\sumctu x A.B}
		}{\Gamma \judtm{\fst~p}{A}}{}
	} =
	\inference{
		\interp \Gamma \sez \interp p : \Sigma(\var x : \interp{\El~A}[\iota]).(\interp{\El~B}[\iota])
	}{\interp \Gamma \sez \fst \interp p : \interp{\El~A}[\iota]}{},
\end{equation}
\begin{equation}
	\interp{
		\inference{
			\Gamma \judtm{p}{\sumctu x A.B}
		}{\Gamma \judtm{\snd~p}{B[\fst~p/x]}}{}
	} =
	\inference{
		\interp \Gamma \sez \interp p : \Sigma(\var x : \interp{\El~A}[\iota]).(\interp{\El~B}[\iota])
	}{\interp \Gamma \sez \snd \interp p : \interp{\El~B}[\iota][\id, \fst \interp p / \var x]}{}.
\end{equation}
Then $\beta$- and $\eta$-rules lift from the model.

\subsubsection{Projections for pointwise pairs}
\begin{equation}
	\interp{
		\inference{
			\sharp \setminus \Gamma \judtm{p}{\El~\sumptw x A.B}
		}{\Gamma \judtm{\fstptw p}{\El~A}}{}
	} =
	\inference{
	\inference{
	\inference{
		\interp{\sharp \setminus \Gamma} \sez \interp p : \Sigma(\ftrvar \coshp x : (\coshp \interp{\El~A})[\iota]).(\interp{\El~B}[\iota])
	}{\interp \Gamma \sez \forsub \sharp \interp p : \Sigma(\ftrvar \coshp x : (\coshp \interp{\El~A})[\iota]).(\sharp \interp{\El~B}[\iota])}{}
	}{\interp \Gamma \sez \fst (\forsub \sharp \interp p) : (\coshp \interp{\El~A})[\iota]}{}
	}{\interp \Gamma \sez \vartheta(\fst (\forsub \sharp \interp p)) : \interp{\El~A}[\iota]}{},
\end{equation}
\begin{equation}
	\interp{
		\inference{
			\Gamma \judtm{p}{\El~\sumptw x A.B}
		}{\Gamma \judtm{\sndptw p}{\El~B[\fstptw p / x]}}{}
	} =
	\inference{
		\interp \Gamma \sez \interp p : \Sigma(\ftrvar \coshp x : (\coshp \interp{\El~A})[\iota]).(\interp{\El~B}[\iota])
	}{\interp \Gamma \sez \snd \interp p : \interp{\El~B}[\iota][\id, \fst \interp p / \ftrvar \coshp x]}{}
\end{equation}
One can show that $\fst \interp p$ is the appropriate term to appear in the substitution, for the conclusion to be well-typed.

\subsubsection{Elimination of parametric pairs}
We have to interpret the rule
\begin{equation}
		\inference{
			\Gamma, \ctxvar \nu z {\El~\sumpar x A.B} \judty{\El~C} \\
			\Gamma, \ctxpar{x}{\El~A}, \ctxvar{\nu}{y}{\El~B} \judtm{c}{\El~C[\varclr{\pairpar x y}/\varclr z]} \\
			\nu \setminus \Gamma \judtm{p}{\El~\sumpar x A.B}
		}{\Gamma \judtm{\ind^\nu_{\Sigmapar}(\varclr z.\parclr C, \parclr x.\varclr y.c, \varclr p)}{\El~C[\varclr p / \varclr z]}}{t-indpair}.
\end{equation}
Thanks to the $\forsub \nu$ operator, it is sufficient to interpret
\begin{equation}
	\inference{
		\Gamma, \ctxvar \nu z {\El~\sumpar x A.B} \judty{\El~C} \\
		\Gamma, \ctxpar{x}{\El~A}, \ctxvar{\nu}{y}{\El~B} \judtm{c}{\El~C[\varclr{\pairpar x y}/\varclr z]}
	}{\Gamma, \ctxvar \nu z {\El~\sumpar x A.B} \judtm{\ind^\nu_{\Sigmapar}(\varclr z.\parclr C, \parclr x.\varclr y.c, \varclr z)}{\El~C}}{t-indpair}.
\end{equation}
We have
\begin{equation*}
	\inference{
	\inference{
	\inference{
		\interp \Gamma, \ftrvar \sharp x : (\sharp \interp{\El~A})[\iota], \ftrvar \nu y : (\nu \interp{\El~B})[\iota] \sez c : \interp{\El~C}[\iota][\wkn{\ftrvar \sharp x} \wkn{\ftrvar \nu y}, (\nu \hatinquotshp)(\ftrvar \sharp x, \ftrvar \nu y) / \ftrvar \nu z]
	}{\interp \Gamma, \ftrvar \nu z : \Sigma(\ftrvar{\sharp} x : (\sharp \interp{\El~A})[\iota]).((\nu \interp{\El~B})[\iota]) \sez c [\wkn{\ftrvar \nu z}, \fst~\ftrvar \nu z / \ftrvar \sharp x, \snd~\ftrvar \nu z/\ftrvar \nu y] : \interp{\El~C}[\iota][\wkn{\ftrvar \nu z}, (\nu \hatinquotshp)(\ftrvar \nu z)/\ftrvar \nu z]}{}
	}{\interp \Gamma, \ftrvar \nu z : \paren{\nu \Sigma(\ftrvar \sharp x : \sharp \interp{\El~A}).\interp{\El~B}} [\iota] \sez c [\wkn{\ftrvar \nu z}, \fst~\ftrvar \nu z / \ftrvar \sharp x, \snd~\ftrvar \nu z/\ftrvar \nu y] : \interp{\El~C}[\iota][\wkn{\ftrvar \nu z}, (\nu \hatinquotshp)(\ftrvar \nu z)/\ftrvar \nu z]}{}
	}{\interp \Gamma, \ftrvar \nu z : \paren{\nu \quotshp \Sigma(\ftrvar \sharp x : \sharp \interp{\El~A}).\interp{\El~B}} [\iota] \sez c [\wkn{\ftrvar \nu z}, \fst~\ftrvar \nu z / \ftrvar \sharp x, \snd~\ftrvar \nu z/\ftrvar \nu y] [\wkn{\ftrvar \nu z}, (\nu \hatinquotshp)(\ftrvar \nu z)/\ftrvar \nu z]\inv : \interp{\El~C}[\iota]}{}
\end{equation*}
This is best read bottom-up. In the last step, we get rid of $\quotshp$ using \cref{thm:elim-quotshp}. In the middle, we simply rewrite the context using that $\Sigma$ commutes with lifted functors such as $\nu$, and that $\iota \sharp = \id$ to turn $[\iota \subext]$ into $[\iota]$ on the $\Sigma$-type's codomain. Above, we split up $\ftrvar \nu z$ in its components.

In order to see that the substitution in the type of the premise is correct, note that we have $\interp{  \nu \setminus (\Gamma, \ctxpar x {\El~A}, \ctxvar \nu y {\El~B}) \judtm{\pairpar x y}{\sumpar x A.B}  }$. Interpreting this and applying $\forsub \nu$, one finds $(\nu \hatinquotshp)(\ftrvar \sharp x, \ftrvar \nu y)$ after working through some tedious case distinctions. 

\subsection{Identity types}\label{sec:idtp-semantics}
We have
\begin{align*}
	& \interp{
		\inference{
			\Gamma \judtm{A}{\El~\uni \ell} \qquad
			\Gamma \judtm{a, b}{\El~A}
		}{\Gamma \judtm{a \idtp A b}{\El~\uni \ell}}{t-Id}
	} = \nn\\ &
	\inference{
	\inference{
		\inference{
			\interp \Gamma \sez \interp A : \uniDD_\ell
		}{\sharp \shp \interp \Gamma \sez \ElDD \interp A \dtype_\ell}{} \qquad
		\inference{
		\inference{
			\interp \Gamma \sez \interp a, \interp b : \ElDD \interp A [\sharp \varsigma] [\iota]
		}{\shp \interp \Gamma \sez \interp a [\varsigma]\inv, \interp b[\varsigma]\inv : \ElDD \interp A[\iota]}{}
		}{\sharp \shp \interp \Gamma \sez \ftrtm{\sharp}{(\interp a [\varsigma]\inv)}, \ftrtm{\sharp}{(\interp b [\varsigma]\inv)} : \sharp \ElDD \interp A}{}
	}{\sharp \shp \interp \Gamma \sez \ftrtm{\sharp}{(\interp a [\varsigma]\inv)} \idtp{\sharp \ElDD \interp A} \ftrtm{\sharp}{(\interp b [\varsigma]\inv)} \dtype_\ell}{}
	}{\interp \Gamma \sez \tycodeDD{\ftrtm{\sharp}{(\interp a [\varsigma]\inv)} \idtp{\sharp \ElDD \interp A} \ftrtm{\sharp}{(\interp b [\varsigma]\inv)}} : \uniDD_\ell}{}
\end{align*}
Observe:
\begin{align*}
	\interp{\El~a \idtp A b}
	&= \paren{\ftrtm{\sharp}{(\interp a [\varsigma]\inv)} \idtp{\sharp \ElDD \interp A} \ftrtm{\sharp}{(\interp b [\varsigma]\inv)}} [\sharp \varsigma]
	= \paren{\ftrtm{\sharp}{\interp a} \idtp{\sharp \interp{\El\,A}} \ftrtm{\sharp}{\interp b}}.
\end{align*}
If we further apply $[\iota]$, we get
\begin{equation}
	\interp{\El~a \idtp A b}[\iota]
	= \paren{ \ftrtm{\sharp}{\interp a}[\iota] \idtp{(\sharp \interp{\El\,A})[\iota]} \ftrtm{\sharp}{\interp b}[\iota] }.
\end{equation}

For reflexivity, we have
\begin{equation}
	\interp{
		\inference{
			\leftflat \Gamma \judtm a {\El~A}
		}{\Gamma \judtm{\refl\,\parclr a}{\El~a \idtp A a}}{t-refl}
	} =
	\inference{
	\inference{
		\interp{\sharp \setminus \Gamma} \sez \interp a : \interp{\El\,A} [\iota]
	}{\interp \Gamma \sez \forsub \sharp \interp a : (\sharp \interp{\El~A})[\iota]}{}
	}{\interp \Gamma \sez \refl~(\forsub \sharp \interp a) : \forsub \sharp \interp a \idtp{(\sharp \interp{\El~A})[\iota]} \forsub \sharp \interp a}{}
\end{equation}
where $\forsub \sharp t = \ftrtm \sharp t [\iota]$.

We also need to interpret the $\J$-rule:
\begin{equation}
	\inference{
		\leftflat \Gamma \judtm{a, b}{\El~A} \qquad
		\Gamma, \ctxpar{y}{\El~A}, \ctxvar \nu {w}{\El~a \idtp A y} \judty {\El~C} \\
		\nu \setminus \Gamma \judtm{e}{\El~a \idtp A b} \qquad
		\Gamma \judtm{c}{\El~C[\parclr a/\parclr y, \varclr{\refl~a}/\varclr w]}
	}{\Gamma \judtm{
		\J^\nu(\parclr a, \parclr b, \parclr y.\varclr w.\parclr C, \varclr e, c)
	}{\El~C[\parclr b/\parclr y, \varclr e/\varclr w]}}{t-J}
\end{equation}
First of all, note that if $\nu \in \accol{\coshp, \idmod, \sharp}$, then $\nu \interp{\El~a \idtp A b} = \interp{\El~a \idtp A b}$, because
\begin{equation}
	\nu \interp{\El~a \idtp A b}
	= \nu \paren{\ftrtm{\sharp}{\interp a} \idtp{\sharp \interp{\El\,A}} \ftrtm{\sharp}{\interp b}}
	= \paren{\ftrtm{\nu \sharp}{\interp a} \idtp{\nu \sharp \interp{\El\,A}} \ftrtm{\nu \sharp}{\interp b}}
	= \paren{\ftrtm{\sharp}{\interp a} \idtp{\sharp \interp{\El\,A}} \ftrtm{\sharp}{\interp b}}.
\end{equation}
Hence, we may assume that $\nu = \idmod$. We then have
\begin{equation}
	\inference{
		\interp \Gamma \sez \forsub \sharp \interp a, \forsub \sharp \interp b : (\sharp \interp{\El~A})[\iota] \\
		\interp \Gamma, \ftrvar \sharp y : (\sharp \interp{\El~A})[\iota], \var w : \forsub \sharp \interp a \idtp{(\sharp \interp{\El~A})[\iota]} \ftrvar \sharp y \sez \interp{\El~C}[\iota] \dtype \\
		\interp \Gamma \sez \interp e : \forsub \sharp \interp a \idtp{(\sharp \interp{\El~A})[\iota]} \forsub \sharp \interp b \\
		\interp \Gamma \sez \interp c : \interp{\El~C}[\iota][\id, \forsub \sharp \interp a / \ftrvar \sharp y, \refl~(\forsub \sharp \interp a) / \var w]
	}{\interp \Gamma \sez \J(\forsub \sharp \interp a, \forsub \sharp \interp b, \ftrvar \sharp y.\var w.\interp{\El~C}[\iota], \interp e, \interp c) : \interp{\El~C}[\iota][\id, \forsub \sharp \interp b / \ftrvar \sharp y, \interp e / \var w]}{}
\end{equation}

\subsubsection{The reflection rule}
\begin{remark}[Erratum]
	The original version of this report, claimed to prove the reflection rule:
	\begin{equation}
		\inference{
			\Gamma \judtm{a, b}{\El~A} \qquad
			\Gamma \judtm{e}{\El~ a \idtp A b}
		}{\Gamma \judtmeq a b {\El~A}}{t-rflct}.
	\end{equation}
	\begin{proof}[Erroneous proof]
		Indeed, we have
		\begin{equation}
			\inference{
				\interp \Gamma \sez \interp e : \ftrtm \sharp{\interp a}[\iota] \idtp{\sharp (\ElDD \interp{A} [\sharp \varsigma]) [\iota]} \ftrtm \sharp{\interp b}[\iota]
			}{\interp \Gamma \sez \ftrtm \sharp{\interp a}[\iota] = \ftrtm \sharp{\interp b}[\iota] : \sharp(\ElDD \interp{A} [\sharp \varsigma]) [\iota]}{}
		\end{equation}
		Now for any $\gamma : \DSub{W}{\interp \Gamma}$, we have
		\begin{equation}
			\ftrtm \sharp{\interp a}[\iota]\dsub \gamma
			= \ftrtm \sharp{\interp a}\dsub{\fpshadj \flat(\gamma \kappa)}
			= \fpshadj \flat(\interp a \dsub{\gamma \kappa}).
		\end{equation}
		Hence, we can conclude that for any $\gamma$, we have $\flat W \Dsez \interp a \dsub{\gamma \kappa} = \interp b \dsub{\gamma \kappa} : \ElDD \interp{A} [\sharp \varsigma] [\iota] \dsub{\gamma \kappa}$. This means that $\interp a$ and $\interp b$ are equal on bridges, but maybe not on paths.
		
		The type $\ElDD \interp{A} [\sharp \varsigma] [\iota] \dsub{\gamma \kappa}$ we wrote there is correct, because
		\begin{equation}
			\fpshadj \flat(\iota \interp \Gamma \circ \gamma \circ \kappa W)
			= \fpshadj \flat(\fpshadj \flat(\gamma \circ \kappa W \circ \kappa \flat W))
			= \fpshadj \flat(\fpshadj \flat(\gamma \circ \kappa W))
			= \fpshadj \flat(\gamma \circ \kappa W)
			= \iota \interp \Gamma \circ \gamma,
		\end{equation}
		so that if $\flat W \Dsez t : \ElDD \interp A [\sharp \varsigma] \dsub{\iota \gamma \kappa}$, then $W \Dsez \fpshadj \flat(t) : \sharp(\ElDD \interp A [\sharp \varsigma]) \dsub{\iota \gamma}$.
		
		By discreteness, we can form
		\begin{equation}\label{eq:error-reflection}
			\inference{
				\interp \Gamma \sez \interp a, \interp b : \ElDD \interp A [\sharp \varsigma][\iota]
			}{\shp \interp \Gamma \sez \interp a [\varsigma]\inv, \interp b[\varsigma]\inv : \ElDD \interp A [\iota]}{}.
		\end{equation}
		Clearly if we can show $\interp a [\varsigma]\inv = \interp b[\varsigma]\inv$, then $\interp a = \interp b$. This is essentially saying that if $\interp a$ and $\interp b$ act the same way on bridges, then they are equal, so we should be almost there.
		
		Pick $\fpshadj \shp(\overline \gamma) : \DSub{W}{\shp \Gamma = \flat \quotshp \Gamma}$, i.e. $\overline \gamma : \DSub{\shp W}{\quotshp \Gamma}$, i.e. $\gamma : \DSub{\shp W}{\Gamma}$. Recall that $\varsigma = (\kappa \quotshp)\inv \circ \inquotshp$. Now
		\begin{equation}
			\kappa \quotshp \Gamma \circ \fpshadj \shp(\overline \gamma) = \overline \gamma \circ \varsigma W = \inquotshp(\gamma) \circ \varsigma W
		\end{equation}
		so we have $\varsigma \Gamma \circ \gamma \circ \varsigma W = (\kappa \quotshp \Gamma)\inv \circ \inquotshp \gamma \circ \varsigma W = \fpshadj \shp(\overline \gamma)$. Hence
		\begin{align*}
			\interp a [\varsigma]\inv \dsub{\fpshadj \shp(\overline \gamma)}
			= \interp a \dsub{\gamma \circ \varsigma W} &= \interp{a} \dsub{\gamma \circ \kappa \shp W} \psub{\varsigma W} \\ \nn 
			&= \interp b \dsub{\gamma \circ \kappa \shp W} \psub{\varsigma W} = \ldots = \interp b [\varsigma]\inv \dsub{\fpshadj \shp(\overline \gamma)}.
		\end{align*}
		Because $\fpshadj \shp (\overline \gamma)$ is a fully general defining subsitution of $\shp \Gamma$, we can conclude that $\interp a [\varsigma]\inv = \interp b[\varsigma]\inv$ and hence $\interp a = \interp b$.
	\end{proof}
	The error is in \cref{eq:error-reflection}. The idea there is that $\sharp \varsigma \circ \iota = \iota \circ \varsigma$, so that $\interp a$ and $\interp b$ would a type of the form $X[\varsigma]$ and the application of $[\varsigma]\inv$ would be valid. However, when we make some implicit arguments explicit we see that the premise is:
	\begin{equation}
		\interp \Gamma \sez \interp a, \interp b : \ElDD \interp A [\sharp \varsigma \interp{\sharp \setminus \Gamma}][\iota \interp \Gamma], \tagbis
	\end{equation}
	so that even though $\sharp \varsigma \circ \iota = \iota \circ \varsigma$ (as an equation of natural transformations), the equation cannot be applied here as the natural transformations have been instantiated on different contexts. The context $\interp{\sharp \setminus \Gamma} \cong \flat \interp \Gamma$ is discrete, i.e. all its paths are constant. By mixing it up with the context $\interp \Gamma$, we ended up assuming that $\interp \Gamma$ is discrete (which need not be the case if it contains parametric variables). Since we already knew that $\interp a$ and $\interp b$ are equal on bridges, we could then conclude that they are completely equal, as there would be only trivial paths.
	
	We see two ways to fix this erratum. Either we actually move back to the context $\sharp \setminus \Gamma$ in the reflection rule's conclusion, making it weaker and breaking the proof of function extensionality for parametric arguments. Or we conjecture the principle of pathhood irrelevance --- that any bridge can be a path in at most one way --- so that equality in $\sharp \setminus \Gamma$ entails equality in $\Gamma$.
\end{remark}

\paragraph{Solution 1: A more careful reflection rule}
\begin{lemma}\label{thm:careful-reflection}
	The model supports the following reflection rule:
	\begin{equation}
		\inference{
			\sharp \setminus \Gamma \judtm{a, b}{\El~A} \qquad
			\Gamma \judtm{e}{\El~ a \idtp A b}
		}{\sharp \setminus \Gamma \judtmeq a b {\El~A}}{}. \tagbis
	\end{equation}
\end{lemma}
\begin{proof}
	We have
	\begin{equation}
		\inference{
		\inference{
		\inference{
		\inference{
		\inference{
			\interp \Gamma \sez \interp e : \ftrtm \sharp{\interp a}[\iota \interp \Gamma] \idtp{\sharp (\ElDD \interp{A} [\sharp \varsigma \interp{\sharp \setminus \Gamma}]) [\iota \interp \Gamma]} \ftrtm \sharp{\interp b}[\iota \interp \Gamma]
		}{\interp \Gamma \sez \ftrtm \sharp{\interp a}[\iota \interp \Gamma] = \ftrtm \sharp{\interp b}[\iota \interp \Gamma] : \sharp(\ElDD \interp{A} [\sharp \varsigma \interp{\sharp \setminus \Gamma}]) [\iota \interp \Gamma]}{}
		}{\flat \interp \Gamma \sez \ftrtm \flat{\interp a} = \ftrtm \flat{\interp b} : \flat (\ElDD \interp{A} [\sharp \varsigma \interp{\sharp \setminus \Gamma}][\iota \interp{\sharp \setminus \Gamma}])}{}
		}{\flat \interp \Gamma \sez \kappa(\ftrtm \flat{\interp a}) = \kappa(\ftrtm \flat{\interp b}) : \ElDD \interp{A} [\sharp \varsigma \interp{\sharp \setminus \Gamma}] [\iota \interp{\sharp \setminus \Gamma}] [\kappa \interp{\sharp \setminus \Gamma}]}{}
		}{\flat \interp{\sharp \setminus \Gamma} \sez \interp a [\kappa \interp{\sharp \setminus \Gamma}] = \interp b [\kappa \interp{\sharp \setminus \Gamma}] : \ElDD \interp{A} [\sharp \varsigma \interp{\sharp \setminus \Gamma}] [\iota \interp{\sharp \setminus \Gamma}] [\kappa \interp{\sharp \setminus \Gamma}]}{}
		}{\interp{\sharp \setminus \Gamma} \sez \interp a = \interp b : \ElDD \interp{A} [\sharp \varsigma \interp{\sharp \setminus \Gamma}] [\iota \interp{\sharp \setminus \Gamma}]}{} \tagbis
	\end{equation}
\end{proof}

\paragraph{Solution 2: Pathhood irrelevance}
\begin{definition}
	A type $T$ has irrelevant pathhood when the weakening of paths $(W, \ctxpath{\var i}) \Dsez t : T \dsub \gamma$ to bridges $(W, \ctxbrid{\var i}) \Dsez t \psub{\var i/\var i} : T \dsub{\gamma (\var i/\var i)}$ is an injection.
\end{definition}
\begin{lemma}
	If $T$ is pathhood irrelevant, then the operation
	\begin{equation}
		\inference{
			\Gamma \sez t : T
		}{\flat \Gamma \sez t[\kappa] : T[\kappa]}{} \tagbis
	\end{equation}
	is injective.
\end{lemma}
\begin{conjecture}
	The interpretation of any type that can be constructed in ParamDTT, has irrelevant pathhood.
\end{conjecture}
\begin{proof}[Justification]
	Pathhood irrelevance is preserved by all functors used in the model (in particular, by all modalities). The discrete universe of discrete types, has irrelevant pathhood by construction. Of the other types, the $\Weld$ type is the most dangerous one, but as we only allow propositions of the form $\idpr i j$ internally, one cannot use $\Weld$ to identify bridges without identifying the corresponding paths.
\end{proof}
Then the model supports the following rule:
\begin{equation}
		\inference{
			\Gamma \judtm{a, b}{\El~A} \qquad
			\sharp \setminus \Gamma \judtmeq a b {\El~A}
		}{\Gamma \judtmeq a b {\El~A}}{}, \tagbis
\end{equation}
which we can combine with Solution 1 to obtain the original reflection rule.

\subsubsection{Function extensionality}
Using the reflection rule, we can derive function extensionality internally:
\begin{equation}
	\inference{
	\inference{
	\inference{
		\inference{
			\leftflat \Gamma \judtm{f, g}{\El~\prodvar \mu x A.B}
		}{(\sharp \setminus \Gamma), \ctxvar \mu x A \judtm{\apvar \mu f x, \apvar \mu g x}{\El~B}}{1}
		\qquad
		\inference{
		\inference{
			\Gamma \judtm{p}{\El~\prodvar \mu x A.{\apvar \mu f x} \idtp{B} {\apvar \mu g x}}
		}{\sharp \setminus \Gamma \judtm{p}{\El~\prodvar \mu x A.{\apvar \mu f x} \idtp{B} {\apvar \mu g x}}}{2}
		}{(\sharp \setminus \Gamma), \ctxvar \mu x A \judtm{\apvar \mu p x}{{\apvar \mu f x} \idtp{B} {\apvar \mu g x}}}{3}
	}{(\sharp \setminus \Gamma), \ctxvar \mu x A \judtmeq{\apvar \mu f x}{\apvar \mu g x}{\El~B}}{4}
	}{\sharp \setminus \Gamma \judtmeq{f}{g}{\El~\prodvar \mu x A.B}}{5}
	}{\Gamma \judtm{\refl~\parclr{f}}{\El~f \idtp{\prodvar \mu x A.B} g}}{6}
\end{equation}
Here we used (1) weakening and application, (2) weakening of variances, (3) weakening and application, (4) the reflection rule, (5) $\lambda$-abstraction and the $\eta$-rule and (6) reflexivity and conversion.
\begin{remark}
	If we use the more careful reflection rule (\cref{thm:careful-reflection}), then after (4) we end up with $\ctxvar {\sharp \setminus \mu} x A$ instead of $\ctxvar \mu x A$, so that we only obtain function extensionality for functions of a modality of the form $\sharp \setminus \mu$, i.e. only for pointwise and continuous functions.
\end{remark}

\subsubsection{Uniqueness of identity proofs}
The model supports uniqueness of identity proofs:
\begin{equation}
	\inference{
		\Gamma \judtm{e, e'}{a \idtp A b}
	}{\Gamma \judtmeq{e}{e'}{a \idtp A b}}{t=-UIP}.
\end{equation}
To prove this, we need to show
\begin{equation}
	\inference{
		\interp \Gamma \sez \interp e, \interp{e'} : {\ftrtm{\sharp}{\interp a}}\idtp{(\sharp \interp{\El~A})[\iota]}{\ftrtm{\sharp}{\interp b}}
	}{\interp \Gamma \sez \interp e = \interp{e'} : {\ftrtm{\sharp}{\interp a}}\idtp{(\sharp \interp{\El~A})[\iota]}{\ftrtm{\sharp}{\interp b}}}{}
\end{equation}
But of course, for any $\gamma : \DSub{W}{\Gamma}$, we have $\interp e \dsub \gamma = \star = \interp{e'} \dsub \gamma$.

\section{Internal parametricity: glueing and welding}
\subsection{The interval}
We interpret the interval as a type $\Gamma \sez \IX \dtype$ that exists in any context $\Gamma$ and is natural in $\Gamma$, i.e.\ it is a closed type. Hence, $\IX \dsub \gamma$ will not depend on $\gamma$. Instead, for $\gamma : \DSub W \Gamma$, we set $\IX \dsub \gamma = (\PSub W {(\ctxbrid{\var i})})$.
\begin{lemma}
	The type $\Gamma \sez \IX \dtype$ is discrete.
\end{lemma}
\begin{proof}
	Pick a term $(W, \ctxpath{\var j}) \Dsez t : \IX \dsub{\gamma(\facewkn{\var j})}$. Then $t$ is a primitive substitution $t : \PSub{(W, \ctxpath{\var j})}{(\ctxbrid{\var i})}$ which necessarily factors over $(\facewkn{\var j})$.
\end{proof}
The interval can be seen as a type:
\begin{equation}
	\interp{\Gamma \judty \IX} = \sharp \interp \Gamma \sez \IX \dtype,
\end{equation}
or as an element of the universe:
\begin{equation}
	\interp{\Gamma \judtm{\IX}{\uni 0}} = \interp \Gamma \sez \tycodeDD \IX : \uniDD_0,
\end{equation}
and then $\interp{\El\,\IX} = \IX [\sharp \varsigma] = \IX$. The terms $\interp{\Gamma \judtm{0, 1}{\IX}}$ are modelled by $\interp \Gamma \sez 0, 1 : \IX$ where $W \Dsez 0 \dsub \gamma = (0/\var i, \facewkn W) : \IX \dsub \gamma$ and similar for 1. All other rules regarding the interval are straightforwardly interpreted now that we know that $\IX$ is semantically a type like any other.

\subsection{Face predicates and face unifiers}\label{sec:face-predicates}
\subsubsection{The discrete universe of propositions}
If we had an internal face predicate judgement $\Gamma \sez \parclr P ~\name{fpred}$, analogous to the type judgement $\Gamma \judty T$, then the most obvious interpretation would be $\sharp \interp \Gamma \sez \interp P \prop$. However, in order to satisfy $\interp{\coshp \setminus \Delta} = \sharp \interp \Delta$ for contexts $\Delta$ that contain face predicates, we only want to consider face predicates that absorb $\sharp$. One can show that these take the form $\sharp \interp \Gamma \sez \sharp P \prop$, where $\flat \interp \Gamma \sez P \prop$. The latter corresponds to $\flat \interp \Gamma \sez \tycode P : \Prop$, which in turn corresponds to $\sharp \interp \Gamma \sez \iota(\tycode P)[\kappa]\inv : \sharp \Prop$. So whereas $\uniDD_\ell$ was defined as $\flat \coshp \uniPsh_\ell$, we define the discrete universe of propositions $\PropD = \flat \coshp (\sharp \Prop) = \flat \Prop$. We have
\begin{equation}
	\inference{
	\inference{
	\inference{
	\inference{
	\inference{
		\Gamma \sez P : \PropD = \flat \coshp \sharp \Prop
	}{\shp \Gamma \sez \kappa(P[\varsigma]\inv) : \coshp \sharp \Prop}{}
	}{\sharp \shp \Gamma \sez \vartheta(\kappa(P[\varsigma]\inv)[\iota]\inv) : \sharp \Prop}{}
	}{\flat \sharp \shp \Gamma = \shp \Gamma \sez \iota\inv(\vartheta(\kappa(P[\varsigma]\inv)[\iota]\inv)[\kappa]) : \Prop}{}
	}{\shp \Gamma \sez \El~\iota\inv(\vartheta(\kappa(P[\varsigma]\inv)[\iota]\inv)[\kappa]) \prop}{}
	}{\sharp \shp \Gamma \sez \sharp \El~\iota\inv(\vartheta(\kappa(P[\varsigma]\inv)[\iota]\inv)[\kappa]) \prop}{}
\end{equation}
We took a significant detour here in order to emphasize the parallel with $\uniDD_\ell$. We could have more simply done
\begin{equation}
	\inference{
	\inference{
	\inference{
		\Gamma \sez P : \PropD = \flat \Prop
	}{\shp \Gamma \sez \kappa(P[\varsigma]\inv) : \Prop}{}
	}{\shp \Gamma \sez \El~\kappa(P[\varsigma]\inv) \prop}{}
	}{\sharp \shp \Gamma \sez \sharp \El~\kappa(P[\varsigma]\inv) \prop}{}
\end{equation}
We show that these are equal. Making the interesting part of the former term more precise, we get:
\begin{align*}
	\iota\inv((\vartheta \sharp)((\kappa\coshp \sharp)(P[\varsigma \Gamma]\inv)[\iota \shp \Gamma]\inv)[\kappa \sharp \shp \Gamma])
	&= \iota\inv((\vartheta \sharp)((\kappa\coshp \sharp)(P[\varsigma \Gamma]\inv)[\iota \shp \Gamma]\inv)[\iota \shp \Gamma]) \\
	&= \iota\inv((\vartheta \sharp)((\kappa\coshp \sharp)(P[\varsigma \Gamma]\inv))) \\
	&= \iota\inv((\vartheta \sharp \circ \kappa\coshp \sharp)(P[\varsigma \Gamma]\inv)) \\
	&= \iota\inv((\iota \circ \kappa)(P[\varsigma \Gamma]\inv))
	= \kappa(P [\varsigma \Gamma]\inv),
\end{align*}
which is the corresponding part of the latter term. We set $\ElD~P = \El~\kappa(P[\varsigma]\inv)$ and inversely $\tycodeD P = \kappa\inv(\tycode P)[\varsigma]$.

\subsubsection{The face predicate formers}
We interpret $\interp{\Gamma \judty \IF} = (\sharp \interp \Gamma \sez \PropD \dtype)$, giving meaning to the face predicate judgement. The identity predicate is interpreted as:
\begin{equation}
	\interp{
		\inference{\Gamma \judtm{i, j}{\IX}}{\Gamma \judtm{\idpr i j}{\IF}}{f-eq}
	} =
	\inference{
	\inference{
	\inference{
		\interp \Gamma \sez \interp i, \interp j : \IX
	}{\shp \interp \Gamma \sez \interp i [\varsigma]\inv, \interp j [\varsigma]\inv : \IX}{}
	}{\shp \interp \Gamma \sez \interp i [\varsigma]\inv \idtp{\IX} \interp j [\varsigma]\inv \prop}{}
	}{\interp \Gamma \sez \tycodeD{\interp i [\varsigma]\inv \idtp{\IX} \interp j [\varsigma]\inv} : \PropD}{}.
\end{equation}
Other connectives are interpreted simply by decoding and encoding, e.g.
\begin{equation}
	\interp{
		\inference{\Gamma \judtm{P, Q}{\IF}}{\Gamma \judtm{P \wedge Q}{\IF}}{f-$\wedge$}
	} =
	\inference{
	\inference{
	\inference{
		\interp \Gamma \sez \interp P, \interp Q : \PropD
	}{\shp \interp \Gamma \sez \ElD \interp P, \ElD \interp Q \prop}{}
	}{\shp \interp \Gamma \sez \ElD \interp P \wedge \ElD \interp Q \prop}{}
	}{\interp \Gamma \sez \tycodeD{\ElD \interp P \wedge \ElD \interp Q} : \PropD}{}
\end{equation}

\subsubsection{Context extension with a face predicate}
\begin{equation}
	\interp{
		\inference{
			\Gamma \ctx \qquad
			\leftflat \Gamma \judtm P \IF
		}{\Gamma, \ctxface P \ctx}{c-f}
	} =
	\inference{
		\interp \Gamma \ctx \qquad
		\inference{
		\inference{
		\inference{
		\inference{
			\interp{\sharp \setminus \Gamma} \sez \interp P : \PropD
		}{\shp \interp{\sharp \setminus \Gamma} \sez \ElD \interp P \prop}{}
		}{\sharp \shp \interp{\sharp \setminus \Gamma} \sez \sharp \ElD \interp P \prop}{}
		}{\sharp \interp \Gamma \sez (\sharp \ElD \interp P)[\sharp \varsigma] \prop}{}
		}{\interp \Gamma \sez (\sharp \ElD \interp P)[\sharp \varsigma][\iota] \prop}{}
	}{\interp \Gamma, \_ : (\sharp \ElD \interp P)[\sharp \varsigma][\iota] \ctx}{}
\end{equation}
Note that we have
\begin{equation}
	(\sharp \ElD \interp P)[\sharp \varsigma]
	= \sharp ((\ElD \interp P)[\varsigma])
	= \sharp \El~\kappa(\interp P).
\end{equation}
By analogy to types, we will denote this as $\interp{\El~P}$, even though $\El~P$ does not occur in the syntax. We can then write extended contexts as $\interp \Gamma, \_ : \interp{\El~P}[\iota]$.

\begin{proof}[Addendum to the proof of \cref{thm:leftflat}]\footnote{Broken hyperlink.}
	The fact that $\flat \interp{\Gamma, \ctxface P} = \flat \interp{\sharp \setminus (\Gamma, \ctxface P)}$ follows trivially from $\flat \interp \Gamma = \flat \interp{\sharp \setminus \Gamma}$. Since propositions are discrete, extending the context with a proposition preserves its discreteness and hence the fact that $\kappa$ for that context is an isomorphism.
\end{proof}
\begin{proof}[Addendum to the proof of \cref{thm:leftsharp}]\footnote{Broken hyperlink.}
	The fact that $\interp{\coshp \setminus(\Gamma, \ctxface P)} = \sharp\interp{\Gamma, \ctxface P}$ follows from
	\begin{equation}
		\sharp(\interp{\El~P}[\iota]) = \sharp \interp{\El~P} = \interp{\El~P} = \interp{\El~P}[\iota\sharp]. \qedhere
	\end{equation}
\end{proof}

\subsubsection{Face unifiers} The use of face unifiers is motivated from a computational perspective and is a bit unpractical semantically. Since $\Prop$ is the subobject quantifier of $\widehat \bpcubecat$, extending a context with a proposition amounts to taking a subobject of the context. Every syntactic face unifier $\sigma : \Delta \to \Gamma$ has an interpretation $\interp \sigma : \interp \Delta \to \interp \Gamma$. One can show that the union of the images of all interpretations of all face unifiers to a context $\Gamma$, is equal to all of $\interp \Gamma$. Hence, checking whether something works under all face unifiers, amounts to checking whether it works. One can also show that $P \Rightarrow Q$ means $\interp P \subseteq \interp Q$. Then the rule
\begin{equation}
	\inference{
		\Gamma \judtm{P, Q}{\IF} \qquad
		P \Leftrightarrow Q
	}{\Gamma \judtmeq P Q \IF}{f=}
\end{equation}
is trivial. We also have
\begin{align*}
	& \interp{
		\inference{
			\Gamma \judtm{i, j}{\IX} \qquad
			\top \Rightarrow (\idpr i j)
		}{\Gamma \judtmeq i j \IX}{i=-f}
	} = 
	\inference{
	\inference{
	\inference{
	\inference{
		\interp \Gamma \sez \top \subseteq \tycodeD{(\interp i [\varsigma]\inv \idtp{\IX} \interp j [\varsigma]\inv)} : \PropD
	}{\shp \interp \Gamma \sez \top \subseteq (\interp i [\varsigma]\inv \idtp{\IX} \interp j [\varsigma]\inv) \prop}{}
	}{\shp \interp \Gamma \sez \star : \interp i [\varsigma]\inv \idtp{\IX} \interp j [\varsigma]\inv}{}
	}{\shp \interp \Gamma \sez \interp i [\varsigma]\inv = \interp j [\varsigma]\inv : \IX}{}
	}{\interp \Gamma \sez \interp i = \interp j : \IX}{}
\end{align*}

\subsection{Systems} The interpretation of systems is straightforward.

\subsection{Welding}
We interpret
\begin{equation}
	\inference{
		\Gamma \judtm{P}{\IF} \qquad
		\Gamma, \ctxface P \judtm{T}{\El~\uni \ell} \qquad
		\Gamma \judtm{A}{\El~\uni \ell} \\
		\leftsharp \Gamma, \ctxface P \judtm{f}{\El~A \to T}
	}{\Gamma \judtm{\Weldsys{A}{\Weldsysclause{P}{T}{f}}}{\El~\uni \ell}}{t-Weld}
\end{equation}
as
\begin{equation*}
	\inference{
	\inference{
		\left\{
		\begin{matrix}
			\inference{
				\interp \Gamma \sez \interp P : \PropD
			}{\sharp \shp \interp \Gamma \sez \sharp \ElD \interp P \prop}{}
			&
			\inference{
			\inference{
			\inference{
				\interp \Gamma, \var p : (\sharp \ElD \interp P)[\sharp \varsigma][\iota] \sez \interp T : \uniDD_\ell 
			}{\sharp \shp \paren{\interp \Gamma, \var p : (\sharp \ElD \interp P)[\sharp \varsigma][\iota]} \sez \ElDD \interp T : \dtype_\ell}{}
			}{\sharp \paren{\shp \interp \Gamma, \var p : (\sharp \ElD \interp P)[\iota]} \sez \ElDD \interp T [\sharp (\subext \varsigma)] \dtype_\ell}{}
			}{\sharp \shp \interp \Gamma, \var p : \sharp \ElD \interp P \sez \ElDD \interp T [\sharp (\subext \varsigma)] \dtype_\ell}{} \\
			\qquad \\
			\inference{
				\interp \Gamma \sez \interp A : \uniDD_\ell
			}{\sharp \shp \interp \Gamma \sez \ElDD \interp A \dtype_\ell}{}
			&
			\inference{
				\sharp \interp \Gamma, \var p : \sharp \ElD \interp P [\sharp \varsigma] \sez \interp f : \ElDD \interp A [\sharp \varsigma][\wknvar p] \to \ElDD \interp T [\sharp \varsigma]
			}{\sharp \shp \interp \Gamma, \var p : \sharp \ElD \interp P \sez \interp f [\varsigma \subext]\inv : \ElDD \interp A [\wknvar p] \to \ElDD \interp T[\sharp(\subext \varsigma)]}{}
		\end{matrix}
		\right.
	}{\sharp \shp \interp \Gamma \sez \Weldsys{\ElDD \interp A}{\Weldsysclauseb{\sharp \ElD \interp P}{\ElDD \interp T [\sharp (\subext \varsigma)]}{\interp f [\varsigma \subext]\inv}} \dtype_\ell}{}
	}{\interp \Gamma \sez \tycodeDD{\Weldsys{\ElDD \interp A}{\Weldsysclauseb{\sharp \ElD \interp P}{\ElDD \interp T [\sharp (\subext \varsigma)]}{\interp f [\varsigma \subext]\inv}}} : \uniDD_\ell}{}.
\end{equation*}
We have
\begin{align*}
	\interp{\El~\Weldsys{A}{\Weldsysclause{P}{T}{f}}}
	&= \ElDD \interp{\Weldsys{A}{\Weldsysclause{P}{T}{f}}} [\sharp \varsigma] \\
	&= \Weldsys{\interp{\El~A}}{\Weldsysclauseb{\interp{\El~P}}{\interp{\El~T}}{\interp f}}.
\end{align*}

The constructor
\begin{equation}
	\inference{
		\Gamma \judty{\El~\Weldsys{A}{\Weldsysclause{P}{T}{f}}} \qquad
		\Gamma \judtm{a}{\El~A}
	}{\Gamma \judtm{\weldsys{\sysclause P f} a}{\El~\Weldsys{A}{\Weldsysclause{P}{T}{f}}}}{t-weld}
\end{equation}
becomes
\begin{equation}
	\inference{
		\interp \Gamma \sez \interp a : \interp{\El~A} [\iota]
	}{\interp \Gamma \sez \weldsys{\sysclauseb{\interp{\El~P} [\iota]}{\interp f} [\iota]} \interp a : \Weldsys{\interp{\El~A}}{\Weldsysclauseb{\interp{\El~P}}{\interp{\El~T}}{\interp f}}[\iota]}{}.
\end{equation}

For the eliminator
\begin{equation}
	\inference{
		\Gamma, \ctxvar \nu {y}{\El~\Weldtp P A T f} \judty{\El~C} \qquad
		\Gamma, \ctxface P, \ctxvar \nu y {\El~T} \judtm{d}{\El~C} \\
		\Gamma, \ctxvar \nu x {\El~A} \judtm{c}{\El~C[\varclr{\weldtm P f x} / \varclr y]} \qquad
		\Gamma, \ctxface P, \ctxvar \nu x {\El~A} \judtmeq{c}{d[\varclr{f\,x} / \varclr y]}{\El~C[\varclr{f\,x} / \varclr y]} \\
		\nu \setminus \Gamma \judtm{b}{\El~\Weldtp P A T f}
	}{\Gamma \judtm{
		\ind^\nu_\Weld(\varclr y.\parclr C, \sys{\sysclause{P}{\varclr y.d}}, \varclr x.c, \varclr b)	
	}{\El~C[\varclr b/\varclr y]}}{t-indweld}
\end{equation}
first note that
\begin{align*}
	\nu \interp{\El~\Weldsys{A}{\Weldsysclause{P}{T}{f}}}
	&= \nu \Weldsys{\interp{\El~A}}{\Weldsysclauseb{\interp{\El~P}}{\interp{\El~T}}{\interp f}} \\
	&= \Weldsys{\nu \interp{\El~A}}{\Weldsysclauseb{\interp{\El~P}}{\nu \interp{\El~T}}{\lambda \ftrvar \nu x.\ftrtm \nu {(\interp f \var x)}}}
\end{align*}
because lifted functors preserve $\Weld$ and $\nu \interp{\El~P} = \interp{\El~P}$ for $\nu \in \accol{\coshp, \idmod, \sharp}$. Taking that into account, all of this boils down to straightforward use of the eliminator of the $\Weld$-type for presheaves.

\subsection{Glueing}
We similarly get
\begin{equation}
	\interp{\El~\Gluesys{A}{\Gluesysclause{P}{T}{f}}}
	= \Gluesys{\interp{\El~A}}{\Gluesysclauseb{\interp{\El~P}}{\interp{\El~T}}{\interp f}}.
\end{equation}

The constructor
\begin{equation}
	\inference{
		\Gamma \judty{\El~\Gluetp P A T f} \qquad
		\Gamma, \ctxface P \judtm t {\El~T} \qquad
		\Gamma \judtm a {\El~A} \qquad
		\Gamma, \ctxface P \judtmeq{f\,t}{a}{\El~A}
	}{\Gamma \judtm{\gluetm P a t}{\El~\Gluetp P A T f}}{t-glue}
\end{equation}
becomes
\begin{equation}
	\inference{
		\interp \Gamma, \var p : \interp{\El~P}[\iota] \sez \interp t : \interp{\El~T} \\
		\interp \Gamma \sez \interp a : \interp{\El~A} \\
		\interp \Gamma, \var p : \interp{\El~P}[\iota] \sez \interp f [\iota] \interp t = \interp a [\wknvar p] : \interp{\El~A}[\iota]
	}{\interp \Gamma \sez \gluesys{\interp a}{\sysclauseb{\interp{\El~P}[\iota]}{\interp t}} : \Gluesys{\interp{\El~A}}{\Gluesysclauseb{\interp{\El~P}}{\interp{\El~T}}{\interp f}}[\iota]}{}.
\end{equation}

The eliminator
\begin{equation}
	\inference{
		\Gamma \judtm{b}{\El~\Gluetp P A T f}
	}{\Gamma \judtm{\ungluetm P f b}{\El~A}}{t-unglue}
\end{equation}
becomes
\begin{equation}
	\inference{
		\interp \Gamma \sez \interp b : \Gluesys{\interp{\El~A}}{\Gluesysclauseb{\interp{\El~P}}{\interp{\El~T}}{\interp f}}[\iota]
	}{\interp \Gamma \sez \ungluesys{\sysclauseb{\interp{\El~P} [\iota]}{\interp f [\iota]}} \interp b : \interp{\El~A}[\iota]}{}.
\end{equation}

\subsection{The path degeneracy axiom}
We will interpret the path degeneracy axiom
\begin{equation}
	\inference{
		\Gamma \judty A \qquad
		\leftflat{\Gamma} \judtm{p}{\Pipar(i : \IX).A}
	}{\Gamma \judtm{\degaxof p}{{p}\idtp{\Pipar(i : \IX).A}{\paren{\lamannotpar{i}{\IX}\appar p 0}}}}{t-degax}
\end{equation}
via the stronger rule
\begin{equation}
	\inference{
		\Gamma \judtm{p}{\prodpar i \IX.A}
	}{\Gamma \judtmeq{p}{\lamannotpar i \IX.\appar p 0}{\prodpar i \IX.A}}{}
\end{equation}
after which the axiom follows by reflexivity. For simplicity, we put the variables in the context. We have
\begin{equation}
	\inference{
	\inference{
		\interp \Gamma, \ftrvar \sharp i : \sharp \IX \sez \interp a : \ElDD \interp A [\sharp \varsigma][\iota][\wkn{\ftrvar \sharp i}]
	}{\shp \paren{\interp \Gamma, \ftrvar \sharp i : \sharp \IX} \sez \interp a [\varsigma]\inv : \ElDD \interp A [\iota][\shp \wkn{\ftrvar \sharp i}]}{}
	}{\shp \interp \Gamma \sez \interp a [\varsigma]\inv [\shp \wkn{\ftrvar \sharp i}]\inv : \ElDD \interp A [\iota]}{}
\end{equation}
because $\shp \wkn{\ftrvar \sharp i} : \shp \paren{\interp \Gamma, \ftrvar \sharp i : \sharp \IX} \cong \shp \interp \Gamma$ can be shown to be an isomorphism. Since these operations are invertible, we have
\begin{equation}
	\interp a = \interp a [\varsigma]\inv [\shp \wkn{\ftrvar \sharp i}]\inv [\varsigma] [\wkn{\ftrvar \sharp i}]
\end{equation}
and the right hand side is clearly invariant under $\loch[\id, \ftrtm \sharp 0/\ftrvar \sharp i]$.

To see in general that $\shp \wkn{\ftrvar \sharp i} : \shp(\Gamma, \ftrvar \sharp i : \IX) \cong \shp \Gamma$ is an isomorphism, pick $\overline{(\gamma, \fpshadj \flat(\vfi))} : \DSub{W}{\shp(\Gamma, \ftrvar \sharp i : \IX)}$. We show that $\overline{(\gamma, \fpshadj \flat(\vfi))} = \overline{(\gamma, \ftrtm \sharp 0)}$. We have $\fpshadj \flat(\vfi) : \PSub{W}{\sharp (\ctxbrid{\var i})}$ (there is some benevolent abuse of notation involved here, related to the fact that $\IX$ is a closed type) and hence $\vfi : \PSub{\flat W}{(\ctxbrid{\var i})}$. Now $\vfi$ factors as $\vfi = (\facewkn{\flat W})(k / \var i)$. Similarly, one can show that $\ftrtm \sharp 0 = \fpshadj \flat((\facewkn{\flat W})(0/\var i))$.

By discreteness of $\shp$, we have
\begin{equation}
	\overline{(\gamma (\facewkn{\var i}), \fpshadj \flat(\facewkn{\flat W}))}
	= \overline{(\gamma (\facewkn{\var i}), \fpshadj \flat(\facewkn{\flat W}))}(0/\var i, \facewkn{\var i})
	= \overline{(\gamma (\facewkn{\var i}), \ftrtm \sharp 0(\facewkn{\var i}))}
	= \overline{(\gamma, \ftrtm \sharp 0)} (\facewkn{\var i}).
\end{equation}
Restricting both sides by $(k/\var i)$, we find what we wanted to prove.

\section{Sizes and natural numbers}
\subsection{The natural numbers}
In any presheaf category, we can define a closed type $\Nat$ by setting $\Nat \dsub \gamma = \IN$ and $n \psub \vfi = n$. We have
\begin{equation*}
	\inference{\Gamma \ctx}{\Gamma \sez \Nat \type_0}{}, \qquad
	\inference{\Gamma \ctx}{\Gamma \sez 0 : \Nat}{}, \qquad
	\inference{\Gamma \sez n : \Nat}{\Gamma \sez \suc\,n : \Nat}{} \qquad
	\inference{
		\Gamma, \var m : \Nat \sez C \type \\
		\Gamma \sez c_0 : C[\id, 0/\var m] \\
		\Gamma, \var m : \Nat, \var c : C \sez c_\suc : C[\id,\suc\,\var m/\var m] \\
		\Gamma \sez n : \Nat
	}{\Gamma \sez \ind_\Nat(\var m.C, c_0, \var m.\var c.c_\suc, n) : C[\id, n/\var m]}{}
\end{equation*}
and all these operators are natural in $\Gamma$ \emph{and} are respected by lifted functors, e.g.\ $\fpsh F \Nat = \Nat$ and $\ftrtm{\fpsh F}{(\suc\,n)} = \suc(\ftrtm{\fpsh F}{n})$.

In $\widehat{\bpcubecat}$, $\Nat$ is a discrete type since $n \psub{0/\var i, \facewkn{\var i}} = n$ since in general $n \psub \vfi = n$. We can now interpret the inference rules for natural numbers:
\begin{equation}
	\interp{
		\inference{
			\Gamma \ctx
		}{\Gamma \judtm{\Nat}{\uni 0}}{t-Nat}
	} =
	\inference{
	\inference{
		\interp \Gamma \ctx
	}{\sharp \shp \interp \Gamma \sez \Nat \dtype_0}{}
	}{\interp \Gamma \sez \tycodeDD \Nat : \uniDD_0}{}.
\end{equation}
We have $\interp{\El~\Nat} = \Nat$.
\begin{equation}
	\interp{
		\inference{
			\Gamma \ctx
		}{\Gamma \judtm{0}{\El~\Nat}}{t-0}
	} =
	\inference{
		\interp \Gamma \ctx
	}{\interp \Gamma \sez 0 : \Nat}{},
	\qquad \qquad
	\interp{
		\inference{
			\Gamma \judtm{n}{\El~\Nat}
		}{\Gamma \judtm{\suc\,n}{\El~\Nat}}{t-s}
	} =
	\inference{
		\interp \Gamma \sez \interp n : \Nat
	}{\interp \Gamma \sez \suc \interp n : \Nat}{}.
\end{equation}
The induction principle
\begin{equation}
	\inference{
		\Gamma, \ctxvar \nu m {\El~\Nat} \judty {\El~C} \qquad
		\Gamma \judtm{c_0}{\El~C[\varclr 0/\varclr m]} \\
		\Gamma, \ctxvar \nu m {\El~\Nat}, \ctxctu c {\El~C} \judtm{c_\suc}{\El~C[\varclr{\suc\,m}/\varclr m]} \\
		\nu \setminus \Gamma \judtm n {\El~\Nat}
	}{
		\Gamma \judtm{\ind^\nu_\Nat(\varclr m.\parclr C, c_0, \varclr m.c.c_\suc, \varclr n)
	}{\El~C[\varclr n / \varclr m]}}{t-indnat}
\end{equation}
is interpreted as
\begin{equation}
	\inference{
		\interp \Gamma, \var m : \Nat \sez \interp{\El~C}[\iota] \dtype \qquad
		\interp \Gamma \sez \interp{c_0} : \interp{\El~C}[\iota][\id, 0/\var m] \\
		\interp \Gamma, \var m : \Nat, \var c : \interp{\El~C}[\iota] \sez \interp{c_\suc} : \interp{\El~C}[\iota][\suc\,\var m/\var m] \\
		\interp \Gamma \sez \forsub \nu \interp n : \Nat
	}{\interp \Gamma \sez \ind_\Nat(\var m.\interp{\El~C}[\iota], \interp{c_0}, \var m.\var c.\interp{c_\suc}, \forsub \nu \interp n)}{},
\end{equation}
where we use extensively that $\nu \Nat = \Nat$.

\subsection{Sizes}$~$

\medskip

\subsubsection{In the model}

\begin{proposition}
	We have a discrete type $\Size$ with the following inference rules:
	\begin{equation}
		\inference{\Gamma \ctx}{\Gamma \sez \Size \dtype_0}{} \qquad
		\inference{\Gamma \ctx}{\Gamma \sez \szero : \Size}{} \qquad
		\inference{\Gamma \sez n : \Size}{\Gamma \sez \ssuc n : \Size}{} \qquad
		\inference{
			\shp \Gamma \sez P \prop \qquad
			\Gamma, \var p : P[\varsigma] \sez n : \Size
		}{\Gamma \sez \sfillsys{\sfillsysclause{\var p : P}{n}} : \Size}{}
	\end{equation}
	\begin{equation}
		\inference{
			\shp \Gamma \sez i : \IX \qquad
			\Gamma \sez m, n : \Size \\
			\Gamma, \var p : \paren{(i \idtp \IX 0) \vee (i \idtp \IX 1)}[\varsigma] \sez m[\wknvar p] = n[\wknvar p] : \Size
		}{\Gamma \sez m = n : \Size}{} \quad
		\inference{
			\Gamma \sez m, n : \Size
		}{\Gamma \sez m \sqcup n : \Size}{}
	\end{equation}
	satisfying the expected equations.
\end{proposition}
For closed types $T$, the set $T \dsub \gamma$ is independent of $\gamma : \PSub W \Gamma$; hence we will denote it as $\DSub W T$. Similarly, we will write $W \Dsez t : T$ for $W \Dsez t : T \dsub \gamma$.
\begin{proof}
	For $\gamma : \DSub W \Gamma$, we set $(\DSub W \Size) = \IN^{\PSub{()}{\shp W}}$, i.e.\ a term $n : \Size \dsub \gamma$ consists of a natural number for every vertex of the cube $\shp W$, which is $W$ with all path dimensions contracted. Put differently still, a term $n : \DSub W \Size$ consists of a natural number for every vertex of the cube $W$, such that numbers for path-adjacent vertices are equal. Writing $n \ssub \psi$ for the vertex corresponding to $\psi : \PSub{()}{\shp W}$, we define accordingly $n \psub \vfi \ssub \psi = n \ssub{\shp \vfi \circ \psi}$. This implies that we have in general $n \ssub \psi = n \psub{\psi'} \ssub{\id_{()}}$ for any $\psi' : \PSub{()}{W}$ such that $\shp \psi' = \psi$. Such a $\psi'$ always exists, e.g. $\psi' = (\psi, 0/\var i^\IP \in W)$. Hence, we will avoid the $\ssub \loch$ notation and say that a term $W \Dsez n : \Size$ is determined by all of its vertices $() \Dsez n \psub \vfi : \Size$ which are in fact functions $\IN^{\PSub{()}{()}}$ but can be treated as naturals since $\PSub{()}{()}$ is a singleton. Every such term $n$ has the property that $() \Dsez n \psub \vfi : \Size$ is independent of how $\vfi$ treats path variables, i.e. if $\shp \vfi = \shp \psi$, then $n \psub \vfi = n \psub \psi$.
	
	To see that $\Size$ is discrete, pick $(W, \ctxpath{\var i}) \Dsez n : \Size$. Then $n \psub{0/\var i, \facewkn{\var i}} = n$ because $\shp(0/\var i, \facewkn{\var i}) = \id$.
	
	We define $W \Dsez \szero : \Size$ by setting $() \Dsez \szero \psub{\vfi} = 0 : \Size$ for all $\vfi : \PSub{()}{W}$.
	
	We define $W \Dsez \ssuc n : \Size$ by setting $() \Dsez (\ssuc n) \psub \vfi = n \psub \vfi + 1 : \Size$.
	
	Assume we have $\shp \Gamma \sez P \prop$ and $\Gamma, \var p : P[\varsigma] \sez n : \Size$. Then we define $\Gamma \sez \sfillsys{\sfillsysclause{\var p : P}{n}}$ as follows: pick $\gamma : \DSub{()}{\Gamma}$. Then we set $\sfillsys{\sfillsysclause{\var p : P}{n}} \dsub \gamma = n \dsub{\gamma, \star / \var p}$ if $P [\varsigma] \dsub \gamma = \accol{\star}$, and $\sfillsys{\sfillsysclause{\var p : P}{n}} \dsub \gamma = 0$ if $P [\varsigma] \dsub \gamma = \eset$. We need to show that this respects paths, i.e. if $\vfi, \psi : \PSub{()}{W}$, $\gamma : \DSub{W}{\Gamma}$ and $\shp \vfi = \shp \psi$, then we should prove that $\sfillsys{\sfillsysclause{\var p : P}{n}} \dsub{\gamma \vfi} = \sfillsys{\sfillsysclause{\var p : P}{n}} \dsub{\gamma \psi}$. First, we show that $\varsigma \gamma \vfi = \varsigma \gamma \psi$. Since $\varsigma$ decomposes as $\kappa\inv \inquotshp$, it suffices to show that $\kappa \varsigma \gamma \vfi = \kappa \varsigma \gamma \psi$. Now we have
	\begin{equation}
		\kappa \quotshp \Gamma \circ \varsigma \Gamma \circ \gamma \circ \vfi
		= \fpshadj \shp\inv(\varsigma \Gamma \circ \gamma \circ \vfi) \circ \varsigma()
		= \fpshadj \shp\inv(\varsigma \Gamma \circ \gamma) \circ \shp \vfi \circ \varsigma(),
	\end{equation}
	and similar for $\psi$, which proves the equality since $\shp \vfi = \shp \psi$. But this implies that $P[\varsigma]\dsub{\gamma \vfi} = P[\varsigma]\dsub{\gamma \psi}$. Since we also have $n \dsub{\gamma \vfi} = n \dsub{\gamma \psi}$, we can conclude that $\sfill$ respects paths.
	
	For the equality expressing codiscreteness, assume the premises and pick $\gamma : \DSub{()}{\Gamma}$. Then we have $i\dsub{\varsigma \gamma} : \DSub{()}{\IX}$, which is either 0 or 1. Hence, $\paren{(i \idtp \IX 0) \vee (i \idtp \IX 1)} [\varsigma] \dsub \gamma = \accol \star$, and we find
	\begin{equation}
		m \dsub \gamma = m[\wknvar p] \dsub{\gamma, \star} = n[\wknvar p] \dsub{\gamma, \star} = n \dsub \gamma.
	\end{equation}

	We define $W \Dsez m \sqcup n : \Size$ by setting $() \Dsez (m \sqcup n) \psub \vfi : \Size$ equal to the maximum of $m \psub \vfi$ and $n \psub \vfi$.
\end{proof}

We can easily lift $\szero$, $\ssuc$ and $\smax{}{}$ to $\sharp \Size$, e.g. we have $\var m : \Size \sez \ssuc \var m : \Size$, whence $\ftrvar \sharp m : \sharp \Size \sez \ftrtm{\sharp}{(\ssuc \var m)} : \sharp \Size$, and then using substitution we can derive
\begin{equation}
	\inference{
		\Gamma \sez n : \sharp \Size
	}{\Gamma \sez \ftrtm{\sharp}{(\ssuc \var m)}[\bullet, n/\ftrvar \sharp m] : \sharp \Size}{}.
\end{equation}
We will denote the latter term as $\ssuc n$. Similarly, we can lift $\sfill$:
\begin{equation}
	\inference{
	\inference{
		\inference{
		\inference{
			\Gamma \sez P \prop
		}{\flat \Gamma \sez \flat P \prop}{}
		}{\shp \flat \Gamma \sez (\flat P)[\varsigma \inv] \prop}{}
		\qquad
		\inference{
			\Gamma, \var p : P \sez n : \sharp \Size
		}{\flat \Gamma, \ftrvar \flat p : \flat P \sez \iota\inv(n[\kappa]) : \Size}{}
	}{\flat \Gamma \sez \sfillsys{\sfillsysclause{\ftrvar \flat p : (\flat P)[\varsigma\inv]}{\iota\inv(n[\kappa])}} : \Size}{}
	}{\Gamma \sez \iota \paren{\sfillsys{\sfillsysclause{\ftrvar \flat p : (\flat P)[\varsigma\inv]}{\iota\inv(n[\kappa])}}}[\kappa]\inv : \sharp \Size}{}
\end{equation}
We will denote this result as $\sfillsyssharp{\sfillsysclause{\var p : P}{n}}$. Note that $\varsigma()$ and $\kappa()$ are the identity. For $\gamma : \DSub{()}{\Gamma}$, we have
\begin{align*}
	\sfillsyssharp{\sfillsysclause{\var p : P}{n}} \dsub{\gamma}
	&= \iota \paren{\sfillsys{\sfillsysclause{\ftrvar \flat p : (\flat P)[\varsigma\inv\Gamma]}{\iota\inv(n[\kappa\Gamma])}}}[\kappa\Gamma]\inv \dsub{\gamma} \\
	&= \iota \paren{\sfillsys{\sfillsysclause{\ftrvar \flat p : (\flat P)[\varsigma\inv\Gamma]}{\iota\inv(n[\kappa\Gamma])}}} \dsub{\fpshadj \shp(\gamma \circ \varsigma ()\inv)} \\
	&= \iota \paren{\sfillsys{\sfillsysclause{\ftrvar \flat p : (\flat P)[\varsigma\inv\Gamma]}{\iota\inv(n[\kappa\Gamma])}}} \dsub{\fpshadj \shp(\gamma)} \\
	&= \fpshadj \flat \paren{\sfillsys{\sfillsysclause{\ftrvar \flat p : (\flat P)[\varsigma\inv\Gamma]}{\iota\inv(n[\kappa\Gamma])}} \dsub{\fpshadj \shp(\gamma) \circ \kappa()}} \\
	&= \fpshadj \flat \paren{\sfillsys{\sfillsysclause{\ftrvar \flat p : (\flat P)[\varsigma\inv\Gamma]}{\iota\inv(n[\kappa\Gamma])}} \dsub{\fpshadj \shp(\gamma)}}.
\end{align*}
Now
\begin{equation}
	(\flat P)[\varsigma\inv \Gamma][\varsigma \Gamma] \dsub{\fpshadj \shp(\gamma)}
	= (\flat P) \dsub{\fpshadj \shp \Gamma} = P \dsub \gamma.
\end{equation}
So we make the expected case distinction: if $P \dsub \gamma = \accol \star$, then
\begin{align*}
	\sfillsyssharp{\sfillsysclause{\var p : P}{n}} \dsub{\gamma}
	&= \fpshadj \flat \paren{\iota\inv(n[\kappa\Gamma]) \dsub{\fpshadj \shp(\gamma)}}
	= n[\kappa\Gamma]\dsub{\fpshadj \shp(\gamma)} = n \dsub \gamma.
\end{align*}
If $P \dsub \gamma = \eset$, then we just get $0$. So while it looks ugly, this is precisely the construction we would expect.
\begin{proposition}
	We have an inequality proposition
	\begin{equation}
		\inference{\Gamma \sez m, n : \sharp \Size}{\Gamma \sez m \leq n \prop}{}
	\end{equation}
	that satisfies reflexivity, transitivity, $0 \leq n$, $\ssuc m \leq \ssuc n$ if $m \leq n$, $\sfillsyssharp{\sfillsysclause{\var p : P}{m}} \leq \sfillsyssharp{\sfillsysclause{\var p : P}{n}}$ if $m \leq n$, and $m \leq m \sqcup n$ and $n \leq m \sqcup n$.
\end{proposition}
\begin{proof}
	Assume we have $W \Dsez \fpshadj \flat(m), \fpshadj \flat(n) : \sharp \Size$, i.e. $\flat W \Dsez m, n : \Size$. Then we set $(\fpshadj \flat(m) \leq \fpshadj \flat(n))$ equal to $\accol \star$ if $() \Dsez m \psub \vfi \leq n \psub \vfi : \Size$ for every $\vfi : \PSub{()}{\flat W}$. Otherwise, we set it equal to $\eset$. This can be shown to satisfy the required properties.
\end{proof}
\begin{proposition}
	We have
	\begin{equation}
		\inference{
			(\mu, \beta : \mu \to \sharp) \in \accol{(\coshp, \iota \vartheta), (\Id, \iota), (\sharp, \id)} \\
			\Gamma, \ftrvar \mu n : \mu \Size \sez A \type \\
			\Gamma \sez f : \Pi(\ftrvar \mu n : \mu \Size).\paren{  \Pi(\ftrvar \mu m : \mu \Size).(\uparrow \beta(\ftrvar \mu m) \leq \beta(\ftrvar \mu n)) \to A[\ftrvar \mu m/\ftrvar \mu n]  } \to A
		}{\Gamma \sez \fix^\mu\,f : \Pi(\ftrvar \mu n : \mu\Size).A}{}
	\end{equation}
	(where we omit weakening and other uninteresting parts of substitutions).
\end{proposition}
\begin{proof}
	We define $\fix^\mu\,f = \lambda \ftrvar \mu n.f~\ftrvar \mu n~(\lambda \ftrvar \mu m.\lambda \var p.\fix^\mu~f~\ftrvar \mu m)$ (where we omit weakening substitutions), which we will prove to be a well-founded definition by induction essentially on the greatest vertex of $\ftrvar \mu n$.
	\begin{description}
		\item[$\mu = \Id$] Pick $\gamma : \DSub W \Gamma$ and $n : \DSub W \Size$. Let $\omega : \PSub{()}{W}$ attain a maximal vertex $n \psub \omega$ of $n$. We show that we can define $(\fix~f)\dsub \gamma \cdot n$ as $f \dsub \gamma \cdot n \cdot (\lambda \var m.\lambda \var p.\fix~f~\var m) \dsub \gamma$, assuming that $(\fix~f) \dsub \gamma \cdot m$ is already defined for all $m : \DSub W \Size$ such that all vertices of $m$ are less than $n \psub \omega$. In other words, we have to show that $(\lambda \var m.\lambda \var p.\fix~f~\var m) \dsub \gamma$ is already defined. But this function is determined completely by defining terms of the form $(\lambda \var m.\lambda \var p.\fix~f~\var m) \dsub \gamma \cdot m \cdot p$, where $W \Dsez m : \Size$ and $W \Dsez p : \ssuc \iota(m) \leq \iota(n)$. Note that $\iota(m) = \fpshadj \flat(m \psub{\kappa})$. The existence of $p$ then implies that $m \psub{\kappa \vfi} + 1 \leq n \psub{\kappa \vfi}$ for all $\vfi : \PSub{()}{\flat W}$. Since any $\psi : \PSub{()}{W}$ factors as $\psi = \psi \circ \kappa () = \kappa W \circ \flat \psi$, we can conlude that all vertices of $m$ are less than the corresponding ones of $n$, hence less than $n \psub \omega$. Then $(\fix~f) \dsub \gamma \cdot m$ is already defined, and one can show that
		\begin{equation}
			(\lambda \var m.\lambda \var p.\fix~f~\var m) \dsub \gamma \cdot m \cdot p = (\fix~f) \dsub \gamma \cdot m.
		\end{equation}
		Hence, the definition is well-founded.
		
		\item[$\mu = \sharp$] Pick $\gamma : \DSub W \Gamma$ and $\fpshadj \flat(n) : \DSub{W}{\sharp \Size}$, i.e. $n : \DSub{\flat W}{\Size}$. Let $\omega : \PSub{()}{\flat W}$ attain a maximal vertex $n \psub \omega$ of $n$. We show that we can define $(\fix^\sharp~f)\dsub \gamma \cdot \fpshadj \flat(n)$ as $f \dsub \gamma \cdot \fpshadj \flat(n) \cdot (\lambda \ftrvar \sharp m.\lambda \var p.\fix^\sharp~f~\ftrvar \sharp m) \dsub \gamma$. So all $(\lambda \ftrvar \sharp m.\lambda \var p.\fix^\sharp~f~\ftrvar \sharp m) \dsub \gamma \cdot m \cdot p = (\fix^\sharp~f) \dsub \gamma \cdot m$ have to be defined. But the existence of $W \Dsez p : \fpshadj \flat(m+1) \leq \fpshadj \flat(n)$ asserts that $(\fix^\sharp~f) \dsub \gamma \cdot m$ is already defined by the induction hypothesis. 
		
		\item[$\mu = \coshp$] Analogous. \qedhere
	\end{description}
\end{proof}

\subsubsection{The type $\Size$}
We can now proceed with the interpration of $\Size$ in ParamDTT. Just like with $\Nat$, the interpretation of $\szero$, $\ssuc$ and $\smax{}{}$ is entirely straightforward. For t-Size-fill, we have
\begin{align*}
	&\interp{\inference{
		\Gamma \judtm{P}{\IF} \qquad
		\Gamma, \ctxface{P} \judtm{n}{\Size}
	}{\Gamma \judtm{\sfillsys{\sfillsysclause P n}}{\Size}}{t-Size-fill}}
	= \\
	&\inference{
		\inference{
		\inference{
		\inference{
			\interp \Gamma \sez \interp P : \PropD
		}{\shp \interp \Gamma \sez \ElD{\interp P} \prop}{}
		}{\sharp \shp \interp \Gamma \sez \sharp \ElD{\interp P} \prop}{}
		}{\shp \interp \Gamma \sez (\sharp \ElD{\interp P})[\iota] \prop}{}
		\qquad
		\interp \Gamma, \var p : (\sharp \ElD{\interp P})[\sharp \varsigma][\iota] \sez \interp n : \Size
	}{\interp \Gamma \sez \sfillsys{\sfillsysclause{\var p : (\sharp \ElD{\interp P})[\iota]}{\interp n}} : \Size}{}
\end{align*}
Then we can also interpet t=-Size-codisc.

\subsubsection{The inequality type}
The inequality type is interpreted as
\begin{equation}
	\interp{
		\inference{
			\Gamma \judtm{m, n}{\El~\Size}
		}{\Gamma \judtm{m \leq n}{\El~\uni 0}}{t-$\leq$}
	} =
	\inference{
	\inference{
	\inference{
	\inference{
		\interp \Gamma \sez \interp m, \interp n : \Size
	}{\shp \interp \Gamma \sez \interp m[\varsigma]\inv, \interp n[\varsigma]\inv : \Size}{}
	}{\sharp \shp \interp \Gamma \sez \ftrtm{\sharp}{(\interp m[\varsigma]\inv)}, \ftrtm{\sharp}{(\interp n[\varsigma]\inv)} : \sharp \Size}{}
	}{\sharp \shp \interp \Gamma \sez \ftrtm{\sharp}{(\interp m[\varsigma]\inv)} \leq \ftrtm{\sharp}{(\interp n[\varsigma]\inv)} \prop}{}
	}{\interp \Gamma \sez \tycodeDD{\ftrtm{\sharp}{(\interp m[\varsigma]\inv)} \leq \ftrtm{\sharp}{(\interp n[\varsigma]\inv)}} : \uniDD_0}{}
\end{equation}
We have $\interp{\El~ m \leq n} = \ftrtm{\sharp}{\interp m} \leq \ftrtm{\sharp}{\interp n}$, and $\interp{\El~m \leq n}[\iota] = \forsub \sharp \interp m \leq \forsub \sharp \interp n$.

As an example of how we interpret simple inequality axioms, we take the following:
\begin{equation}
	\interp{
		\inference{
			\sharp \setminus \Gamma \judtm{n}{\El~\Size}
		}{\Gamma \judtm{\name{zero}_\leq~ n}{\El~0 \leq n}}{t-$\leq$-zero}
	} =
	\inference{
		\interp \Gamma \sez \forsub \sharp \interp n : \sharp \Size
	}{\interp \Gamma \sez \star : \szero \leq \forsub \sharp \interp n}{}.
\end{equation}

The filling rule is a bit more complicated. We need to prove
\begin{equation}
	\inference{
		\sharp \setminus \Gamma \judtm P \IF \qquad
		\sharp \setminus \Gamma, \ctxface P \judtm{m, n}{\Size} \qquad
		\Gamma, \ctxface P \judtm{e}{m \leq n}
	}{\Gamma \judtm{\leqfillsys{\leqfillsysclause P e}}{\sfillsys{\sfillsysclause P m} \leq \sfillsys{\sfillsysclause P n}}}{t-$\leq$-fill}.
\end{equation}
First, we unpack the proposition (see the section on face predicates):
\begin{equation}
	\inference{
		\interp{\sharp \setminus \Gamma \judtm P \IF} = \paren{\interp{\sharp \setminus \Gamma} \sez \interp P : \PropD}
	}{\interp{\Gamma} \sez (\sharp \ElD \interp P)[\sharp \varsigma][\iota] \prop}{}.
\end{equation}
We have $\interp \Gamma, \var p : (\sharp \ElD \interp P)[\sharp \varsigma][\iota] \sez \interp e : \forsub \sharp m \leq \forsub \sharp n$, and we need to prove
\begin{equation}
	\interp \Gamma \sez \ldots : \forsub \sharp \paren{ \sfillsys{\sfillsysclause{\var p : (\sharp \ElD{\interp P})[\iota]}{\interp m}} } \leq \forsub \sharp \paren{ \sfillsys{\sfillsysclause{\var p : (\sharp \ElD{\interp P})[\iota]}{\interp n}} }.
\end{equation}
Now, precisely in those cases where the $\sfill$s evaluate to $\interp m$ and $\interp n$ respectively, we have evidence that $\forsub \sharp \interp m \leq \forsub \sharp \interp n$. This allows us to construct the conclusion.

\subsubsection{The $\sfix$ rule}
The fix rule
\begin{equation*}
	\inference{
		\Gamma, \ctxvar \nu n {\El~\Size} \judty{\El~A} \\
		\Gamma \judtm{f}{\El~\prodvar \nu n \Size.(\prodvar \nu m \Size.(\ssuc\,m \leq n) \to A[\varclr m / \varclr n]) \to A}
	}{\Gamma \judtm{\sfix^\nu\,f}{\El~\prodvar \nu n \Size.A}}{t-fix}
\end{equation*}
is interpreted as
\begin{equation}
	\inference{
		(\nu, \beta : \nu \to \sharp) \in \accol{(\coshp, \iota \vartheta), (\Id, \iota), (\sharp, \id)} \\
		\interp \Gamma, \ftrvar \nu n : \nu \Size \sez \interp{\El~A}[\iota] \dtype \\
		\interp \Gamma \sez \interp f : \Pi(\ftrvar \nu n : \nu \Size).(\Pi(\ftrvar \nu m : \nu \Size).(\beta(\ftrvar \nu m) \leq \beta(\ftrvar \nu n)) \to \interp{\El~A}[\iota][\ftrvar \nu m/\ftrvar \nu n]) \to \interp{\El~A}[\iota]
	}{\interp \Gamma \sez \fix^\nu \interp f : \Pi(\ftrvar \nu n : \nu \Size).\interp{\El~A}[\iota]}{}.
\end{equation}
Since the model supports the definitional version of the equality axiom for $\fix$, the axiom itself can be interpreted as an instance of reflexivity.

\bibliographystyle{alphaurl}
\bibliography{../paper/paramdtt-refs.bib}

\end{document}